\documentclass[english,12pt]{article}
\usepackage[T1]{fontenc}
\usepackage[latin9]{inputenc}
\usepackage{geometry}
\usepackage{amsthm}
\usepackage{amssymb}
\usepackage{amsmath}
\usepackage{setspace}
\usepackage{enumerate}
\usepackage{color}
\usepackage{xcolor}
\usepackage[colorlinks = true,
linkcolor = blue,
urlcolor  = blue,
citecolor = blue, 
anchorcolor = blue]{hyperref}
\usepackage{makecell}
\usepackage{graphicx} 
\usepackage[bottom]{footmisc}
\graphicspath{{Figures/}} 
\usepackage{hyperref}
\usepackage{subcaption}
\usepackage{float}
\usepackage{array}
\usepackage{tabularx}
\usepackage{enumitem}
\usepackage[bottom]{footmisc}
\usepackage{algorithm}
\usepackage{algcompatible}
\usepackage{lipsum}
\usepackage{chngcntr}

\algnewcommand\INPUT{\item[\textbf{Input:}]}%
\algnewcommand\OUTPUT{\item[\textbf{Output:}]}%

\newcolumntype{+}{>{\global\let\currentrowstyle\relax}}
\newcolumntype{^}{>{\currentrowstyle}}

\newcolumntype{M}[1]{>{\centering\arraybackslash}m{#1}}
\newcolumntype{N}{@{}m{0pt}@{}}

\usepackage[english]{babel}
\usepackage{natbib}
\setcitestyle{authoryear, open={(}, close={)}}
\theoremstyle{definition}
\theoremstyle{plain}

\theoremstyle{plain}

\newtheorem*{definition*}{Definition}

\makeatother

\theoremstyle{plain}

\newtheorem{theorem}{Theorem}
\newtheorem*{theorem*}{Theorem}

\newtheorem{definition}{Definition}
\newtheorem{example}{Example}

\newtheorem{lemma}{Lemma}

\newcounter{fig}
\newcounter{tab}
\newcounter{tabApp}

\makeatother

\usepackage{babel}
\providecommand{\theoremname}{Theorem}

\usepackage{clipboard}
\newclipboard{responseclipboard}

\usepackage{xr}

\begin{document}
	
\title{Undergraduate Course Allocation\\ through Competitive Markets\thanks{\protect \emph{Kornbluth:} Harvard University, Department of Economics; \href{mailto:dkornbluth@g.harvard.edu}{dkornbluth@g.harvard.edu}; \emph{Kushnir (corresponding author):} Moroccan Center for Game Theory, University Mohammed VI Polytechnic, Rabat, Morocco; \href{mailto:alexey.kushnir@gmail.com}{alexey.kushnir@gmail.com}. Parts of this project were completed while the authors were at Carnegie Mellon University. We are especially thankful to Dennis Epple and Scott Kominers for their unwavering support of this project. We are indebted to the individuals who provided us with course enrollment data and advised us on undergraduate course allocation processes. They have requested to remain anonymous. We owe special thanks to Maryam Saeedi, who generously let us use her server for our computations. We are also grateful to Martin Bichler, Eric Budish, Rachel Childers, Karam Kang, Andrew Komo, Rebecca Lessem, Eric Maskin, Anh Nguyen, Marek Pycia, Tomasz Strzalecki, Davide Viviano, and participants at various seminars and conferences for their numerous suggestions. We also acknowledge Claude, ChatGPT, and Refine.ink for their assistance in polishing the paper and the code. The editor and three anonymous referees provided us with invaluable feedback that significantly improved the paper.}}

\author{Daniel Kornbluth \and Alexey Kushnir}

\date{September 30, 2026}
\maketitle

\begin{spacing}{1.5}
\begin{abstract}
Prevailing methods of course allocation at undergraduate institutions involve reserving seats to give priority to designated groups of students. We present a competitive equilibrium-based mechanism that assigns course seats using student preferences and course priorities. This mechanism satisfies approximate notions of stability, efficiency, envy-freeness, and strategy-proofness. We evaluate its performance relative to a mechanism widely used in practice using preferences estimated from university data. Our empirical findings demonstrate an improvement in student satisfaction and allocation fairness. The number of students who envy another student of weakly lower priority declines by 8 percentage points, or roughly 500 students.

\vspace{4mm}
		
\end{abstract}
\end{spacing}

\thispagestyle{empty}
\newpage

\label{sec:intro}

\par  Every academic term, over 6,500 undergraduate institutions across the United States assign course schedules to a total of nearly twenty million students.\footnote{Based on \href{https://nces.ed.gov/programs/digest/d22/tables/dt22_303.10.asp}{National Center for Education Statistics, 2017-2018}.} Based on course timing and prerequisites, the set of schedules a student can be assigned in a given term is limited. Moreover, student preferences over possible schedules are heterogeneous. In the face of room size and teaching constraints, universities use factors such as seniority and student department to complete the challenging task of allocating seats in overdemanded courses.

\par Course allocation is a well-studied problem in the context of graduate business schools. Many institutions employ course bidding systems that require strategic decision-making from students, often leading to inefficient allocations \citep[][]{sonmez2010course}. Recently, several business schools replaced course bidding systems with the \textit{approximate competitive equilibrium from equal incomes} (A-CEEI) mechanism \citep{Budish2011,budish2023practical}, which seeks to alleviate these issues. This mechanism elicits student preferences over course schedules, assigns students almost-equal budgets of fake money, and internally determines an approximate competitive equilibrium, in which each student receives her most-preferred affordable course schedule subject to course prices. The competitive equilibrium is approximate due to a small \textit{market-clearing error}, measured by the assignment of a course's seats over and under its capacity. This mechanism treats all business school students equally, as students are randomly assigned budgets and face the same course prices.

\par At many undergraduate institutions, students are not treated equally, making course allocation a two-sided matching problem. Each course divides students into levels of priority based on student characteristics like year of study and department. In commonly-used course allocation mechanisms, students select courses in order of seniority (with ties broken randomly) and priorities based on departments are enforced by reserving seats in each course. Top US universities that closely follow this process include Princeton, Johns Hopkins, Duke, Vanderbilt, Washington University in St. Louis, Columbia, Notre Dame, and Carnegie Mellon.\footnote{In general, there is some heterogeneity in how universities elect to assign seats to students. Some universities use priorities (e.g., Dartmouth) or reserves (e.g., Vanderbilt) to assign seats, whereas others delegate the decision to departments (e.g., Cornell). Some use two phases (e.g., UCLA), only allowing students to enroll in a limited number of courses in the first phase. Seniority can be based on academic year (e.g., Carnegie Mellon), number of earned credits (e.g., University of North Carolina), or time to graduation (e.g., Washington University in St. Louis). At some universities, few classes are over-enrolled (e.g., Caltech), whereas at others, the process is rather congested (e.g., UC Berkeley). In most cases, courses have a small number of priority levels, rather than one \citep[][]{Budish2011} or complete strict priorities \cite[][]{roth1984stability}.} 

\par In this paper, we approach the task of allocating course schedules to students in the context of a many-to-many two-sided matching problem with heterogeneous student preferences and weak course priorities. We present a novel \textit{deterministic} allocation mechanism, the \textit{Pseudo-Market with Priorities} (PMP) mechanism, that uses fake money and competitive equilibrium to allocate course seats to students without transfers. This mechanism elicits student preferences over course schedules, assigns almost-equal budgets to students, and finds an approximate competitive equilibrium allocation under \textit{priority-specific prices}. These prices respect course priorities by setting a \textit{cutoff} priority level for each course, where higher priority students obtain seats in the course for free, students at the cutoff level of priority face some price, and lower priority students cannot afford a seat.

\par The A-CEEI mechanism is the special case of our mechanism in which no course prioritizes any student over any others. Our first result shows that, with arbitrary course priorities, the PMP mechanism \textit{maintains the same, small worst-case bound on the market-clearing error} found in \cite{Budish2011} for the A-CEEI mechanism. This bound is increasing in the dimension of the space of prices, which equals the number of courses in settings without priorities. Our result holds despite the expansion from one price per course to priority-specific prices in each course. The key insight of the proof is the introduction of a parameterization within the priority-specific price space. The parameterization simplifies the space of prices to a single value per course, enabling a simultaneous search for both the cutoff priority level and the corresponding price. This strategy simplifies the complexity of the problem and allows for the application of some previously established results.

The PMP mechanism has several advantageous theoretical properties. Adjusting course capacities to accommodate the market-clearing error, the PMP mechanism's outcomes ensure that no student can be assigned a preferred schedule without violating course priorities or capacity constraints (\textit{approximate stability}). It is also not possible to reassign course seats among a group of students, benefitting some members and hurting none, while ensuring that the outcome respects course priorities to the same degree (\textit{approximate priority-constrained efficiency}). Among similarly prioritized students, its outcome limits envy; if a student prefers the schedule assigned to a student with the same or a lower priority level in each course, by removing one course from the preferred schedule, the student now prefers her own (\textit{schedule envy bounded by a single course}). Finally, in large populations, the PMP mechanism is resistant to strategic manipulations by students (\textit{strategy-proof in the large}).

\par To better guide practitioners in the course allocation process, we compare the performance of the PMP mechanism with several alternative mechanisms. We use university data on the final allocation of $6\,023$ students to $756$ courses from the Spring 2019 term at a private institution in the mid-Atlantic region. In total, $23\,369$ course seats are occupied, with $11.2\%$ of courses enrolled at or above their reported capacities. 

\par As a benchmark, we consider a \textit{Random Serial Dictatorship with course reserves} (RSD) mechanism, which closely resembles what is used in practice. Students select courses in order of seniority (with ties broken randomly) and priorities based on student department and year of study are enforced by reserving seats in each course.\footnote{The course allocation process at the university in question is more complex.  At the beginning of the process, the registrar sets a large ``base'' number of reserved seats. After students complete the initial registration stage, departments reduce these reserves and admit students from waiting lists in order to fill available seats. We do not model the post-registration adjustment period. Instead, we present two sets of results: one based on the \textit{base} course reserves and another based on the course reserves that the registrar should set at the start of the process if the reserves could not be subsequently adjusted. We refer to the latter as \textit{optimal course reserves}. These two sets of results provide reasonable lower and upper bounds on the performance of the mechanism used in practice.} 
We also consider the Deferred Acceptance mechanism with a single and with multiple tie-breakings (denoted DA-STB and DA-MTB, respectively). These mechanisms allocate courses to students by extending the Gale-Shapley algorithm \citep[][]{galeshapley1962} to many-to-many matching, where priority ties are resolved by either a single tie-breaking rule or course-specific tie-breakings. The Deferred Acceptance mechanism has been successful in many practical applications, including matching applicants to residency programs, students to schools, and workers to job positions \cite[see][]{roth2018marketplaces}. We use student and course characteristics to estimate a model of student utilities and evaluate the performance of the PMP, RSD, DA-STB, and DA-MTB mechanisms with a variety of measures of student satisfaction and allocation fairness. We analyze six additional mechanisms in Appendix \ref{sec:appendix-additional-simulations-results}. 

Our results demonstrate that the PMP mechanism delivers a substantially fairer outcome than the benchmark RSD mechanism and the two DA mechanisms. Relative to the RSD benchmark, considerably fewer students - around $8\%$ of the total population, or almost $500$ students - prefer the course schedule assigned to a student of weakly lower priority in each course. 
Across all years of study, the PMP mechanism produces a more even distribution of student utilities. While the PMP, DA-STB, and DA-MTB mechanisms respect the course priority structure, about $17\%$ of students would benefit from taking a seat that is assigned to a student with strictly lower priority in the RSD mechanism.

These fairness improvements come without sacrificing student satisfaction. Averaged across runs, only around 48\% of students have different schedules across the four considered mechanisms. Among these students,  38\% strictly prefer the PMP mechanism to the RSD mechanism and 27\% strictly prefer the RSD mechanism. Additionally, mean utility increases for students of all years of study, with students in later years of study benefitting most. In contrast, the two DA mechanisms result in approximately the same number of students who strictly prefer the DA mechanism and who strictly prefer RSD, with more indifferences between DA-STB and RSD and less between DA-MTB and RSD.

{\bf Literature review.} This paper contributes to the literature analyzing the allocation of objects to agents through competitive equilibria. As competitive equilibria from equal incomes are not guaranteed to exist, many papers focus on random allocation mechanisms. \cite{hylland1979efficient} first proposed the use of fake money and competitive equilibrium to randomly allocate objects to agents through what are referred to as pseudo-market mechanisms. \cite{he2018pseudo} incorporate a priority structure into pseudo-markets and analyze random mechanisms with a focus on unit-demand settings such as school choice.\footnote{\cite{he2018pseudo} explain how their results extend in many-to-many settings with additive utility.}  \cite{miralles2021foundations}, \cite{echenique2021constrained} and \cite{nguyen2021stability} consider pseudo-market solutions to random allocation problems with complex constraints. \cite{PyciaPseudoMarkets} provides an excellent survey of the use of pseudo-markets for random allocations in environments without transfers.

In contrast to these papers, we only consider \textit{deterministic} mechanisms and study a many-to-many matching problem. \cite{Budish2011} introduces the idea of competitive markets that may only approximately clear. He proposes a deterministic mechanism, the \textit{approximate competitive equilibrium from equal incomes} mechanism, which is approximately Pareto-efficient, strategy-proof in the large, and bounds student envy by a single course. This mechanism was successfully implemented at Wharton Business School and Columbia Business School \cite[see][]{budish2017course, budishkessler2021}.\footnote{There are few other studies of real-world course allocation mechanisms. \cite{budish2012multi} use field data to study a non-strategy-proof course allocation mechanism at Harvard Business School. \cite{bichlerrandomized} compare the performance of two random mechanisms at the Technical University of Munich. \cite{rusznak2021seat} analyze course allocation at a large university in Hungary.}
A key contribution of the present study is extending Budish's approach to two-sided matching settings with weak course priorities, a central feature of the undergraduate course allocation problem.

Two recent papers also analyze deterministic assignments in many-to-many matching settings. \cite{nguyenvohra2022} establish the existence of a competitive equilibrium when all agent preferences satisfy a novel geometric substitutes property, which is a strict generalization of the gross substitutes property \citep[see][]{kelso1982job}. When there is an upper limit on the number of goods $k$ an agent can acquire, they prove that there always exists an approximate competitive equilibrium with no good over-enrolled by more than $k-1$ units. This ``good-by-good'' bound dominates the aggregate market-clearing bound of \cite{Budish2011} when preferences are close substitutes, whereas the aggregate market-clearing bound dominates when student preferences are close complements. Though our main theoretical results follow the aggregate market-clearing bound approach, we ensure in our simulations that each approximate competitive equilibrium allocation found has no courses that are over-enrolled by more than $k-1$ seats.

Similarly, \cite{lin2022allocation} analyze deterministic allocation mechanisms through approximate competitive equilibria with a good-by-good bound. They consider complex feasibility constraints and prove the existence of an approximate competitive equilibrium that satisfies stability and fairness properties similar to the ones analyzed in this paper. We establish some additional properties of our mechanism, proving that it is strategy-proof in the large and results in an outcome that is approximately priority-constrained efficient. With simulations on randomly generated data, they compare their approach to draft mechanisms in the context of reassigning sporting event season tickets to families. We address the problem of allocating courses to undergraduate students by estimating preferences from real-world university data and comparing our proposed mechanism to a version of the mechanism used at many top universities in the United States. In Appendix \ref{sec:appendix-additional-simulations-results}, we also provide a comparison of our mechanism with the deferred acceptance mechanism with single and multiple tie-breakings \cite[][]{galeshapley1962}, the deferred acceptance mechanism with reserves \cite[][]{hafalir2013effective, westkamp2013analysis}, some alternative pseudo-market mechanisms \cite[][]{Budish2011}, and a variant of the probabilistic serial mechanism \cite[][]{bogomolnaia2001new}. Many of these mechanisms have been studied in many-to-one matching settings, but their properties in many-to-many matching with indifferences are not well-understood. 

\par The remainder of the paper is organized as follows. We model undergraduate course allocation in Section \ref{sec:Environment} and investigate the properties of the PMP mechanism in Section \ref{sec:pseudo-market}. In Section \ref{sec:simulations}, we present simulation results using university data. We offer concluding remarks in Section \ref{sec:discussion}. Proofs are in Appendices \ref{sec:appendix-proofs} and \ref{sec:theorem-1-omitted-details}. Details on the student utility model calibration and additional simulation results are in Appendices \ref{sec:appendix-student-utility-estimation},  \ref{sec:appendix-additional-simulations-results}, and \ref{sec:supplementary-materials}.

\section{Environment}
\label{sec:Environment}

\par Course allocation is a many-to-many matching problem described by the tuple $({\cal S}, {\cal C}, Q, V, \mathcal{R})$.
\begin{itemize}[itemsep=-0.5mm]
	\item ${\cal S}=\left\{ 1,...,N\right\} $ is a set of students; in reference to students, we use she/her pronouns. 
	\item ${\cal C}=\left\{ 1,...,M\right\} $ is a set of courses. 
	\item $Q=(q_1,...,q_M)$ is a vector of course capacities; each course $c \in {\cal C}$ has $q_c$ seats.
			
	\item $V=(\succsim_1,...,\succsim_N)$ is a vector of student preferences over \textit{course schedules}: subsets of ${\cal C}$. Students typically only consider a subset of all possible course schedules due to factors such as course meeting times and prerequisites. Additionally, each student can take at most $k$ total course seats. We assume that these restrictions are incorporated into student preferences and that student preferences are \textit{strict}. We also assume, for simplicity of exposition, that $1\leq k\leq M/2$. These restrictions still permit students to view courses as substitutes or complements.
			
	\item $\mathcal{R}=\{r_{s,c}\}_{s\in {\cal S},c\in {\cal C}}$ is a \textit{course priority structure}. $r_{s,c}\in\{1,...,R\}$ specifies the level of priority of student $s$ in course $c$, with $R \geq 1$ and a larger number meaning a higher level of priority. Multiple students can share the same level of priority in a course.
\end{itemize}

\par Our setting connects two well-studied problems in market design in a natural way. If $r_{s, c} = 1$ for all students $s \in {\cal S}$ and courses $c \in {\cal C}$, this is the combinatorial assignment problem of \cite{Budish2011}. If $\{r_{s, c}\}_{s \in {\cal S}}$ forms a strict ordering over all students for each course $c$, this is a two-sided matching problem in which students have strict preferences over subsets of courses and courses strictly rank individual students.

\par Universities often categorize students into four years of study based on semesters or credits completed. In addition, universities place students into departments, each of which offers a set of courses. These practices illustrate two examples of course priorities: \textit{year-specific} and \textit{department-specific} priorities. In a year-specific priority structure, $R = 4$, with $r_{s, c}$ equal to the year of study of student $s$ in every course $c$. In a department-specific priority structure, $R = 2$, with $r_{s, c} = 2$ if student $s$ and course $c$ are in the same department and $r_{s, c} = 1$ otherwise. Year-specific priorities give an advantage to students closer to graduation in every course, whereas department-specific priorities give an advantage to students that require a given course to graduate over those that do not. We discuss \textit{hybrid priority structures} that combine year-specific priorities and department-specific priorities in Section \ref{sec:simulations}.

\par We consider deterministic allocations of courses to students. An \textit{allocation} $x = (x_s)_{s \in {\cal S}}$ assigns the course schedule $x_s$ to each student $s \in {\cal S}$. For ease of notation, we view $x_s$ as both a subset of courses in ${\cal C}$ and a vector from the set $\{0, 1\}^M$. Allocation $x$ is \emph{feasible} if $\sum_{s \in {\cal S}}x_{s, c} \leq q_c$ for each $c\in {\cal C}$, assigning no course a greater number of students than its capacity. Allocation $x$ is \textit{individually rational} if, for every $s\in {\cal S}$, $x_s=\max_{\succeq_s}\{x'_s\,:\,x'_s\subseteq x_s\}$, ensuring that $s$ does not want to drop any of her assigned courses. A pair $(s,C)$ of student $s\in {\cal S}$ and course schedule $C\subset {\cal C}$ is \textit{a block} of $x$ if a student prefers $C$ to her current course schedule $x_s$ and each course in $C$ not in $x_s$ has an available seat or a lower priority student assigned to a seat. Formally, $(s, C)$ is a block if $C = \max_{\succeq_s}\{x'_s\,:\,x'_s\subseteq x_s\cup C\}$, $C\neq x_s$, and for each $c\in C \setminus x_s$, $\sum_{s\in {\cal S}} x_{s,c}<q_c$ or there is a student $s'\in {\cal S}$ with $c\in x_{s'}$ and $r_{s,c}>r_{s',c}$.

We evaluate allocations based on stability, efficiency, and fairness. An allocation $x$ is \textit{stable} if it is \textit{feasible}, \textit{individually rational}, and admits no \textit{blocks} \citep[see][]{roth_sotomayor_1990,echenique2006}.\footnote{A related concept is justified course envy. Allocation $x$ \textit{prevents justified course envy} if there are no students $s,s'\in {\cal S}$ and course $c\in {\cal C}$ such that $r_{s,c}>r_{s',c},\,c\notin x_s,\,c\in x_{s'}$ and $c\in \max_{\succeq_s}\{x'_s\,:\,x'_s\subseteq x_s\cup c\}$. The absence of justified course envy prevents envy towards students of a lower priority, but does not account for individual rationality and the possibility of envy towards several students assigned to several courses.} 
In the next section, we analyze approximate market equilibria, in which some courses might have seats assigned slightly above capacity or empty demanded seats unfilled. We say such allocations have a \textit{market-clearing error}. To account for allocations with market-clearing errors, we use an approximate version of stability.
\begin{definition}
	\label{definition:approximate-stable}
	An allocation $x$ is \textbf{approximately stable} if it is stable with capacity $q'_c=\sum_{s\in {\cal S}}x_{s,c}$ for each course $c\in {\cal C}$.
\end{definition}
\noindent The relationship between stability and approximate stability is similar to the one between the notions of Pareto efficiency and approximate Pareto efficiency in \cite{Budish2011}. An allocation $y$ \textit{Pareto dominates} an allocation $x$ if there is at least one student who strictly prefers her course schedule in $y$ and all other students weakly prefer their course schedules in $y$. Allocation $x$ is \textit{Pareto efficient} if no allocation $y$ Pareto dominates $x$, and \textit{approximately Pareto efficient} if no $y$ that Pareto dominates $x$ assigns weakly fewer seats in each course. 
Approximate Pareto efficiency is a meaningful notion of efficiency in environments where a market-clearing error is present.

In our setting, a meaningful notion of efficiency should also account for course priorities. Following \cite{schlegel2020welfare}, we extend course priorities over individual students to priorities over subsets of students using first-order stochastic dominance. We say that allocation $y$ dominates allocation $x$ for course $c$, and write $y_{c}\succeq_{c}x_{c}$, if for all levels of priority $r\in{\{1, ..., R\}}$, $\sum_{s \in {\cal S}:r_{s,c}\geq r}y_{s,c}\geq\sum_{s\in {\cal S}:r_{s,c}\geq r}x_{s,c}$. Using this definition, we consider the following notion of priority-constrained efficiency.
\begin{definition}
	\label{definition:priority-constrained-efficiency}
	An allocation $x$ is \textbf{approximately priority-constrained efficient} if, in an environment with capacity $q'_c = \sum_{s \in {\cal S}} x_{s, c}$ for each course $c \in {\cal C}$, for each allocation $y$ that Pareto dominates $x$, there is a course $c' \in {\cal C}$ for which $y_{c'} \nsucceq_{c'} x_{c'}$.
\end{definition}	
\noindent As with approximate stability, we account for a possible market-clearing error in the allocation $x$ by adjusting the course capacities to $q'_c = \sum_{s \in {\cal S}} x_{s, c}$. 
Except for this adjustment, our definition coincides with the priority-constrained efficiency introduced by \cite{schlegel2020welfare}. 

The seminal measure of fairness is envy-freeness, as introduced by \cite{foley1967resource}: an allocation $x$ prevents \emph{schedule envy} if there are no students $s,s' \in {\cal S}$ such that $x_{s^{\prime}} \; \succ_s \; x_s$. However, without using lotteries, we cannot hope to allocate courses in a completely envy-free way among students (e.g., if two students with the same priorities desire a course with a single seat). Extending the notion originally proposed by \cite{Budish2011}, we consider the more permissive concept of \textit{schedule envy bounded by a single course} toward students of weakly lower priority in every course \citep[see also][]{lin2022allocation}.

\begin{definition}
	\label{definition:justified-envy-by-a-single-course}
	An allocation $x$ has \textbf{schedule envy bounded by a single course} toward students of weakly lower priority if, for any $s, s' \in \cal S$ such that $r_{s,c} \geq r_{s',c}$ for all $c \in {\cal C}$, either $x_s \; \succsim_s \; x_{s'}$ or there exists some course $c^*$ such that $x_s \; \succsim_s \; (x_{s'} \setminus \{c^*\})$.
\end{definition}
\noindent This criterion of fairness minimizes each student's envy towards students with weakly lower priority in every course.

\label{definition:SP} Allocations are found through \textit{mechanisms}, which systematically elicit student preferences over course schedules. A mechanism is \textit{strategy-proof} if there is no student, who, by reporting manipulated preferences, receives an allocation she strictly prefers to the course schedule she would get from reporting her true preferences. Evidence from business schools demonstrates that requiring strategic play on behalf of students can complicate efficiency \citep[see][]{budish2012multi, budishkessler2021, sonmez2010course}. We assume that strategic play is only a concern for students, as course priorities are set based on commonly observable factors such as student year of study or department.

Mechanisms that utilize market equilibria fail to be strategy-proof when students can influence course prices. A mechanism is \textit{strategy-proof in the large} if truthful reporting is approximately optimal in an environment with a large number of students. For market equilibria, this amounts to being strategy-proof in a limit market in which each student regards the market prices as exogenous to her reported preferences \citep[see][]{AzevedoBudish2019,Budish2011}.\footnote{\cite{he2018pseudo} and \cite{liu2016ordinal} consider two related notions of asymptotic incentive compatibility and asymptotic strategy-proofness for cardinal and ordinal preferences, respectively. \cite{nguyen2025efficiency} recently introduced a stronger notion of uniform strategy-proofness.} To avoid unnecessary heavy notation early in the paper, we introduce this concept formally in the proof of Theorem \ref{theorem:strategy-proof}. 
 
\section{Pseudo-Markets with Priorities}
\label{sec:pseudo-market}
In this section, we present our novel mechanism, which allocates courses to students by extending the concept of approximate competitive equilibrium from equal incomes from \cite{Budish2011} to settings with course priorities. For this purpose, we assign each student $s \in {\cal S}$ a budget of fake money $b^*_s$, with $1\leq\min_{s}b^*_{s}\leq\max_{s}b_{s}^{*}\leq 1+\beta$ for some small $\beta>0$. The parameter $\beta$ is the maximum allowable budget inequality across students. We also allow for slack in the market-clearing condition, which is bounded by $\alpha\geq 0.$ 
\begin{definition}
	\label{definition:competitive-equilibrium}	
	The allocation $x^{*},$ priority-specific prices $p^{*},$ and budgets $b^{*}$ constitute
	an $(\alpha,\beta)$-\textbf{Pseudo-Market Equilibrium with Priorities} if:
	
	\begin{enumerate}
		\item Each student $s \in {\cal S}$ is assigned her most-preferred affordable course schedule: 
		\begin{equation*}
		x_{s}^{*}=\max_{\left(\succsim_{s}\right)}\left\{ x'_s\subseteq {\cal C}:\sum_{c\in {\cal C}}p_{c,r_{s,c}}^{*}x_{s, c}'\leq b_{s}^{*}\right\} 
		\end{equation*}
		
		\item Each course $c \in {\cal C}$ has a cutoff priority level $r^*_c$ such that:
		\vspace{-1mm}
		\begin{equation}
			p^{*}_{c,r}\in\begin{cases}
				\left\{ 0\right\}  & r>r^*_c\\{}
				[0,\overline{b}) & r=r^{*}_c\\
				[\overline{b},+\infty) & r<r^{*}_c
			\end{cases}\label{eq:price-comp} \tag{$\ast$}
		,\end{equation}
		\vspace{-1mm}
		where $\overline{b} = 1 + \beta + \epsilon$ for some small $\epsilon > 0$.
		
		\item The market-clearing error is less than $\alpha$: $||z^{*}||_{2}\leq\alpha,$ where $z^{*}=(z_{1}^{*},...,z_{M}^{*})$ and: 
		\begin{enumerate}
			\item $z_{c}^{*}=\sum_{s\in {\cal S}}x_{s, c}^{*}-q_{c}$ if $p_{c,1}^{*}>0,$
			\item $z_{c}^{*}=\max(\sum_{s\in {\cal S}}x_{s, c}^{*}-q_{c},0)$ if $p_{c,1}^{*}=0$.
		\end{enumerate}
		\item Student budgets are almost equal: $1\leq\min_{s \in {\cal S}}b_{s}^{*}\leq\max_{s \in {\cal S}}b_{s}^{*} \leq 1+\beta$.
	\end{enumerate}
\end{definition}

Unlike \cite{Budish2011}, the above definition of an approximate competitive equilibrium allows for course prices to depend on priority levels, making the vector of prices $p^*=\left\{ p^*_{c,r}\right\} _{c\in {\cal C},r\in\{1, ..., R\}}\in\mathbb{R}^{M\!R}$. Condition (\ref{eq:price-comp}) ensures that each course has a \textit{cutoff level of priority}, where higher priority students obtain seats in the course for free, students at the cutoff level of priority face some price, and lower priority students cannot afford a seat. In doing so, condition (\ref{eq:price-comp}) guarantees no students are allocated a seat in a course when a higher priority student who would like a seat does not receive one. A similar condition appeared in \citet{he2018pseudo} in the context of random allocation with priorities.

The market-clearing error for a course depends only on whether its price is zero at $r=1$, the lowest priority level. Under-enrollment is only counted as an error if the price at level $r = 1$ is positive. This requirement is a non-trivial extension of \cite{Budish2011}'s definition to settings with course priorities. Alternatively, if a course's market-clearing error depends on whether the price at the cutoff priority level is zero, one can change the cutoff level and artificially lower the error. For example, consider an under-enrolled course with a price of zero at and above the cutoff priority level and a price of $\overline{b}$ below the cutoff. By the alternative definition, the market-clearing error is zero. However, the unfilled course seats should be counted towards the market-clearing error, as they may be eliminated if the price below the cutoff level is decreased.

\cite{Budish2011} establishes the existence of an approximate competitive equilibrium from equal incomes for any budget inequality $\beta > 0$ with an upper bound on the market-clearing error of $\alpha = \sqrt{kM/2}$.\footnote{\cite{Budish2011} shows the existence of an approximate equilibrium from equal incomes with a market-clearing error at most $\sqrt{\min\{2k, M\} M}/2$. As we assume $k \leq M/2$, the upper bound reduces to $\sqrt{kM/2}$.} This bound is proportional to the square root of the dimension of the price space. In the presence of $R$ priority levels for each of the $M$ courses, the dimension of the price space is $M\!R$. This could potentially require a large market-clearing error for a Pseudo-Market Equilibrium with Priorities to exist. The main idea of the proof of Theorem \ref{theorem:existence} is to reduce the effective dimension of the price space from $M\!R$ to $M$, where we can establish the existence of a Pseudo-Market Equilibrium with Priorities. In particular, for any vector $t \in \mathbb{R}^M$, we can define priority-specific prices $p\in \mathbb{R}^{M\!R}$ by setting: 
	
\vspace{-3mm}	
\begin{equation*}
	p_{c,r}(t)=\max(t_{c}-(r-1)\overline{b},0)
\end{equation*}

\vspace{0mm}	
\noindent for each course $c \in {\cal C}$ and priority level $r \in \{1, ..., R\}$. This parameterization guarantees that, for any $t \in [0, R\overline{b})^M$, the corresponding price vector $p$ satisfies condition (\ref{eq:price-comp}). We then prove the existence of an approximate equilibrium in the $M$-dimensional space.

\begin{theorem}[Existence]
	\label{theorem:existence}
	For any $\beta > 0$, there exists a $(\sqrt{kM/2}, \beta)$-Pseudo-Market Equilibrium with Priorities.
\end{theorem}

\noindent Theorem \ref{theorem:existence} establishes the existence of a Pseudo-Market Equilibrium with Priorities with the same upper bound on the market-clearing error as the approximate competitive equilibrium from equal incomes, extending the main result of \cite{Budish2011} to environments with course priorities. The proof of Theorem \ref{theorem:existence} is presented in Appendix \ref{sec:appendix-proofs}.

\subsection{The Pseudo-Market with Priorities Mechanism}
\label{subsec:PMP-mechanism}

Next, using the Pseudo-Market Equilibrium with Priorities, we introduce the Pseudo-Market with Priorities (PMP) mechanism and investigate approximate notions of stability, efficiency, fairness, and strategy-proofness it satisfies.

\begin{definition}
\label{definition:Pseudo-Market-Mechanism} 
The Pseudo-Market with Priorities mechanism with market-clearing error $\alpha$ and budget inequality $\beta$ is defined through the following steps:
\begin{enumerate}
	\item Each student $s$ reports preferences $\succsim'_s$ over course schedules.
	\item Each student $s$ is assigned a random budget $b^*_s$ drawn uniformly from $[1, 1+ \beta]$. 
	\item Compute an $(\alpha,\beta)$-Pseudo-Market Equilibrium with Priorities $(x^*,p^*,b^*)$ using an anonymous method. Allocate courses to students according to $x^*$.
\end{enumerate} 
\end{definition}
\noindent In general, the PMP mechanism depends on the level of allowable market-clearing error $\alpha$ and the level of budget inequality $\beta$. We avoid this dependence in our exposition for results that hold for all non-negative $\alpha\geq 0$ and $\beta\geq 0$.

We first establish that the PMP mechanism results in an \textit{approximately stable} allocation. Definitions \ref{definition:competitive-equilibrium} and \ref{definition:Pseudo-Market-Mechanism} ensure that the allocation is individually rational. With courses regarded as being at their full capacity, condition (\ref{eq:price-comp}) on priority-specific prices prevents blocks.\footnote{An alternative statement holds. Let $(x^*,p^*,b^*)$ be an $(\alpha,\beta)$-Pseudo-Market Equilibrium with Priorities, and, for each $c \in {\cal C}$, let $q'_c = \sum_{s\in {\cal S}} x^*_{s, c}$ if $p^*_{c, 1} > 0$ and $q'_c = \max(\sum_{s\in {\cal S}} x^*_{s, c}, q_c)$ if $p^*_{c, 1} = 0$. Then, allocation $x^*$ is stable in the environment with adjusted capacities $q'_c,$ $c\in {\cal C}$.}

\begin{theorem}[Stability]
	\label{theorem:stability}
	The Pseudo-Market with Priorities mechanism results in an allocation that is approximately stable.
\end{theorem}

\noindent The proof is postponed to Appendix \ref{sec:appendix-proofs}. The main reason that the outcome of the PMP mechanism is not stable is the possibility of over- and under-enrolled courses. Stability requires a feasible allocation and only allows a course to have unfilled seats if no student can benefit from taking one. The PMP mechanism can violate these conditions when the market-clearing error is nonzero. However, Theorem \ref{theorem:existence} suggests instances will be rare. 

The presence of course priorities prevents the PMP mechanism from being Pareto efficient. This is similar to how stable allocations in matching markets can fail to be Pareto efficient \citep[see][]{roth_sotomayor_1990}. We illustrate this in the following example.
\begin{example}
	\label{example:efficiency}
	Consider a setting with two students ${\cal S}=\{1,2\}$ and two courses ${\cal C}=\{A,B\}$, with $q_A = q_B = k = 1$. Student preferences are $\{A\}\succ_1 \{B\} \succ_1 \varnothing$ and $\{B\}\succ_2 \{A\} \succ_2 \varnothing$. Student budgets are $b^*_1=1$ and $b^*_2=1+\beta$ for some $0<\beta<1$. Course priorities are the opposite of student preferences, with $r_{1,B}=r_{2,A}=2$ and $r_{1,A}=r_{2,B}=1$. Budgets $b^*$, cutoff priority levels $r^*_A=r^*_B=2$, priority-specific prices $p^*_A=p^*_B=(2,1)$, and allocation:
	\begin{center}
		\begin{tabular}{ccc} 
			Student & $x^*_{s,A}$ & $x^*_{s,B}$ \\ [0.5ex] 
			\hline
			1 & 0 & 1 \\ 
			2 & 1 & 0
		\end{tabular}
	\end{center}
	constitute a $(0, \beta)$-Pseudo-Market Equilibrium with Priorities. The allocation is not Pareto efficient, as both students would be made better off by exchanging their assigned seats.	
\end{example}
\noindent The next result shows that the outcome of the PMP mechanism satisfies a notion of constrained efficiency. Adjusting course capacities to $q'_c = \sum_{s \in {\cal S}} x^*_{s, c}$, no allocation that Pareto dominates the outcome of the PMP mechanism respects priorities to the same degree. For each course, the distribution of assigned student priority levels cannot first-order stochastically dominate the corresponding distribution for the PMP outcome (see Definition \ref{definition:priority-constrained-efficiency}).\footnote{Similarly to Theorem \ref{theorem:stability}, an alternative statement holds. Let $(x^*,p^*,b^*)$ be an  $(\alpha,\beta)$-Pseudo-Market Equilibrium with Priorities, and, for each $c \in {\cal C}$, let $q'_c = \sum_{s\in {\cal S}} x^*_{s, c}$ if $p^*_{c, 1} > 0$ and $q'_c = \max(\sum_{s\in {\cal S}} x^*_{s, c}, q_c)$ if $p^*_{c, 1} = 0$. Then, allocation $x^*$ is priority-constrained efficient in the environment with adjusted capacities $q'_c,$ $c\in {\cal C}$. In environments without course priorities, this result reduces to Proposition 2 in \cite{Budish2011}.}

\begin{theorem}[Efficiency]
	\label{theorem:Pareto-improvements}
	The Pseudo-Market with Priorities mechanism results in an allocation that is approximately priority-constrained efficient.
\end{theorem}
\noindent The proof of the above result resembles the proof of \cite{schlegel2020welfare} for single-unit settings with random allocations and is postponed to Appendix \ref{sec:appendix-proofs}.\footnote{We assume each agent has strict preferences over course schedules. So, unlike \cite{schlegel2020welfare} and \cite{miralles2021foundations}, we do not require that a student chooses the cheapest course schedule when multiple course schedules are optimal.} 

Though the concepts of approximate stability and priority-constrained efficiency are related, neither implies the other. If we assume in Example \ref{example:efficiency} that both students have the same priority in each course, then there are two (approximately) stable allocations: $x_1 = \{A\}$, $x_2 = \{B\}$ and  $x_1 = \{B\}$, $x_2 = \{A\}$. However, only the first allocation is approximately priority-constrained efficient. Here, approximate priority-constrained efficiency selects among approximately stable allocations, favoring more Pareto optimal outcomes. In turn, suppose that student $1$ is the only student in the environment and that $x_1 = \{A, B\}$. This allocation is approximately priority-constrained efficient, but not approximately stable, as it is not individually rational. 

Theorem \ref{theorem:stability} prevents the possibility that a student can benefit from taking a seat in a course that a student with lower priority obtains. However, priorities in undergraduate course allocation are weak, with hundreds of students often at the same level of priority across many courses. The next result addresses fairness between students whose priorities are weakly ordered in the same way across all courses. 
\begin{theorem}[Fairness]
	\label{theorem:fairness}
	If $\beta\leq \frac{1}{k-1}$, the Pseudo-Market with Priorities mechanism results in an allocation that has schedule envy bounded by a single course toward students of weakly lower levels of priority.
\end{theorem}
\noindent Theorem \ref{theorem:fairness} extends the result established in \cite{Budish2011} for settings without priorities. With priorities based on seniority and student department, the above result implies that, with a sufficiently small budget inequality, the PMP mechanism produces an allocation in which any two students from the same department and the same year of study have schedule envy bounded by a single course toward one another.\footnote{This result can be strengthened. The PMP mechanism results in an allocation $x^*$ such that for any two students $s, s' \in \cal S$ such that $r_{s,c} \geq r_{s',c}$ for each $c \in x^*_{s'}$, envy is bounded by a single course. In particular, rather than all courses, $s$ needs only to be at a weakly higher priority level in the courses being taken by $s'$.}

Our final theorem establishes that students have almost no incentive to manipulate the PMP mechanism in large populations. This is commonly the case at undergraduate institutions, where student bodies often consist of thousands of students.
\begin{theorem}[Strategy-Proofness]
	\label{theorem:strategy-proof}
	The Pseudo-Market with Priorities mechanism is strategy-proof in the large. 
\end{theorem}	 
\noindent We define the concept of strategy-proofness in the large in the proof of Theorem \ref{theorem:strategy-proof} in Appendix \ref{sec:appendix-proofs}. To prove this theorem, we leverage on a result from \cite{AzevedoBudish2019}, which establishes a connection between strategy-proofness in the large and envy-freeness. Students can be partitioned into groups based on their course priorities, where each student in a group is at the same level of priority in every course. The PMP mechanism is a \textit{semi-anonymous} mechanism, as, within groups, each student faces the same course prices and the same random lottery over budgets. As a student cannot prefer the allocation of any lower-budget student in the same group, the PMP mechanism is \textit{envy-free but for tie-breaking}. Our result then follows from \cite{AzevedoBudish2019}, which shows that any semi-anonymous mechanism that is envy-free but for tie-breaking is strategy-proof in the large.\footnote{One can also consider a continuum replication of a course allocation problem with course priorities. By leveraging on the price structure as in Theorem \ref{theorem:existence}, the steps of Theorem 4 of \cite{Budish2011} can be adapted to obtain the result. The alternative proof is available upon request.}

We conclude this section by considering special cases discussed in previous literature. In the case of \textit{no priorities}, our mechanism reduces to the approximate competitive equilibrium from equal incomes mechanism of \cite{Budish2011}. On the other hand, courses may have \textit{strict priorities} over students. In this case, we have a two-sided matching problem. As students can view sets of courses as complements, no stable matching is guaranteed to exist. Still, Theorem \ref{theorem:existence} confirms a Pseudo-Market Equilibrium with Priorities exists with a small error. So, our mechanism does not necessarily coincide with any stable matching mechanism. 

If courses have strict priorities and students have preferences over course schedules that are \textit{responsive} to strict preferences over individual courses, the student-proposing deferred acceptance mechanism results in the student-optimal stable matching \citep[see][]{roth1984stability}, which can be supported with a Pseudo-Market Equilibrium with Priorities where the market clears exactly, the cutoff level of priority is at the lowest-priority student to receive a seat, and the price of each course is zero at the cutoff. The set of equilibria with zero market-clearing error exactly coincides with the set of stable matchings.\footnote{See \cite{miralles2021foundations} for a related result showing that every efficient assignment can be decentralized through prices in random allocation settings.} In large markets, we expect the set of stable matchings to be small \citep[see][]{kojima2009incentives} and, hence, PMP outcomes should not differ much from each other. 

When each student can only enroll in \textit{one course}, our setting reduces to the unit demand variation of the problem, known in the literature as school choice. The standard mechanisms in school choice are ordinal \citep[][]{abdulkadirouglu2003school} or elicit only restricted information about cardinal preferences from agents \citep[][]{abdulkadirouglu2015expanding}. With weak course priorities, the deferred acceptance algorithm may not return a student-optimal stable matching \citep[][]{erdil2008}. At the same time, any stable matching that is priority-constrained efficient, student-optimal or not, can be supported by a Pseudo-Market Equilibrium with Priorities with zero market-clearing error.

Allowing random allocations, \cite{he2018pseudo} show how to extend the pseudo-market approach of \cite{hylland1979efficient} to school choice settings. Their pricing structure ensures that their pseudo-market equilibrium allocation is fair. We adopt a similar pricing structure but analyze deterministic allocations rather than random allocations. While deterministic allocations are preferable from a practical market design point of view, they are associated with additional complications in equilibrium existence.

\section{Simulations}
\label{sec:simulations}

\vspace{-1mm}
\par Next, we analyze the performance of the Pseudo-Market with Priorities mechanism with student preferences estimated using course allocation data from a private institution. We compare the Pseudo-Market with Priorities mechanism with a version of the Random Serial Dictatorship with course reserves, the mechanism used for course allocation at many U.S. universities, and Deferred Acceptance with single and multiple tie-breakings, mechanisms actively used in the allocation of students to public schools.

\subsection{Student and Course Data}

\label{college-department-course-start}

We utilize data from the Spring 2019 semester at a private institution in the mid-Atlantic region \citep{AnonymousUniversity2019}, covering $6\,023$ students across four years of study and $756$ courses across seven colleges (see Table~\ref{tab:students-courses-by-college}). Colleges are the largest constituent units of the university and typically contain several departments. For example, the Department of Economics is housed in the College of Liberal Arts at a typical U.S. university. Colleges differ in the size of their student population and the number and type of courses they offer. Some colleges serve the whole university, whereas others mainly offer courses only for their students.\footnote{We exclude graduate and exchange students. Fourth- and fifth-year students are treated as one group. There are $1565$, $1611$, $1422$, and $1425$ students in the first, second, third, and fourth year of study, respectively.}
\label{subsec:course-student-data}
\begin{table}[t!]
	\begin{center}
	\footnotesize
		\refstepcounter{tab}\label{tab:students-courses-by-college}
		\caption{The number of students and courses in each college.}
		\begin{tabular}{cccccccc}
			\hline
			\hline\\[-2mm]
			\makecell[lc]{College}& A& B & C & D & E & F & G\\[2mm]
			\hline\\[-2mm]
			\makecell[lc]{\# of students}& 853 & 1642 & 259 & 1274 & 745 & 741 & 509 \\[2mm]
			\makecell[lc]{\# of courses}& 180 & 84 & 12 & 269 & 88 & 84 & 39 \\[2mm]
			\hline
			\hline
		\end{tabular}
	\end{center}
	
	\vspace{1.7mm}
	\begin{spacing}{0.8}
		{\footnotesize \textit{Notes:} The table presents the number of students and courses in each of the seven colleges. We use capital letters to denote college names.}
	\end{spacing}
\end{table}

\begin{table}[b!]
	\begin{center}
		\footnotesize
		\refstepcounter{tab}\label{tab:binding-capacity}
		
		\vspace{1mm}
		\caption{The percentage of courses at or above maximum capacity.}
		\begin{tabular}{p{0.35\textwidth}p{0.07\textwidth}p{0.07\textwidth}p{0.07\textwidth}p{0.07\textwidth}p{0.07\textwidth}p{0.07\textwidth}}
			\hline
			\hline\\[-1ex]
			\makecell[lc]{Seats taken above capacity}& $\geq 0$& $\geq 1$ & $\geq 2$ & $\geq 3$ & $\geq 4$ & $\geq 5$\\[2ex]
			\hline\\[-1ex]
			\makecell[lc]{\% of courses}& 11.2\% & 7.3\% & 4.1\% & 3.3\% & 2.5\% & 1.9\% \\[2ex]
			\hline
			\hline
		\end{tabular}
	\end{center}
	
	\vspace{1mm}
	\begin{spacing}{0.8}
	{\footnotesize \textit{Notes:} The table shows the percent of courses that are above course capacity by different number of seats. Maximum capacity is a soft constraint on enrollment for many courses.}
	\end{spacing}
\end{table}

\vspace{1mm}
\noindent {\bf Course enrollments.} The university allows courses to be over-enrolled. Table \ref{tab:binding-capacity} shows that $11.2\%$ of courses meet or exceed their listed capacity, with $7.3\%$ of courses exceeding it. The listed maximum capacity in over-enrolled courses is intended to ensure proper class dynamics and balanced course sections. This fits our treatment of capacity constraints in the Pseudo-Market with Priorities mechanism, which allows over-enrollment through the market-clearing error. We set each course's capacity equal to the maximum of the capacity observed in the data and the true enrollment.

\vspace{3mm}
\noindent {\bf Course seat reservations}. In practice, the course registration process follows a Random Serial Dictatorship with course reserves, where students select courses in order of seniority (with ties broken randomly) and courses make seat reservations for students in certain departments and years of study. Seat reservations allow courses to express priorities for students, ensuring that students in specific groups have access to a fixed number of a course's seats.\footnote{See \cite{CelebiFlynn2024} for a framework analyzing both reservations and priorities.} Table \ref{tab:course-data} presents quantiles of the distributions of course capacities, enrollment, and total number of reserved seats from the processed university data. $33\,071$ course seats were available across courses, $23\,369$ seats were occupied, and $12\,965$ were reserved at the start of course registration.

\begin{table}[t!]
	\begin{center}
		\footnotesize
		\refstepcounter{tab}\label{tab:course-data}
		\caption{Course capacities, enrollment, and reservations.}
		
		\vspace{1mm}
		\begin{tabular}{cccccc}
			\hline
			\hline\\[-1ex]
			\makecell[lc]{Quantile}& 0.10 & 0.25 & 0.50 & 0.75 & 0.90\\[1ex]
			\hline\\[-1ex]
			\makecell[lc]{Maximum Capacity}& 8 & 14 & 24 & 48 & 98 \\[1ex]
			\makecell[lc]{Actual Enrollment}& 3 & 6 & 14 & 34 & 71\\[1ex]
			\makecell[lc]{\# of Reserved Seats}& 0 & 0 & 4 & 19 & 49\\[1ex]
			\hline
			\hline
		\end{tabular}
	\end{center}
	
	\vspace{1mm}
	\begin{spacing}{0.8}
		{\footnotesize \textit{Notes:} This table presents some quantiles of the distributions of course capacities, enrollment, and number of reserved seats from the processed university data.}
	\end{spacing}
\end{table}

Table \ref{tab:course-reservations} provides some examples of course seat reservations. For instance, Microeconomics reserves $25$ seats for Department $1$ and Department $2$ students in Year $1$. In the data, it is possible that these reservations are done independently by the two departments (e.g., among Year $1$ students, $15$ seats are reserved for Department $1$ students and $10$ seats reserved for Department $2$). Whenever this is the case, we combine these reservations into one by uniting the department names and summing the number of reserved seats. After this simplification, each course has at most one reservation for each year of study and one for ``all years of study''.

\newcolumntype{C}[1]{>{\centering\arraybackslash}p{#1}} 
\begin{table}[t!]
	\begin{center}
		
		\footnotesize
		\refstepcounter{tab}\label{tab:course-reservations}
		\caption{Examples of course seat reservations.}
		\begin{tabular}{m{4cm}C{2.5cm}C{2.5cm}C{3.5cm}}
			\hline
			\hline\\[-1.5ex]
			\makecell[lc]{Course}& Departments & Year of Study & Total Reserved Seats\\[1ex]
			\hline\\[-1.5ex]
			\makecell[lc]{Microeconomics}& {Dept 1, Dept 2} & {Year 1} & 25\\[1ex]
			\hline\\[-1.5ex]
			\makecell[lc]{Real Analysis}& {Dept 2} & {Year 3, Year 4} & 50\\[1ex]
			\hline\\[-1.5ex]
			\makecell[lc]{Machine Learning}& {All} & {Year 4} & 10\\[1ex]
			\hline\\[-1.5ex]
			\makecell[lc]{Communication}& {Dept 3} & {All} & 5\\[1ex]
			\hline
			\hline
		\end{tabular}
	\end{center}
	
	\vspace{1mm}
	\begin{spacing}{0.8}
		{\footnotesize \textit{Notes:} Each row of this table corresponds to a seat reservation for a different course. ``All'' denotes that any student, regardless of department/year of study, is eligible for a reserved seat.}
	\end{spacing}
\end{table}

\vspace{0mm}
Using the data on seat reservations, we create a priority structure that reflects the university's intentions. The registration process favors more senior students by giving them earlier time slots to choose courses, resembling \textit{year-specific priorities}. In addition, courses reserve seats for students in specific departments and years of study, resembling \textit{department-specific priorities}.\footnote{Universities that follow these priorities include Princeton, Johns Hopkins, Duke, Vanderbilt, Washington University in St. Louis, Columbia, Notre Dame, and Carnegie Mellon. However, each university operates slightly differently. Within a year of study, many universities (e.g., Duke) randomly stagger students across time slots rather than having all register at once. In addition, some universities (e.g., Vanderbilt) reserve seats for students in specific majors or degree programs rather than departments or years of study.}
We combine these notions into a \textit{hybrid priority structure}, with year-specific priorities taking precedence over department-specific priorities. Example \ref{example:priority-structures} provides an illustration for Real Analysis in Table \ref{tab:course-reservations}. 
As no seats are reserved for 1st- and 2nd-year students, there are no students at the fourth and second levels of priority.

\vspace{-1mm}	
\begin{example}
\label{example:priority-structures}
	The course seat reservations for Real Analysis in Table \ref{tab:course-reservations} result in the following hybrid priority structure for the course:

\vspace{-2mm}
\begin{itemize}[itemsep=-2mm]
	\item $r_{s, c} = 8$: 4th-year students in Dept 2
	\item $r_{s, c} = 7$: 4th-year students in all departments except Dept 2
	\item $r_{s, c} = 6$: 3rd-year students in Dept 2
	\item $r_{s, c} = 5$: 3rd-year students in all departments except Dept 2
	\item $r_{s, c} = 4$: $\varnothing$
	\item $r_{s, c} = 3$: 2nd-year students in all departments
	\item $r_{s, c} = 2$: $\varnothing$
	\item $r_{s, c} = 1$: 1st-year students in all departments
\end{itemize}

\vspace{-6mm}
\end{example}

\smallskip

\noindent{\bf Student utilities}. The data does not contain information on student preferences over course schedules. To recover this information, we use a parametric form for student utilities and the method of simulated moments to calibrate its parameters \citep[see, e.g.,][]{agarwal2015empirical}. We assume that each student can take up to $k=5$ courses and that preferences are additive over courses.\footnote{Less than 3\% of the student body is enrolled in more than five courses in the data.} Student $s$'s utility from taking course $c$, $u_{sc}$, takes the form

\vspace{-5mm}
\begin{equation}
	\label{eq:student-utility}
	u_{sc} = \textcolor{black}{Y'_s \theta} + Y'_s\Gamma H_c + \textcolor{black}{z_c} + \varepsilon_{sc}.
\end{equation}
The parametric form in equation \eqref{eq:student-utility} contains three non-random components. The first, \textcolor{black}{$Y'_s\theta$}, is a \textit{horizontal component} determined by the student's characteristics. $Y_s$ is a binary vector that indicates student $s$'s college $a\in {\cal A} = \{A, ..., G\}$ and year of study $y\in {\cal Y}= \{1, ..., 4\}$. So, $Y_s'\theta=\theta_{ay}$, where vector $\theta\in \mathbb{R}^{|{\cal A}||{\cal Y}|}$ has $|{\cal A}||{\cal Y}|=28$ elements to calibrate. The larger $\theta_{ay}$ is, the more courses student $s$ wants to enroll in. This accounts for variability in the number of courses taken by students from different colleges and years of study (see Table \ref{tab:students-n-courses-taken} in Appendix \ref{sec:appendix-student-utility-estimation}). 

The second component, $Y'_s\Gamma H_c$, is an \textit{interaction term} that determines the heterogeneity of students' preferences over courses. $H_c$ is a binary vector of observed course characteristics indicating course $c$'s college. We restrict matrix $\Gamma$ to have only $49$ parameters $\gamma_{aa'}$, where $a$ is student $s$'s college and $a'$ is course $c$'s college. This accounts for patterns of course enrollment across colleges (see Table \ref{tab:percent-courses-outside-college} in Appendix \ref{sec:appendix-student-utility-estimation}). 

The third component, $\textcolor{black}{z_c}$, is a \textit{vertical component} that measures the popularity or quality of course $c$. The vector $z$ has length $756$, with one entry for each course. The last term $\varepsilon_{sc}$ is a random utility component that is assumed to have a normal distribution $\varepsilon_{sc}\sim N(0,\sigma^2)$ with unobserved parameter $\sigma$. As normalizations, the value of the outside option from not taking a course is $u_{s0}=0$ and the standard deviation of the noise parameter is $\sigma=1$.

We cannot jointly identify the parameters $\theta$, $\Gamma$, and $z$ without additional normalization as they enter in equation \eqref{eq:student-utility} in an additive way. Hence, we normalize the average horizontal component across students in the same college to $\overline{\theta}_{a}=0$ and set the interaction terms $\gamma_{aa}$ to $0$ for each $a\in {\cal A}$. These normalizations allow for an interpretation of $z_c$ as the average student utility from taking course $c$ among students in the same college as course $c$. The value of the parameter $\gamma_{aa'},a\neq a'$ is the average increase in student utility from taking course $c'$ in college $a'$ over taking course $c$ in student's own college $a$ when both courses have the same values of vertical components (i.e., $z_{c}=z_{c'}$). Overall, we have $819$ independent parameters $\theta$, $\gamma$, and $z$ to calibrate.

Last, we assume that each student's choice set is limited to $80$ courses. This assumption is motivated by the observation that any given student typically only considers taking a subset of all offered courses in a given term and is made for computational purposes. We randomly draw these $80$ courses among those taken by at least one student from the same college and year. The probability that a student in college $a$ and year $y$ draws course $c$ is proportional to the share of seats taken in course $c$ by students in college $a$ and year $y$, among all seats taken by such students.\footnote{We expand some students' choice sets to make sure enough reserve-eligible students are interested in each course. This ensures that our calibration process runs smoothly. We elaborate on this in Appendix \ref{sec:appendix-student-utility-estimation}.} Not all $80$ utility entries are necessarily positive.

We use the method of simulated moments to identify the model's parameters. A description of the calibration procedure, interpretations of the calibrated parameter values, and model fit results are left to Appendix \ref{sec:appendix-student-utility-estimation}. Here, we discuss two obstacles in the calibration process. First, in practice, course reservations do not remain constant during the registration process. After the first week of registration, demand for courses becomes more or less clear. To ensure complete course enrollment, departments start relaxing course reservations by admitting students from waiting lists. While there are best practices for handling waiting lists, this phase of enrollment is done solely at each department's discretion. As a result, a course's final enrollment can violate its seat reservations set at the beginning of the process. We use \textit{Hall's Marriage Theorem} to identify course reserve violations and minimally adjust the number of reserved seats. The seat reservations are lowered in 150 of the 756 courses to ensure the number of reserved seats is consistent with the allocation observed in the data.

The second obstacle concerned calibrating the vertical components $z_c$ for courses at maximum capacity. For such courses, no value of $z_c$ can lead to a higher enrollment than observed in the data. We identified the vertical components for courses at maximum capacity using data on how quickly seats in those courses were occupied during the registration process.

\subsection{Course Allocation Mechanisms}
\label{subsec:mechanisms}

We analyze the four mechanisms described below. We assume students report their utilities truthfully, as each mechanism is strategy-proof or strategy-proof in the large.

\vspace{2mm}
\noindent {\bf The Pseudo-Market with Priorities (PMP) mechanism.} Students are randomly assigned evenly spaced budgets between $1$ and $1+\beta$ (with $\beta=\frac{1}{k-1} = 0.25$). Following the hybrid priority structure, students are placed at the level of priority corresponding to their year of study and department in each course. As in the proof of Theorem \ref{theorem:existence}, we parameterize prices with $t\in [0, R\overline{b}]^M$, where $R = 8$ and $\overline{b}=1.251 > 1+\beta$. For each $t \in [0, R\overline{b}]^M$, course $c \in {\cal C}$, and priority level $r\in {\{1, ..., R\}}$, we define priority-specific prices as $p_{c,r}(t)=\max(t_{c}-(r-1)\overline{b},0)$. Then, we look for an equilibrium in the lower-dimensional $t$-space. 

The search process has two phases. Phase I starts with an educated guess for course prices based on students' utility-maximizing course schedules. The algorithm searches in the lower-dimensional space of $t$ parameters for a market-clearing error smaller than the theoretical bound $\alpha = \sqrt{kM/2}\approx 43.5$, iteratively computing each student's utility-maximizing schedule and adjusting prices proportional to the number of over and under-enrolled seats. We move to Phase II if the market-clearing error fails to improve within a certain number of iterations or if the improvement is smaller than $1\%$ after the theoretical bound has been reached. Phase II takes the Phase I prices and works to reduce excess demand, ensuring that no course is over-enrolled by more than $k-1=4$ seats and the distribution of over-enrolled courses is smaller than the one observed in the data (see Table \ref{tab:binding-capacity}).\footnote{Doing so guarantees the allocation has a market-clearing error below the bound in Theorem \ref{theorem:existence} and meets the good-by-good bound of \cite{lin2022allocation}. A proof on the existence of an equilibrium with the good-by-good bound in our setting is available by request.} We return back to Phase I if an allocation with a market-clearing error greater than $\alpha$ is obtained in Phase II.

\vspace{2mm}
\noindent {\bf The Deferred Acceptance mechanism with single and multiple tie-breakings.} These mechanisms use the student-proposing deferred acceptance algorithm to allocate course seats to students. Each course ranks students by priority, using the same hybrid priority structure as in the PMP mechanism. The two variants differ in how they break ties between students at the same level of priority. The first uses a single tie-breaking order across all courses. To break ties, we use the budgets assigned in the PMP mechanism, placing students with larger budgets earlier. Without course seat reservations, this mechanism is equivalent to a random serial dictatorship (with students ordered by seniority and ties broken randomly) and can lead to significant envy between students. The second uses multiple, course-specific tie-breakings. We refer to these mechanisms as DA-STB and DA-MTB, respectively. Both have received a thorough analysis in the many-to-one matching literature, but their many-to-many counterparts are underexplored \citep[see, e.g.,][]{erdil2008,erdil2017two}.

\vspace{2mm}
\noindent{\bf The Random Serial Dictatorship (RSD) with optimal course reserves.}  This mechanism orders students by seniority, breaking ties using the same order as in PMP and DA-STB, and assigns each student the (up to) $k$ courses available that offer her the highest utility. Course seat reservations are treated as described in Section \ref{subsec:course-student-data}. If a student is eligible for one of a course's seat reservations, she can enroll in a reserved or regular seat. Otherwise, the student can only be enrolled in a regular seat. Reserved seats must be occupied before a reserve-eligible student occupies a regular seat.\footnote{When a student is eligible for a course's reservation for her year of study and its reservation for all years of study, the student occupies a seat corresponding to the former reservation.}

In practice, the university's registration system has two phases. The first is a centralized admission process, which uses the RSD mechanism. The second is a decentralized stage, where departments relax reserves to fill unoccupied seats with waitlisted students. This second phase is not well-structured. Importantly, the university administration sets ``base'' course reserves more generously in the first phase than it would if the second phase did not exist. This guarantees that prioritized students have access to course seats but leaves many unfilled reserved seats to be assigned in the second phase. As a result, using the base course reserves observed in the data leads to significant under-performance of the RSD mechanism. Instead, we estimate the \textit{optimal reserved seats} by generating $100$ environments with random student choice sets and utilities, setting a random tie-breaking order, and running the DA-STB mechanism. We take the average number of reserve-eligible students assigned to each course across all environments (rounded to the nearest integer) as the optimal number of reserved seats. We call the result of the RSD mechanism operated with the base course reserves and with the estimated reserved seats \textit{RSD with base course reserves} and \textit{RSD with optimal course reserves}, or simply RSD-Base and RSD-Optimal. These provide reasonable lower and upper bounds for the performance of the mechanism used in practice.\footnote{The term optimal is not used here in the sense of maximizing some objective, but rather to denote a practical procedure to calculate course reserves. The performance of RSD-Optimal and RSD-Base cannot be directly compared to the performance of two-stage mechanism used in practice.}

\vspace{0mm}

In the next section, we present simulation results comparing the PMP, DA-STB, DA-MTB, and RSD-Optimal mechanisms across five measures. These include \textit{two measures of student satisfaction}: (i) each student's preferred mechanism and (ii) mean student utility; and \textit{three measures of allocation fairness}: (iii) the standard deviation of students' utilities, (iv) the amount of envy (toward students with weakly lower priority), and (v) priority violations (students desiring a course being taken by a student of strictly lower priority). From theory alone, we know that no student can experience envy greater than a single course in the PMP mechanism (Theorem \ref{theorem:fairness}), and that the PMP, DA-STB, and DA-MTB mechanisms all prevent priority violations. Otherwise, it is not clear how the four mechanisms will perform relative to one another. Our simulations provide comparisons for the remaining measures and allow us to see if the envy and priority violations have important magnitudes.

\subsection{Simulation results}
\label{subsec:simulation-results}

We evaluate the performance of the Pseudo-Market with Priorities (PMP), Random Serial Dictatorship with optimal course reserves (RSD-Optimal), and Deferred Acceptance with both single (DA-STB) and multiple tie-breakings (DA-MTB) mechanisms over 100 simulation runs.\footnote{Simulations were run in Mathematica $14.1$ on a server with $36$ CPU cores ($2.3$ GHz Intel Xeon E5-2697 v4), $256$ GB of RAM, and the Windows Server 2019 Standard operating system. It takes about thirty-five minutes to calculate the outcomes of all four mechanisms for each run.} Each run corresponds to a random draw of student choice sets, utilities, a single tie-breaking order (common to PMP, DA-STB, and RSD-Optimal), and multiple course-specific tie-breakings for DA-MTB. 

\begin{figure}[t!]
	\begin{center}
	
	\includegraphics[width=0.75\textwidth]{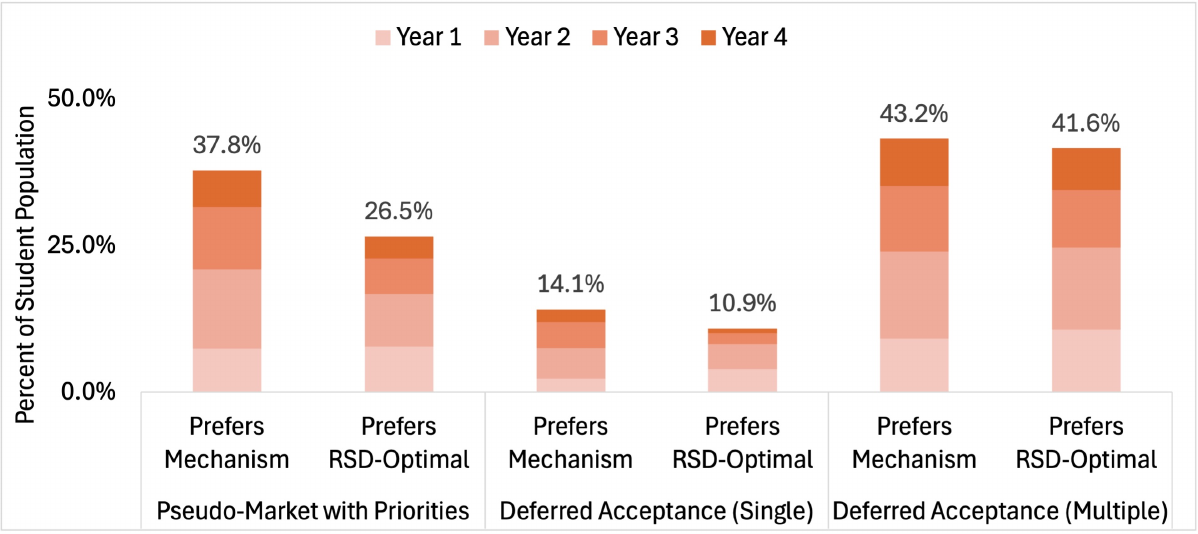}
	\caption{Each student's preferred mechanism by year of study.}
	\label{fig:preferred-mechanism}
	
	\end{center}
	  {\footnotesize \textit{Notes:} This figure reports the percent of students who strictly prefer \textit{the Pseudo-Market with Priorities}, \textit{Deferred Acceptance with single tie-breaking}, and \textit{Deferred Acceptance with multiple tie-breakings} to the benchmark of \textit{Random Serial Dictatorship with optimal course reserves (RSD-Optimal)}, and the percent of students who strictly prefer the benchmark. Results are averages across $100$ runs with different random component draws.  Within each run we consider only students with at least two different schedules across all four mechanisms. The population of affected students is the same across all comparisons, but it may differ across runs.}
\end{figure}

\vspace{2mm}
\noindent {\bf Student satisfaction.} Since only $11.2\%$ of courses are at or above capacity (see Table \ref{tab:binding-capacity}), many students obtain their most preferred schedule under all four mechanisms. Such assignments are uninformative for comparing relative student satisfaction across mechanisms. Thus, within each run we consider only students with at least two different schedules across all four mechanisms. This leaves, on average, $48.4\%$ of students with changing schedules.

We take the \emph{RSD-Optimal} as a benchmark. Figure \ref{fig:preferred-mechanism} displays results on the percentage of students who strictly prefer each of the PMP, DA-STB, and DA-MTB mechanisms to RSD-Optimal and the percentage who strictly prefer RSD-Optimal. Bars represent shares of students with changing schedules and are broken down by year of study. $37.8\%$ of students prefer PMP to RSD-Optimal, as opposed to $26.5\%$ of students preferring RSD-Optimal to PMP. This corresponds to over $300$ more students who strictly prefer the PMP mechanism. First-year students are an exception: a greater share strictly prefers RSD-Optimal to PMP. The weaker result for Year 1 students can be attributed to the gains for more senior students, who, averaging across runs, receive about $150$ more total course seats under PMP than under RSD-Optimal.

In contrast, a roughly equal percentage of students strictly prefer each DA mechanism to RSD-Optimal and strictly prefer RSD-Optimal. Among students with changing schedules, $75.0\%$ are indifferent between DA-STB and RSD-Optimal. These mechanisms coincide when RSD has no course reserves and both mechanisms use the same tie-breaking rule. Only $15.3\%$ of these students are indifferent between DA-MTB and RSD-Optimal, with $43.2\%$ of students strictly preferring DA-MTB and $41.6\%$ strictly preferring RSD-Optimal. These results are a direct consequence of DA-MTB breaking ties on a per-course basis, producing both winners and losers among students.
Direct calculations show that more students strictly prefer the PMP mechanism to each DA mechanism in each year of study.\footnote{The results of the direct comparison of DA-STB and DA-MTB with PMP are available upon request.}

\begin{figure}[t!]
	\begin{center}
	
	\includegraphics[width=0.75\textwidth]{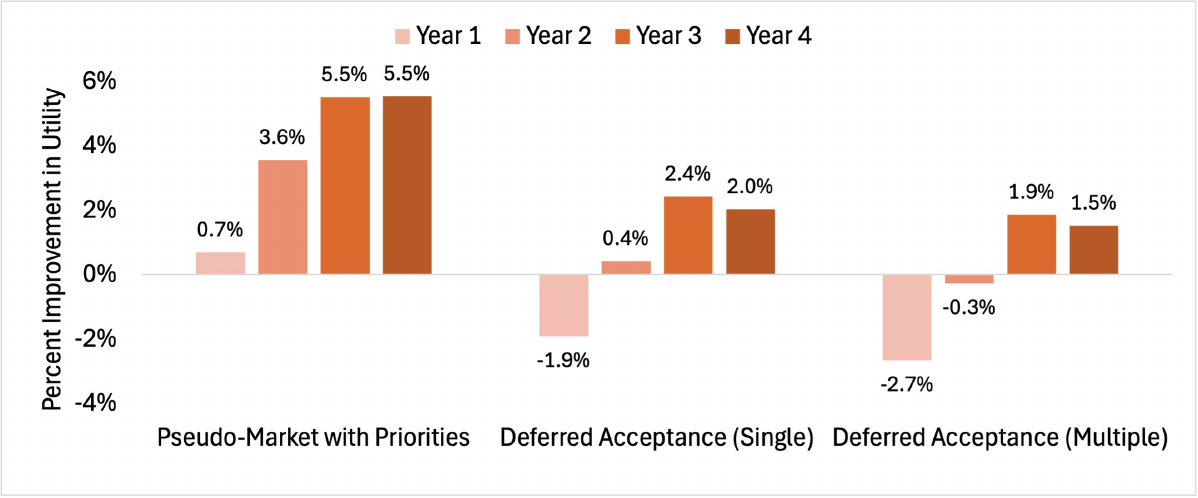}
	\caption{Improvements in utility for students in different years of study.}
	\label{fig:utility-all-changing-students}
	
	\end{center}
  {\footnotesize \textit{Notes:} This figure reports the percent improvement in mean utility for \textit{the Pseudo-Market with Priorities}, \textit{Deferred Acceptance with single tie-breaking}, and \textit{Deferred Acceptance with multiple tie-breakings} mechanisms over the benchmark of \textit{the Random Serial Dictatorship with optimal course reserves (RSD-Optimal)}. Results are averages across $100$ runs with different random component draws. Within each run we consider only students with at least two different schedules across all four mechanisms. The population of affected students is the same across all comparisons, but it may differ across runs.}
\end{figure}

Figure \ref{fig:utility-all-changing-students} presents each mechanism's percent improvement in mean utility relative to the RSD-Optimal mechanism. We again restrict attention to students with at least two different schedules across all four mechanisms. The PMP mechanism improves mean utility for all years of study, with larger gains for more senior students. Under both DA mechanisms, the percent improvement in mean utility is negative for first-year students and close to zero for second-year students. For third- and fourth-year students, the percent improvement in mean utility is similar between both mechanisms and these percentages are smaller than those for the PMP mechanism. Overall, the PMP mechanism provides the largest improvement in mean utility for all years of study compared to both DA mechanisms.

RSD-Optimal represents a reasonable upper bound for the performance of the first stage of the mechanism used in practice. At the other extreme, the RSD mechanism with base course reserves (RSD-Base) serves as a lower bound.\footnote{It is reasonable to expect that the performance of the actual mechanism used in practice lands somewhere in between these two extremes. No precise statement can be made, however, comparing these one-stage extremes to the two-stage mechanism used in practice.} Replacing RSD-Optimal with RSD-Base and completing the same exercise as in Figures 1 and 2, there are $80.0\%$ of students with at least two different schedules (across PMP, RSD-Base, DA-STB, and DA-MTB). Across these students, $70.3\%$ strictly prefer the PMP mechanism, with only $23.5\%$ preferring the RSD-Base outcome. The increase in mean utility is $21.5\%$ for second-year students and exceeds $50\%$ for third- and fourth-year students. Using the university's base reserves still improves outcomes for first-year students, who receive far greater access to reserved seats.

\begin{figure}[t!]
	\begin{center}
		
		\includegraphics[width=0.75\textwidth]{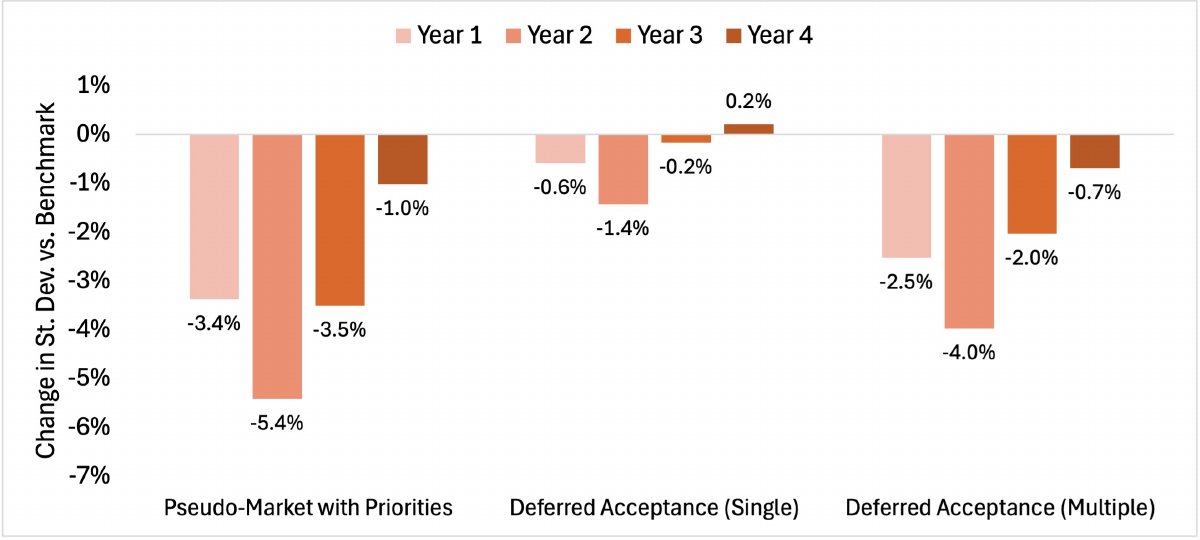}
		\caption{The change in the standard deviation of students' utilities.}
		\label{fig:st-dev-PMP-DA-DAm}
	\end{center}
	{\footnotesize \textit{Notes:} This figure reports the percent change in the standard deviation of students' utilities for \textit{the Pseudo-Market with Priorities}, \textit{Deferred Acceptance with single tie-breaking}, and \textit{Deferred Acceptance with multiple tie-breakings} mechanisms over the benchmark of \textit{the Random Serial Dictatorship with optimal course reserves (RSD-Optimal)}. Results are averages across $100$ runs with different random component draws and are based on \textit{the population of all students}.}
\end{figure}

\vspace{2mm}
\noindent {\bf Allocation fairness.} We compare allocation fairness across mechanisms using three metrics: the standard deviation of students' utilities, the amount of envy, and  priority violations. Removing students with the same schedules would affect either the definition of these metrics or the relative comparison across mechanisms, so we consider the population of all students. First, Figure \ref{fig:st-dev-PMP-DA-DAm} shows that, within each year of study, PMP and DA-MTB both decrease the standard deviation of student utilities relative to RSD-Optimal. The same result holds for all but fourth-year students under the DA-STB mechanism. For each year of study, this finding is the strongest for the PMP mechanism. For each mechanism, the largest reduction occurs for second-year students.

Second, Table \ref{tab:student-envy} reports results on schedule envy toward students of weakly lower priority in each course. The first column displays the percent of students without envy. The values in each subsequent column are obtained by removing one course from each envied schedules and checking whether the student still experiences envy. In line with our theoretical results for the PMP mechanism (see Theorem \ref{theorem:fairness}), no student has envy greater than a single course toward a student of weakly lower priority. As Table \ref{tab:student-envy} shows, this property is not satisfied by the other three mechanisms. About $9.4\%$ of students experience envy in the PMP mechanism, substantially less than in the RSD-Optimal, DA-STB, and DA-MTB mechanisms ($17.5\%$, $14.4\%$, and $13.6\%$, respectively). Relative to the RSD-Optimal benchmark, $8.1$ percentage points less of the population (or nearly 500 students) experience envy in the PMP mechanism.

\begin{table}[t!]
    \begin{center}
        \footnotesize
        
        \refstepcounter{tab}\label{tab:student-envy}
        \caption{The percentage of students who experience schedule envy.}
        
        \vspace{2mm}
        \begin{tabular}{p{0.265\textwidth}>{\centering}m{0.09\textwidth}>{\centering}m{0.09\textwidth}>{\centering}m{0.085\textwidth}>{\centering}m{0.10\textwidth}>{\centering}m{0.11\textwidth}>{\centering\arraybackslash}m{0.10\textwidth}}
            \hline
            \hline\\[-2mm]
            & 0 courses & 1 course & 2 courses & 3 courses & 4 courses & 5 courses\\[2mm]
            \hline\\[-2mm]
            \makecell[lc]{Pseudo-Market\\ with Priorities} & 90.6  (0.6) & 9.4  (0.6) & 0  (0.0) & 0  (0.0) & 0  (0.0) & 0  (0.0) \\[1.5mm]
            \hline\\[-2mm]
            \makecell[lc]{Random Serial Dictatorship\\ with optimal course reserves} & 82.5  (1.2) & 14.3 (1.0) & 2.6  (0.3) & 0.6  (0.2) & 0.06  (0.05) & 0.002  (0.0) \\[1.5mm]
            \hline\\[-2mm]
            \makecell[lc]{Deferred Acceptance\\ with single tie-breaking} & 85.6  (0.6) & 11.8  (0.6) & 2.1  (0.2)& 0.5  (0.2)& 0.03  (0.05)& \makebox[\linewidth][c]{0.002 (0.01)}\\[2mm]
            \hline\\[-2mm]
            \makecell[lc]{Deferred Acceptance\\ with multiple tie-breakings} & 86.4   (0.6)& 12.9  (0.6)& 0.7  (0.2)& 0.02  (0.02)& 0  (0.0)& 0  (0.0)\\[2mm]
            \hline
            \hline
        \end{tabular}
    \end{center}
    
    \vspace{2mm}
    {\footnotesize \textit{Notes:} This table reports the percentage of students who prefer the schedule of a student with weakly lower priority in each course in each of the four mechanisms. The first column includes students who experience no envy. The following columns display the percentage of students whose schedule envy is limited to $1,...,5$ courses. Results are averages across 100 runs with different random component draws and are based on \textit{the population of all students}.}    
\end{table}

Third, we examine fairness between levels of priority. The PMP mechanism is approximately stable, guaranteeing that no student can improve her schedule by taking a seat assigned to a student with strictly lower priority in the course (see Theorem \ref{theorem:stability}). As students' preferences are additive over courses, the DA-STB and DA-MTB mechanisms also ensure that no priorities are violated. However, in the RSD-Optimal mechanism, $16.8\%$ of students would benefit from taking a seat assigned to a student with strictly lower priority. These instances occur when a course reserves too few of its seats, meaning a reserve-eligible student selecting courses late in the RSD-Optimal mechanism loses out on a seat when a reserve-ineligible student selecting earlier obtains one.

Our simulations present important comparisons that would be difficult to derive from theory. While there is some variation across years of study, the PMP mechanism generally outperforms the other three mechanisms in  each measure of student satisfaction and allocation fairness.\footnote{In principle, it may be delivering better outcomes due to the market-clearing error, which allows some courses to be over-enrolled. The average market-clearing error across outcomes is $\approx 21.19$, less than half of the worst-case bound of $\approx 43.5$. In line with the bound in \cite{nguyenvohra2022}, no course is over-enrolled by more than $k - 1  = 4$ seats. About $3\%$ of courses are also over-enrolled, as opposed to $7.3\%$ in the data.} The comparison of DA-STB and DA-MTB in many-to-many settings is not readily available in the literature. Hopefully, these results will direct further theoretical investigations of these mechanisms. Appendix \ref{sec:appendix-additional-simulations-results} contains additional simulation results, including a description of the equilibrium prices in the PMP mechanism (Section \ref{subsec:appendix-pmp-prices}), a comparison with six alternative mechanisms (Section \ref{subsec:appendix-additional-sumulations-other-mechanisms}), results under an alternative priority structure (Section \ref{subsec:appendix-additional-sumulations-dept-first}), and further robustness checks (Section \ref{subsec:appendix-additional-simulations-comparative-statics}).

\vspace{2mm}
\label{paragraph:limitations}
\noindent \textbf{Limitations.} Our student utility model has limitations that a more detailed empirical analysis could address. Notably, assuming additive utilities may be too restrictive, as it fails to consider some complementarities and substitutabilities across courses. Courses covering similar content are substitutes, while a course requiring another as a corequisite (i.e., to be taken at the same time) are complements. Additive utilities allow us to compare the performance of the PMP mechanism with the DA-STB and DA-MTB mechanisms. Unlike these mechanisms, the PMP mechanism can be extended to more general preferences.

One step in expanding this analysis beyond additive utilities is to consider preferences over course sections rather than courses. Universities offer the same course at different times or with different instructors, restricting students from enrolling in two sections of the same course. In the data, many courses have several sections, with some much more popular than others. Since we aggregate enrollment and capacities across sections, if one of a course's sections is not at maximum capacity, the course cannot be either. $17.9\%$ of sections are at maximum capacity, as opposed to only $11.2\%$ of courses. Consequently, we expect the gains from the PMP mechanism presented earlier in this section to be understated.

Alternatively, we could consider additional factors that influence preferences. Our utility model assumes preferences only depend on a student's college and year of study. Allowing utilities to depend on a student's department would capture more heterogeneity in preferences. Information about past enrollment would disentangle preferences over courses from constraints based on the course's prerequisites, pinning down student choice sets. Past enrollment, along with other characteristics (e.g., demographics), are not available in our data.

\section{Discussion}
\label{sec:discussion}
\par This paper explores undergraduate course allocation as a many-to-many matching problem. We design a deterministic mechanism, the \textit{Pseudo-Market with Priorities} (PMP) mechanism, which extends the approximate competitive equilibrium from equal incomes mechanism of \cite{Budish2011} to accommodate course priorities and has the same worst-case bound on the market-clearing error. This mechanism satisfies approximate notions of stability, efficiency, fairness, and strategy-proofness. Our empirical findings exhibit that it outperforms an idealized version of the mechanism used in practice in student satisfaction and fairness.

	The gains from introducing the PMP mechanism will generally depend on university-specific factors. When a large set of equally-prioritized students compete for a small set of seats across courses, the commonly-used RSD mechanism produces allocations with significant envy. Instead, the PMP mechanism bounds envy among these students. Hence, the fairness improvements from introducing our mechanism will be large in environments with many over-demanded courses. Similarly, if a university uses a coarser priority structure, there will be more equally-prioritized students, increasing the gains from introducing the PMP mechanism. In contrast, for universities that employ finer priority structures (e.g., prioritizing students based on units completed or grade point average), the PMP and RSD mechanisms will result in similar outcomes.

A university's practical implementation of the PMP mechanism will also play a role in its success. Our model assumes students report preferences over all possible course schedules. In practice, reporting preferences requires a language simple enough for students to use and expressive enough to elicit students' actual preferences. 
The implementation of the A-CEEI mechanism in \cite{budish2017course} allows students to report cardinal values for individual courses and adjustments for pairs of courses. Changes to this language may be necessary to make it accessible to undergraduate students. The deployment of machine learning-based techniques may also improve the precision of elicited preferences \citep[][]{soumalias2023machine}.

Finding an equilibrium given reported preferences is computationally intense. \cite{budish2017course} estimate that 20 times the time used at Wharton would be necessary to implement the A-CEEI mechanism for the full student body at Ohio State University. However, implementing the PMP mechanism may be less challenging than this estimate suggests. Since a course's price is zero or unaffordable at all but one of its priority levels, students face few courses with a positive but affordable price. In our simulations, with additive student utilities, the average timing of the PMP mechanism improves on the A-CEEI mechanism by a factor of two (see Appendix \ref{subsec:appendix-additional-sumulations-other-mechanisms}). Still, incorporating scheduling constraints and preference complementarities may be challenging and should be carefully addressed in practice.

Finally, we note that our analysis applies to many other matching problems, such as assigning teaching assistants to courses and referees to soccer matches. Extending our results to these and other applications is left for future research.

\newpage
\appendix

\renewcommand{\theequation}{A\arabic{equation}} %
\setcounter{equation}{0}

\renewcommand{\thelemma}{A\arabic{lemma}}
\setcounter{lemma}{0}

\renewcommand{\thedefinition}{A\arabic{definition}}
\setcounter{definition}{0}

\renewcommand{\theproposition}{A\arabic{proposition}} %
\setcounter{proposition}{0}

\renewcommand{\thetheorem}{A\arabic{theorem}} %
\setcounter{theorem}{0}

\section{Appendix: Proofs}
\label{sec:appendix-proofs}

\noindent {\bf Proof of Theorem \ref{theorem:existence}.}  Consider a course allocation problem $({\cal S}, {\cal C}, Q, {\cal V}, \mathcal{R})$. We show that, for any $\beta > 0$, there exists a $(\sqrt{kM/2}, \beta)$-Pseudo-Market Equilibrium with Priorities.

Let $b = (b_1 , ..., b_N)$ be a budget vector that satisfies $1 \leq \min_s(b_s) \leq \max_s (b_s) \leq 1+ \beta$ for some $\beta>0$, and let $\overline{b} = 1 + \beta + \epsilon$ for a fixed $\epsilon > 0$. Consider the $M$-dimensional set ${\cal T} = [0, R\overline{b}]^M$, which conveniently parameterizes priority-specific prices and allows us to look for a competitive equilibrium in a lower dimensional space. In particular, for each parameter $t \in {\cal T}$, course $c \in {\cal C}$, and level of priority $r\in {\{1, ..., R\}}$, we define priority-specific prices as
	\begin{equation}
		p_{c,r}(t)=\max(t_{c}-(r-1)\overline{b},0). \label{eq:p(t)}
	\end{equation}
	With this definition, for each $t\in {\cal T}$ and $c\in {\cal C}$, there is a unique cutoff level of priority $r^*_c(t)\in {\{1, ..., R+1\}}$\footnote{Note that $r^*_c(t) = R+1$ is only achieved if $t_c = R\overline{b}$. The cutoff level $r^*_c$ can always be reduced to $R$.} such that for any $r\in \{1, ..., R\}$, $p_{c,r}(t)$ satisfies 
	\begin{equation}
		p_{c,r}(t)\in\begin{cases}
			\left\{ 0\right\}  & r>r^*_c(t)\\{}
			[0,\overline{b}) & r=r^{*}_c(t)\\
			[\overline{b},R\overline{b}] & r<r^{*}_c(t)
		\end{cases}.\label{eq:cutoff_p(t)}
	\end{equation}
	We will consider an auxiliary enlargement of this set, $\tilde{{\cal T}}= [-1, R\overline{b}+1]^M$ and define $p_{c,r}(\tilde{t})$ using equation (\ref{eq:p(t)}) for $\tilde{t}\in \tilde{\cal T}$. For each $s \in {\cal S}$, define the demand function $d_s: \tilde{{\cal T}} \rightarrow 2^{\cal C}$ by 
	\[
	d_s(\tilde{t})=\max_{\succsim_s}\left\{ x^{\prime}_s \subseteq {\cal C}:\sum_{c\in{\cal C}}x^{\prime}_{s, c}\max(\tilde{t}_{c}-(r_{s,c}-1)\overline{b},0)\leq b_{s}+\tau_{s, x'_s}\right\},
	\]
	where the $\tau_{s, x_s}$ are student- and schedule-specific taxes chosen to ensure that the demand is single-valued. Taxes are chosen to favor more-preferred bundles (i.e., if $x'_s \succ_s x_s$, then $\tau_{s, x'_s} > \tau_{s, x_s}$). For each course $c \in {\cal C}$, the excess demand function $z_c: \tilde{{\cal T}} \rightarrow \mathbb{Z}$ is defined by
	$$
	z_c(\tilde{t}) = \sum_{s\in {\cal S}} x^*_{s, c} - q_c,
	$$
	where $x^*_s = d_s(\tilde{t})$ for all $s \in {\cal S}$. Excess demand is bounded, as $-N \leq z_c \leq N-1$ for all $c \in {\cal C}$. We also define a budget surface for each student $s\in{\cal S}$ and schedule $x_s\subseteq {\cal C}$ as
	
	\[
	H(s,x_s)=\left\{ \tilde{t} \in \tilde{{\cal T}}:\sum_{c\in{\cal C}}x_{s, c}\max(\tilde{t}_{c}-(r_{s,c}-1)\overline{b},0)=b_{s}+\tau_{s, x_s}\right\}.
	\]
	Unlike the case without priorities, the budget surface $H(s,x_s)$ may not be a hyperplane \cite[see][]{Budish2011}. \label{lemma:b1}Lemma \ref{lem:tau} in Appendix \ref{sec:theorem-1-omitted-details} shows that it is still possible to choose $b_{s}$ and $\tau_{s, x_s}$ such that at most $M$ budget constraints intersect for any $\tilde{t}\in\tilde{{\cal T}}$.
	
	Next, we define a truncation function \text{trunc}$:\tilde{{\cal T}} \rightarrow {\cal T}$, where for each $c\in {\cal C}$,
	\begin{equation*}
	\text{trunc}(\tilde{t})_c = \min\{R\overline{b}, \max\{0, \tilde{t}_c\}\}.
	\end{equation*}
	Also, for $\gamma \in (0, 1/N)$, we define a t\^{a}tonnement price adjustment function $f: \tilde{{\cal T}} \rightarrow \tilde{{\cal T}}$ by
	\begin{equation*}
	f(\tilde{t}) = \text{trunc}(\tilde{t}) + \gamma z(\text{trunc}(\tilde{t})).
	\end{equation*}
	Suppose that $f$ has a fixed point $\tilde{t}^{*} = f(\tilde{t}^{*}),$ and denote its truncation by $t^*=\text{trunc}(\tilde{t}^{*})$. We show that the prices $\{p_{c, r}(t^*)\}_{c \in {\cal C}, {r \in \{1, ..., R\}}}$ defined by equation (\ref{eq:p(t)}), allocation $x^{*}_s = d_s(t^*)$, and budgets $b^{*}_s = b_s + \tau_{s,x^{*}_s}$ for all $s \in {\cal S}$ constitute an exact competitive equilibrium (or ($0,\beta$)-Pseudo-Market Equilibrium with Priorities as in Definition \ref{definition:competitive-equilibrium}). 
	\begin{itemize}
		\item The definition of the demand function $d_s(t^*)$ implies that any course schedule that student $s$ prefers to $x^{*}_s = d_s(t^*)$ must cost strictly more than $b^{*}_s = b_s + \tau_{s,x^{*}_s}$. 
		\item Prices $\{p_{c, r}(t^*)\}_{c \in {\cal C}, r \in {\{1, ..., R\}}}$ and cutoff levels of priority defined by (\ref{eq:p(t)}) and (\ref{eq:cutoff_p(t)}) ensure that condition (\ref{eq:price-comp}) of Definition \ref{definition:competitive-equilibrium} is satisfied. 
		\item $p_{c,1}(t^*)>0$ implies $z_c(t^{*}) = 0$. To see this, note that equation (\ref{eq:p(t)}) implies $\tilde{t}^*_c>0$. In addition, we must have $\tilde{t}_c^{*} < R\overline{b}$; otherwise (\text{trunc}$(\tilde{t}^*))_c=R\overline{b}$ and $z_c(\tilde{t}^*)<0$, which contradicts the fixed point equation. Hence, $\tilde{t}_c^{*} \in (0, R\overline{b})$, so the fixed point equation ensures that $z_c(t^{*}) = 0$.
		\item $p_{c,1}(t^*)=0$ implies $z_c(t^{*}) \leq 0$. Note that $p_{c,1}(t^*)=0$ means $\tilde{t}_c^{*}\in [-1,0]$. So, we have $t_c^{*}\equiv\text{trunc}(\tilde{t}^*_c)=0$ and $p_{c,1}(t^*)=0$, so the fixed point equation ensures that $z_c(t^*)\leq 0$.
		
	\end{itemize}
	Overall, if $f$ has a fixed point $\tilde{t}^{*} = f(\tilde{t}^{*})$, its truncation $t^{*}=\text{trunc}(\tilde{t}^{*})$ is a competitive equilibrium price vector for allocation $x^{*}_s = d_s(t^*)$ and budgets $b_s^{*}$ for all $s \in {\cal S}$.  Though, the fixed point of operator $f$ might fail to exist. Following \cite{Budish2011}, we define a ``convexification'' of $f$, $F:\tilde{\cal T} \rightarrow \tilde{\cal T}$, by
	\[F(\tilde{t}) = co\{y: \exists \text{ a sequence } \tilde{t}^w \rightarrow \tilde{t}, \tilde{t} \neq \tilde{t}^w \in \tilde{\cal T} \text{ such that } f(\tilde{t}^w) \rightarrow y\},\]
	where $co$ denotes the convex hull of the set. $F$ is nonempty and convex and $\tilde{\cal T}$ is compact and convex. From Lemma 2.4 of \cite{cromme1991fixed}, $F$ is an upper hemicontinuous correspondence and hence has a fixed point by Kakutani's fixed point theorem. We denote a fixed point by $\tilde{t}^{*} \in F(\tilde{t}^{*})$, and let again $t^{*} = \text{trunc}(\tilde{t}^{*})$ be its truncation.
	
	The remainder of the proof uses $F$ to show that there exist budgets for which excess demand at prices $t^*$ are at most a distance of $\sqrt{kM/2}$ from a perfect market-clearing excess demand vector. The reduction of the $M\!R$-dimensional space of prices $\{p_{c,r}\}_{c\in{\cal C}, {r\in {\{1, ..., R\}}}}$ to $M$-dimensional space $\tilde{{\cal T}}= [-1, R\overline{b}+1]^M$ allows us to use the steps of \cite{Budish2011} to establish the existence of a $(\sqrt{kM/2}, \beta)$-Pseudo-Market Equilibrium with Priorities. For completeness, we provide adapted steps in Appendix \ref{sec:theorem-1-omitted-details}.\qed

\vspace{4mm}
\noindent {\bf Proof of Theorem \ref{theorem:stability}.} Let $(x^*,p^*,b^*)$ be an outcome of the Pseudo-Market with Priorities mechanism and suppose each course $c\in {\cal C}$ instead has capacity $q'_c=\sum_{s\in {\cal S}}x^*_{s,c}$. We show that the allocation $x^*$ is stable. Since each course is at full capacity, $x^*$ must be feasible. Definitions \ref{definition:competitive-equilibrium} and \ref{definition:Pseudo-Market-Mechanism} ensure that $x^*_s$ is individually rational as well. Last, suppose that $(s, C)$ is a block. Then, $C= \max_{\succeq_s}\{x'_s\,:\,x'_s\subseteq x^*_s\cup C\}$, $C\neq x^*_s$, and for each $c\in C$ with $c\notin x^*_s$, there is a student $s'\in {\cal S}$ with $c\in x^*_{s'}$ and $r_{s',c}<r_{s,c}$. There must exist at least one course $c\in C$ with $c\notin x^*_s$ such that $p^*_{c,r_{s,c}}>0$; otherwise, student $s$ should have been assigned $C$ instead of $x^*_s$. The structure of equilibrium prices implies that for any student $s'$ with a lower priority than $s$ for course $c$ we must have $p^*_{c,r_{s',c}}\geq\overline{b}$. Therefore, we obtain a contradiction: $x^*$ cannot possibly assign a student $s'$ a seat in course $c$. Hence, there are no blocks, and the outcome of every Pseudo-Market with Priorities mechanism is approximately stable.	\qed

\vspace{4mm}
\noindent {\bf Proof of Theorem \ref{theorem:Pareto-improvements}.} Let $(x^*,p^*,b^*)$ be an outcome of the Pseudo-Market with Priorities mechanism. We show that the allocation $x^*$ is approximately priority-constrained efficient. Suppose each course $c\in {\cal C}$ has capacity $q'_c=\sum_{s\in {\cal S}}x^*_{s,c}$ and that $y$ is an allocation that Pareto dominates $x^*$. Allocation $y$ cannot assign more students to some course $c$ than $x^*$ does as it would imply that $\sum_{s\in {\cal S}}y_{s,c}>q'_c$. If $y$ assigns fewer students to some course $c$ than $x^*$ does, then $\sum_{s \in {\cal S}: r_{s,c} \geq 1} x^*_{s, c} > \sum_{s \in {\cal S}: r_{s,c} \geq 1} y_{s, c}$, confirming that $y_c \nsucceq_c x_c^*$.

Now, suppose that allocation $y$ has the same number of seats assigned in each course. Then, there is a student $s'$ who strictly prefers their assigned course schedule in $y$ to their schedule in $x^*$. For each student $s \in {\cal S}$, if $y_s \neq x^*_s$, then 

\vspace{-4mm}
$$
\sum_{c\in {\cal C}}p^*_{c,r_{s,c}}y_{s,c}\geq b^*_{s}\geq\sum_{c\in {\cal C}}p^*_{c,r_{s,c}}x^*_{s,c}.
$$
So, across all students, we have that

\vspace{-4mm}
$$
\sum_{s\in {\cal S}}\sum_{c\in {\cal C}}p^*_{c,r_{s,c}}y_{s,c}>\sum_{s\in {\cal S}}\sum_{c\in {\cal C}}p^*_{c,r_{s,c}}x^*_{s,c}.
$$
Rearranging,

\vspace{-4mm}
$$
\sum_{c\in {\cal C}}\sum_{r=1}^{R}p^*_{c,r}\sum_{s: \,r_{s,c}= r}y_{s,c}>\sum_{c\in {\cal C}}\sum_{r=1}^{R}p^*_{c,r}\sum_{s: \,r_{s,c}= r}x^*_{s,c}.
$$
Hence, we obtain

{\small
\vspace{-6mm}
$$
0<\sum_{c\in {\cal C}}\sum_{r=1}^{R}p^*_{c,r}\sum_{s: \,r_{s,c}= r}(y_{s,c}-x^*_{s,c})=\sum_{c\in {\cal C}}\sum_{r=2}^{R}(p^*_{c,r}-p^*_{c,r-1})\sum_{s: \,r_{s,c}\geq r}(y_{s,c}-x^*_{s,c})+p^*_{c,1}\sum_{s: \,r_{s,c}\geq 1}(y_{s,c}-x^*_{s,c}).
$$}

\noindent The final term must equal zero, as allocations $y$ and $x^*$ assign the same number of seats in each course. Also, $p^*_{c,r}-p^*_{c,r-1}\leq0$ for all levels $r$, so we must have $\sum_{s:\,r_{s,c}\geq r}(y_{s,c}-x^*_{s,c})<0$ for some course $c$ and level $r$. This implies $y_{c}\nsucceq_{c}x^*_{c}$ and confirms that allocation $x^*$ is approximately priority-constrained efficient.\qed

\vspace{4mm}
\noindent {\bf Proof of Theorem \ref{theorem:fairness}.}  Let $(x^*,p^*,b^*)$ be an outcome of the Pseudo-Market with Priorities mechanism and suppose that $\beta \leq \frac{1}{k-1}$. We show that the allocation $x^*$ has schedule envy bounded by a single course toward students of weakly lower levels of priority. Consider two students $s, s' \in \cal S$ such that $r_{s,c} \geq r_{s',c}$ for all $c \in {\cal C}$. Denote the prices faced by $s$ and $s'$ as $p^*_s=\{p^*_{c,r_{s,c}}\}_{c\in {\cal C}}$ and $p^*_{s^{\prime}}=\{p^*_{c,r_{s^{\prime},c}}\}_{c\in {\cal C}}$ and the course schedule assigned to student $s'$ as $x^*_{s^{\prime}}=\{c_{j_1},...,c_{j_{k^{\prime}}}\}$ for $c_{j_1},...,c_{j_{k^{\prime}}}\in \{1,...,M\}$ and $k^{\prime}\leq k$. Suppose that $s$ envies $s^{\prime}$ and $s$ cannot afford the course schedule of $s'$ even if one course from $x^*_{s^{\prime}}$ is dropped. That is,

\vspace{-4mm}
$$
		p^*_s\cdot x^*_{s^{\prime}}\backslash \{c_{j_\ell}\}>b^*_s
$$ 
for $\ell=1,...,k^{\prime}$. Since $r_{s,c} \geq r_{s',c}$ for all $c \in {\cal C}$, $p^*_{s'}\geq p^*_s$, guaranteeing that 

\vspace{-4mm}
$$
p^*_{s^{\prime}}\cdot x^*_{s^{\prime}}\backslash \{c_{j_\ell}\}>b^*_s
$$
for $\ell=1,...,k^{\prime}$.  As $b^*_{s^{\prime}}\geq p^*_{s^{\prime}}\cdot x^*_{s^{\prime}}$, summing these inequalities over $\ell$, we obtain

\vspace{-4mm}
$$
(k'-1)b^*_{s^{\prime}}\geq (k^{\prime}-1)p^*_{s^{\prime}}\cdot x^*_{s^{\prime}}>k'b^*_s.	
$$ 
This implies $\frac{b^*_{s^{\prime}}}{b^*_{s}}>\frac{k'}{k'-1}\geq\frac{k}{k-1}\geq 1+\beta$, which contradicts how budgets are allocated.\qed

\vspace{4mm}
\noindent {\bf Proof of Theorem \ref{theorem:strategy-proof}.} We first formally define a direct mechanism, a semi-anonymous direct mechanism, and the property of being \textit{strategy-proof in the large}. Next, we show that the PMP mechanism is a semi-anonymous direct mechanism that is \emph{envy-free but for tie-breaking}. The result then follows from Appendix C in \cite{AzevedoBudish2019}.

Consider a sequence of environments labeled by the number of students $|{\cal S}^N|\equiv N$. The set of courses ${\cal C}$ is fixed, with the number of courses equal to $M$, but the capacity of each course $q^N_c$ can vary. Each student $s \in {\cal S}^N$ has a type $v_s=(\succeq_s,r_s)$ that describes her preferences $\succeq_s$ over course schedules and her priorities $\{r_{s, c}\}_{c \in \cal C}$. The set of all possible types is denoted as ${\cal V}$. An allocation $x\in X^N_0\equiv \{0,1\}^{MN}$ specifies a schedule $x_s$ for each student $s$. Letting $X=\Delta X_0$ denote the set of random allocations, we assume that each student has a von Neumann-Morgenstern utility function $u_v:X\rightarrow [0,1]$ consistent with her type $v\in {\cal V}$. 

We consider a sequence of \emph{direct mechanisms} $\{\Phi^N\}_{N\in \mathbb{N}}$ that, for each report of student preferences, assigns a distribution over course allocations; that is, $\Phi^N:{\cal V}^N\rightarrow \Delta (X_0^N)$. The PMP mechanism is an example of a \textit{semi-anonymous direct mechanism}, where agents are divided into groups, agents within each group are treated the same way, and agents across groups can be treated differently. Formally, we partition ${\cal S}^N$ into groups according to their course priorities $\left\{S^N_g\right\}_{g\in \{1,...,R\}^M}$. This partition places students who have the same level of priority in every course in one group. All students in the same group face the same prices in the PMP mechanism and the same lottery over budgets.  However, these students might have different preferences $\succeq_s$ and, hence, types $v_s$.  We assume that each student can misrepresent her preferences, but not her priorities. 

To define the property of a direct mechanism being strategy-proof in the large, we consider function $\phi^N:{\cal V}\times \Delta {\cal V} \rightarrow X$ for each $N$ according to

\vspace*{-4mm}
$$
\phi^N(v_s,m) =\sum_{v_{-s}\in {\cal V}^{N-1}} \Phi_s^N(v_{s},v_{-s})Pr(v_{-s}|v_{-s}\sim iid(m)),
$$

\vspace*{-1mm}
\noindent where $\Phi_s^N(v_s,v_{-s})$ denotes the course schedule obtained by student $s$ when she reports type $v_s$ and all other students report $v_{-s}$. $Pr(v_{-s}|v_{-s}\sim iid(m))$ denotes the probability that profile $v_{-s}$ is realized when other students' types $v_{-s}$ are independent and identically distributed according to $m\in \Delta {\cal V}$. In other words, $\phi^N(v_s,m)$ describes the random outcome that student $s$ expects to receive when she reports $v_s$ and the other students' types are distributed independently and identically according to $m$---and they report their types truthfully. Using this notation, we have the following definition.
\begin{definition}
	\label{definition:strateg-proof-in-large}
	A semi-anonymous direct mechanism $\{\Phi^N\}_{N\in \mathbb{N}}$ is \textbf{strategy-proof in the large} if, for any random distribution of reports by other students $m\in \Delta {\cal V}$ with full support and $\varepsilon > 0$, there exists $N_0$ such that, for all $N\geq N_0$ and all $v'_s, v_s\in {\cal V}$ where students can misrepresent only their preferences, we have
	$u_{v_s} [\phi^N(v_s,m)] \geq u_{v_s} [\phi^N(v'_s,m)]-\varepsilon.$	
\end{definition}
\noindent In other words, in a large enough market, reporting preferences truthfully is approximately optimal, for any independent and identical distribution of the other students' types that has full support. Students can only misrepresent their preferences, not their priorities. 

Appendix C in \cite{AzevedoBudish2019} provides an argument explaining why any semi-anonymous mechanism that is \textit{envy-free but for tie-breaking} is strategy-proof in the large. We define envy-freeness but for tie-breaking below.
\begin{definition}
	\label{definition:semi-anonymous}
	A semi-anonymous direct mechanism $\{\Phi^N\}_{N\in {\cal N}}$ is \textbf{envy-free but for tie-breaking} if for each $N$ there exists $x^N:({\cal V}\times [0,1])^N\rightarrow \Delta (X_0^N)$, symmetric over its coordinates, such that $\Phi^N(v) = \int_{\ell \in [0, 1]^N} x^N(v, \ell) d\ell,$ and, for all $s,s',N,v$, and $\ell$, if $\ell_s\geq \ell_{s'}$, and $v_s$ and $v_{s'}$ belong to the same group $S_g$ (that is, $r_s=r_{s'}$), then $	u_{v_s}(x^N_s(v,\ell))\geq u_{v_s}(x^N_{s'}(v,\ell))$.
\end{definition}
\noindent The PMP mechanism assigns each student a budget using a uniformly distributed lottery $\ell_s\sim U[0,1]$ as $b_s=1 + \ell_s\beta$.  Because budgets are random, the PMP assignment for each $v\in {\cal V}^N$ can be represented as $\Phi^N(v) = \int_{\ell \in [0, 1]^N} x^N(v,\ell) d\ell,$ where $x^N(v, \ell)\in \{0,1\}^{MN}$ assigns a course schedule to each student in ${\cal S}^N.$  Definition \ref{definition:competitive-equilibrium} guarantees that student $s$ with budget $b_s$ (lottery number $l_s$) prefers her course schedule to the course schedule of student $s'$ in the same group with budget $b_{s'}<b_s$ ($\ell_{s'}<\ell_s$). This confirms that the PMP mechanism is strategy-proof in the large and completes the proof.\qed

\newpage
{\small
\begin{spacing}{0.9}
	\bibliographystyle{aea}
	\bibliography{references}
\end{spacing}}

\newpage
\section*{\textit{(For Online Publication)}}
\label{sec:online-appendix}

\renewcommand{\theequation}{B\arabic{equation}} %
\setcounter{equation}{0}

\renewcommand{\thelemma}{B\arabic{lemma}}
\setcounter{lemma}{0}

\renewcommand{\thedefinition}{B\arabic{definition}}
\setcounter{definition}{0}

\renewcommand{\theproposition}{B\arabic{proposition}} %
\setcounter{proposition}{0}

\renewcommand{\thetheorem}{B\arabic{theorem}} %
\setcounter{theorem}{0}

\renewcommand{\thetabApp}{B\arabic{tabApp}}
\setcounter{tabApp}{0}

\counterwithin{figure}{section}
\counterwithin{table}{section}
\counterwithin{tab}{section}

\section{Proof of Theorem 1 (Omitted Details)}
\label{sec:theorem-1-omitted-details}

Appendix \ref{sec:appendix-proofs} shows how Theorem 1 of \cite{Budish2011} can be extended to accommodate priority-specific course prices. Below, we confirm it is possible to choose budgets $b_{s}$ and $\tau_{s, x_s}$ such that at most $M$ budget constraints intersect for any $\tilde{t}\in\tilde{{\cal T}}$. This fact is used on p. \pageref{lemma:b1}. 
\begin{lemma}
	\label{lem:tau}
	One can choose taxes $\left\{ \tau_{s, x_s}\right\} _{s\in{\cal S},x_s\subseteq{{\cal C}}}$
	that satisfy the following conditions: 
	\begin{itemize}
		\item [(i)] Taxes are small: $-\varepsilon<\tau_{s, x_s}<\varepsilon$.
		\item [(ii)] Taxes favor more preferred bundles: $\tau_{s, x_s}>\tau_{s,x'_s}$
		for $x'_s\succ_{s}x_s$.
		\item [(iii)] The inequality bounds are preserved:
		\begin{equation*}
			1\leq\min_{s,x_s}(b_{s}+\tau_{s,x_s})\leq\max_{s,x_s}(b_{s}+\tau_{s,x_s})\leq1+\beta.
		\end{equation*}
		\item [(iv)] No perturbed budgets are equal: $b_{s}+\tau_{s,x_s}\neq b_{s'}+\tau_{s',x_{s'}}$.
		\item [(v)] There is no auxiliary price vector $\tilde{t}\in\tilde{\cal T}$
		at which more than $M$ budget constraints $H(s,x_s)$ intersect.
	\end{itemize}
\end{lemma}
\begin{proof}
	\cite{Budish2011} shows that it is possible to choose $\left\{ \tau_{s, x_s}\right\} _{s\in{\cal S},x_s\subseteq{\cal C}}$ that satisfy the first four conditions. We now establish that it is always possible to slightly change taxes such that condition (v) is also satisfied. Let us assume that more than $M$ budget constraints $H(s,x_s)$ intersect and let ${\cal I}$ be the set of students $s$ with some such budget constraint. Note that $|{\cal I}|$ can be smaller than $M$ if multiple budget constraints come from a single student $s$.
	
	 For each $\tilde{t}\in\tilde{\cal T}$, consider prices $\left\{ p_{c,r}(\tilde{t})\right\} _{c\in{\cal C}, r\in {\{1, ..., R\}}}$ defined by equation (\ref{eq:p(t)}) and the cutoffs defined by (\ref{eq:cutoff_p(t)}). 
	The definition of cutoffs $r^*_c(\tilde{t})$ implies that 
	
	\vspace{-4mm}
	$$
	\forall s\in{\cal I},c\in{\cal C}: r_{s,c}>r^{*}_c(\tilde{t}),\hspace{0.3cm}x_{s, c}\cdot p_{c,r_{s,c}}(\tilde{t})=0.
	$$
	In addition, the definition of cutoffs implies $p_{c,r_{s,c}}\geq \overline{b}$ for $1\leq r_{s,c}<r^{*}_c(\tilde{t})$. Therefore, a seat in course $c$ is not allocated to agent $s$ for $r_{s,c}<r^{*}_c(\tilde{t})$; that is, $x_{s, c}=0.$ Hence, we also have
	
	\vspace{-4mm}
	$$
	\forall s\in{\cal I},c\in{\cal C}: r_{s,c}<r^{*}_c(\tilde{t}),\hspace{0.3cm}x_{s, c}\cdot p_{c,r_{s,c}}(\tilde{t})=0.
	$$
	Therefore, we obtain that $x_{s, c}\cdot p_{c,r_{s,c}}(\tilde{t})$ might be non-zero only if $r_{s,c}=r^{*}_c (\tilde{t})$. That is, for $s \in {\cal I}$, entries in agent $s$'s budget constraint are non-zero only if $r_{s,c}=r^*_c(\tilde{t})$. Hence, for a given $\tilde{t} \in \tilde{\cal T}$, if $M$ budget constraints $H(s,x_s)$ intersect at $\tilde{t}$, we obtain that $\tilde{t}$ solves the set of linear equations below:
	\[
	\left\{\begin{array}{cc}
		x_{s,1}\cdot p_{1, r_{s, 1}}(\tilde{t})+x_{s,2}\cdot p_{2, r_{s, 2}}(\tilde{t})+...+x_{s,M}\cdot p_{M, r_{s, M}}(\tilde{t}) &=b_{s}+\tau_{s, x_s}\\
		.... & ....\\
		x_{s', 1}\cdot p_{1, r_{s', 1}}(\tilde{t})+x_{s', 2}\cdot p_{2, r_{s', 2}}(\tilde{t})+...+x_{s', M}\cdot p_{M, r_{s', M}}(\tilde{t}) &=b_{s'}+\tau_{s', x'_{s'}}
	\end{array}\right.
	\]
	Here, $s,s'\in {\cal I}$, the coefficients $x_{s, c} \in \{0,1\}$ are in line with the bundles $x_s$ on the budget constraints, and each price satisfies $p_{c, r_{s, c}} = p_{c, r^*_c}$ or is such that $x_{s,c}\cdot p_{c, r_{s, c}} = 0$. Equivalently, setting $\tilde{x}_{s, c} = 1$ if $x_{s, c} = 1$ and $r_{s, c} = r^*_c$ and $\tilde{x}_{s, c} = 0$, we have the following set of equations:
	
	\[
	\left\{\begin{array}{cc}
		\tilde{x}_{s,1}\cdot p_{1, r^*_1}(\tilde{t})+\tilde{x}_{s,2}\cdot p_{2, r^*_2}(\tilde{t})+...+\tilde{x}_{s,M}\cdot p_{M, r^*_M}(\tilde{t}) &=b_{s}+\tau_{s, x_s}\\
		.... & ....\\
		\tilde{x}_{s', 1}\cdot p_{1, r^*_1}(\tilde{t})+\tilde{x}_{s', 2}\cdot p_{2, r^*_2}(\tilde{t})+...+\tilde{x}_{s', M}\cdot p_{M, r^*_M}(\tilde{t}) &=b_{s'}+\tau_{s', x'_{s'}}
	\end{array}\right.
	\]

	There are more than $M$ such equations, but at most $M$ independent linear equations in prices $\{p_{1, r^*_1}(\tilde{t}), p_{2, r^*_2}(\tilde{t}), ..., p_{M, r^*_M}(\tilde{t})\}$. So, the Rouch\'e--Capelli theorem implies that we can choose $\tau_{s, x_s}$ such that only at most $M$ equations are satisfied for any $\tilde{t}$.\footnote{Note that when we change $\tau_{s, x_s}$, some other budget constraints might start intersecting. Since the set of possible intersecting budget constraints is finite and $\tau_{s, x_s}$ varies continuously, we can always choose $\tau_{s, x_s}$ without influencing the intersection property of the other budget constraints.}
\end{proof}

\vspace{2mm}
Our proof of Theorem \ref{theorem:existence} in Appendix \ref{sec:appendix-proofs} extends steps 1-3 of the proof of Theorem 1 from \cite{Budish2011} to accommodate priority-specific course prices. The rest of this section adapts the remaining steps from this proof to our setting.

\vspace{2mm}
\noindent {\bf Step 4.} Suppose that $t^{*}$ is not on any budget constraint. We show that $t^{*}$ is an exact competitive equilibrium price vector. There is a neighborhood around $t^{*}$ where each agent's demand is unchanging in price. At price $t^{*}$, $f$ is continuous, and as a result, $F(t^{*}) = f(t^{*})$. 
\begin{itemize}
	\item If $t^{*} = \tilde{t}^{*}$, then $F(\tilde{t}^{*}) = F(t^{*}) = f(t^{*})$ and, thus, $t^{*} = \tilde{t}^{*} \in F(\tilde{t}^{*})= f(\tilde{t}^{*})$. Therefore, $t^{*}$ is a fixed point. Hence, as we established earlier, it is an exact competitive equilibrium price vector.
	\item If $t^{*} \neq \tilde{t}^{*}$, we establish the following lemma.
	\begin{lemma}
		For any $\tilde{t} \in \tilde{\cal T}\setminus {\cal T}$, (i) $f(\tilde{t}) = f($\textup{trunc}$(\tilde{t}))$ and (ii) $F(\tilde{t}) \subseteq F($\textup{trunc}$(\tilde{t}))$.
		\label{lem:suplementary}
	\end{lemma}

	\vspace{-6mm}
	\begin{proof}
		Statement (i) follows from the definition of $f$. For statement (ii), consider a point $y$ for which there exists a sequence $\tilde{t}^w \rightarrow \tilde{t}$, $\tilde{t}^w \neq \tilde{t}$ such that $f(\tilde{t}^w) \rightarrow y$. Consider trunc$(\tilde{t}^w)$. As trunc$(\cdot)$ is continuous, this sequence will converge to  trunc$(\tilde{t})$. Statement (i) implies $f($trunc$(\tilde{t}^w))$ converges to $y$. As a result, $y \in F(\tilde{t})$ implies that $y \in F($trunc$(\tilde{t}))$ and, thus, $F(\tilde{t}) \subseteq F($trunc$(\tilde{t}))$.
	\end{proof}
	As a result, since $F(t^{*}) = f(t^{*})$ and $\tilde{t}^{*} \in F(\tilde{t}^{*})$, $\tilde{t}^{*} \in F(t^{*}) = f(t^{*}) = f(\tilde{t}^{*})$ and, thus, $\tilde{t}^{*} = f(\tilde{t}^{*})$. Therefore, $t^{*}$ is an exact competitive equilibrium price vector.
\end{itemize}

\vspace{2mm}
\noindent {\bf Step 5.} Next, suppose that $t^{*}$ is on $1 \leq L \leq M$ budget constraints. We denote $\Phi = \{0, 1\}^L$ and construct a set of $2^L$ price vectors $\{t^{\phi}\}_{\phi \in \Phi}$ that satisfy the following conditions:
\begin{enumerate}
	\item Each $t^{\phi}$ is close enough to $t^{*}$ such that there is a path from $t^{\phi}$ to $t^{*}$ that does not cross any budget constraint.
	\item Each $t^{\phi}$ is on the ``affordable'' side of the $\ell$th budget constraint if $\phi_{\ell} = 0$ and is on the ``unaffordable'' side if $\phi_{\ell} = 1$.
\end{enumerate}
To construct vectors $\{t^{\phi}\}_{\phi \in \Phi}$, note that each of the $L$ intersecting budget constraints defines two sets:
\begin{align*}
	H^0_{\ell} &= \left\{ \tilde{t} \in \tilde{{\cal T}}:\sum_{c\in{\cal C}}x_{s_{\ell}c}\max(\tilde{t}_{c}-(r_{s_{\ell},c}-1)\overline{b},0)\leq b_{s_{\ell}}+\tau_{s_{\ell}, x_{s_{\ell}}}\right\}\\[2mm]
	H^1_{\ell} &= \left\{ \tilde{t} \in \tilde{{\cal T}}:\sum_{c\in{\cal C}}x_{s_{\ell}c}\max(\tilde{t}_{c}-(r_{s_{\ell},c}-1)\overline{b},0)> b_{s_{\ell}}+\tau_{s_{\ell}, x_{s_{\ell}}}\right\}
\end{align*}
The first set delineates the set of prices for which agent $s_{\ell}$ can afford schedule $x_{s_\ell}$, whereas the second set delineates the set of prices for which agent $s_{\ell}$ can't afford $x_{s_{\ell}}$. Let $\phi = (\phi_1, ..., \phi_L) \in \Phi$ be an $L$-dimensional vector of zeros and ones, and the polytope $\pi(\phi) := \cap_{\ell = 1}^L H^{\phi_{\ell}}_{\ell}$ be the set of points in ${\cal T}$ that belongs to the intersection of sets indexed by $\phi$. Let $H=\{H(s,x_s)\}_{s\in{\cal I},x_s\subseteq{\cal C}}$ be the finite set of all budget constraints formed by any student-schedule pair $(s, x_s)$. We then define the distance
\[ 
\delta < \inf_{\tilde{t}^{\prime\prime} \in {\cal T}, H \in {\cal H}} \{||(t^{*} - \tilde{t}^{\prime\prime})||_2 : \tilde{t}^{\prime\prime} \in H, t^{*} \notin H\},
\]
which is such that any budget constraint that $t^{*}$ does not belong to is further than $\delta$ away from $t^{*}$. Now, for each $\phi\in \Phi$, we define $t^{\phi}$ by an arbitrary element of $\pi(\phi) \cap B_{\delta}(t^{*})$, where $B_{\delta}(t^{*})$ is a $\delta$-ball of $t^{*}$. Then, the set of price vectors $\{t^{\phi}\}_{\phi \in \Phi}$ satisfies the two requirements above.

\vspace{2mm}
\noindent {\bf Step 6.} We now show that a perfect market-clearing excess demand vector lies in the convex hull of $\{z(t^\phi)\}_{\phi\in \Phi}$. For this purpose, we first show that for any $y\in F(t^*)$, we must have $y=t^*+\sum_{\phi\in \Phi} \lambda^\phi \gamma z(t^\phi)$, where $\sum_{\phi\in \Phi} \lambda^\phi=1.$ Take some $y\in F(t^*)$. Consider a sequence $t^w \rightarrow t^*, t^* \neq t^w \in \tilde{\cal T} \text{ such that } f(t^w) \rightarrow y'$. Note that the sequence $t^{w}$ consists of a finite number of subsequences $t^{w,\phi} \in\pi(\phi)\cap B_\delta(t^*)$ for some $\phi\in \Phi$.\footnote{Some subsequences can have only a finite number of elements.} Since all elements of the subsequence $t^{w,\phi}$ are on the same side of set $H^\phi_\ell$, every agent has the same choice at every point. Hence, if the subsequence has an infinite number of elements, we have $t^{w,\phi} \rightarrow t^*$ and 

\vspace{-4mm}
$$
f(t^{w,\phi})\rightarrow t^*+\gamma z(t^\phi).
$$
Therefore, the limit of $f(t^w)$ for the original sequence $t^w$ must be also $t^*+\gamma z(t^{\phi'})$ for some $\phi'\in \Phi$. So, $y'=t^*+\gamma z(t^{\phi'})$. Since, by definition, $y$ is a convex combination of such $y'$, we must have $y=t^*+\sum_{\phi\in \Phi} \lambda^\phi \gamma z(t^\phi)$, where $\sum_{\phi\in \Phi} \lambda^\phi=1.$

Lemma \ref{lem:suplementary} implies that $\tilde{t}^*\in F(\tilde{t}^*)\subseteq F(t^*)$. Hence, we must have 
$$
\tilde{t}^*=t^*+\sum_{\phi\in \Phi} \lambda^\phi \gamma z(t^\phi).
$$
for nonnegative $\{\lambda_\phi\}_{\phi \in \Phi}$ with $\sum_{\phi\in \Phi} \lambda^\phi=1.$ We also denote
$$
\zeta=\sum_{\phi\in\Phi}\lambda^\phi z(t^\phi)=\frac{\tilde{t}^*-t^*}{\gamma}.
$$ 
We note that $\zeta\leq 0$ and $\zeta_c< 0$ imply $t_c^*=0$. Hence, vector $\zeta$ is a perfect market-clearing excess demand vector, and it lies in the convex hull of $\{z(t^\phi)\}_{\phi\in \Phi}$.

\vspace{2mm}
\noindent {\bf Steps 7-8.} The structure of excess demand has the same geometric structure as in \cite{Budish2011}. In particular, let $L'$ be the number of agents whose budget constraints intersect at price $t^*$, and renumber agents so that $s=1,...,L'$. We denote the number of budgets of student $s$ that intersect at $t^*$ as $w_s$. Since at most $M$ budgets constraints can intersect, we must have $L\equiv\sum_{s=1}^{L'}w_s\leq M$. We also denote the bundles pertaining to $s$'s budget constraints as $x_s^1\succ ... \succ x_s^{w_s}$.

Similarly to \cite{Budish2011}, we consider bundles that $s$ demands at prices near $t^*$. In the set $H^0(s, x^1_s)$, agent $s$ can purchase her favorite bundle $x^1_s$. Hence, one does not need to know whether prices belong to sets $H^0(s, x_s^2), ..., H^0(s, x_s^{w_s})$. Let us denote the demand for prices at $H^0(s, x^1_s)\cap B_\delta(t^*)\cap\tilde{\cal T}$ as $d_s^0$. Similarly, we consider prices in $H^1(s, x^m_s)\cap H^0(s, x^{m+1}_s)$ and denote the corresponding demands as $d_s^m$ for $m=1,...,w_s$. Overall, agent $s=1,..., L'$ purchases $w_s+1$ distinct bundles at prices near to $t^*$. 

Let us denote the excess demand of the remaining agents as
$$
z_{S\backslash \{1,...,L'\}}(t^*)=\sum_{s=L'+1}^S d_s(t^*)-q.
$$
A perfect market-clearing excess demand vector lies in the convex hull of $\{z(t^\phi)\}_{\phi\in \Phi}$ with the elements 
$$
z_{S\backslash \{1,...,L'\}}(t^*)+\sum_{s=1}^{L'}\sum_{f=1}^{w_s}a_s^fd^f_s
$$
where $0\leq a_s^f\leq 1, s=1,...,L',f=1,...,w_s$ and $\sum_{f=1}^{w_s}a_s^f=1,$ $s=1,...,L'$. \cite{Budish2011} shows in Step 8 of Theorem 1 that there exists an element of the above geometric structure, $z(t^{\phi'})$, within $\sqrt{kM/2}$ distance from the perfect market-clearing excess demand vector. This purely mathematical argument remains unchanged.

\vspace{2mm} 
\label{theorem-1-step-9}\noindent {\bf Step 9.} While $z(t^{\phi'})$ is within $\sqrt{kM/2}$ distance from the perfect market-clearing excess demand vector, it may be that $t^{\phi'} \in \tilde{\mathcal{T}} \setminus \mathcal{T}$ rather than $\mathcal{T}$. Since we know that $t^* \in \mathcal{T}$, we can alter student budgets $b$ so that excess demand at $t^*$ equals $z(t^{\phi'})$.

For any student $s \in \mathcal{S} \setminus \{1, ..., L'\}$, we know that $s$ selects the same course schedule at $t^{\phi'}$ and $t^*$. As a result, setting $b^*_s = b_s + \tau'_{s, x^*_s}$, the price of any bundle preferred by $s$ to $x^*_s$ costs more than $b^*_s$. 

Any student $s \in \{1, ..., L'\}$ may select a different course schedules at $t^{\phi'}$ and $t^*$. For each such student, for each intersecting budget constraint $m = 1, ..., w_s$, if $t^{\phi'} \in H^1(s, x^m_s)$, then the course schedule $x^m_s$ is unaffordable for student $s$ at $t^{\phi'}$. To guarantee that $x^m_s$ is unaffordable at $t^*$ as well, as in \cite{Budish2011}, we change the reverse tax $\tau'_{s, x^m_s}$ to be $\tau''_{s, x^m_s} = \tau'_{s, x^m_s} - \delta_2$, where $\delta_2 > 0$ is small enough to ensure that conditions (i), (ii), and (iii) from Lemma \ref{lem:tau} still hold. If $t^{\phi'} \in H^0(s, x^m_s)$ (or, for any bundle $x \neq x^m_s$ for $m = 1, ..., w_s$), set $\tau''_{s, x} = \tau'_{s, x}$. Taxes still favor more preferred bundles, i.e., $\tau''_{s, x_s} > \tau''_{s x'_s}$ if $x_s \succ_s x'_s$. Then, letting $x^*_s = d_s(t^*, b_s, \tau''_s)$ and setting $b^*_s = b_s + \tau''_{s, x^*_s}$, the price of any bundle preferred by $s$ to $x^*_s$ costs more than $b^*_s$. 

As a result, with all taxes equal to $0$, at budgets $b^*$ and priority-specific course prices $p^*$ derived from $t^*$ using equation (\ref{eq:p(t)}), the resulting allocation $x^*$ generates excess demand equal to $z(t^{\phi'})$, which has a market-clearing error no greater than $\sqrt{kM/2}$.
\qed

\newpage
\renewcommand{\theequation}{C\arabic{equation}} %
\setcounter{equation}{0}

\renewcommand{\thelemma}{C\arabic{lemma}}
\setcounter{lemma}{0}

\renewcommand{\thedefinition}{C\arabic{definition}}
\setcounter{definition}{0}

\renewcommand{\theproposition}{C\arabic{proposition}} %
\setcounter{proposition}{0}

\renewcommand{\thetheorem}{C\arabic{theorem}} %
\setcounter{theorem}{0}

\renewcommand{\thetabApp}{C\arabic{tabApp}}
\setcounter{tabApp}{0}

\section{The Student Utility Model Calibration}
\label{sec:appendix-student-utility-estimation}

This appendix describes the method of the simulated moments approach used to calibrate the student utility model. First, Section \ref{subsec:appendix-Model} outlines the student utility model. Then, Section \ref{subsec:course-reserve-adjustment} discusses adjustments to the university's course reserves data. The empirical moments used to calibrate the model are described in Section \ref{subsec:appendix-moment-conditions}. Sections \ref{subsec:appendix-calibration} and  \ref{subsec:appendix-inference} provide the calibration process, calibration results, and details on model fit.

\subsection{Model}
\label{subsec:appendix-Model}
We consider a model with \textit{student horizontal preferences}, \textit{course vertical components}, and an \textit{interaction term}. In particular, $u_{sc}$, student $s$'s utility from taking course $c$, takes the following parametric form:

\vspace{-2mm}
\begin{equation}
	\label{eq:student-utility-appendix}
	u_{sc} = Y'_s \theta + Y'_s\Gamma H_c + z_c + \varepsilon_{sc},
\end{equation}
where:
\begin{itemize}
	\item $Y_s$ and $H_c$ are binary vectors of student $s$'s observed characteristics and course $c$'s observed characteristics, each with one nonzero element. Denote the set of colleges by ${\cal A}=\{A,...,G\}$ and years of study by ${\cal Y}=\{1,...,4\}$. Then, $Y_s\in \{0,1\}^{|{\cal A}||{\cal Y}|}$ sets the college and year of study of student $s$ and $H_c\in \{0,1\}^{|{\cal A}|}$ sets the college of course $c$.
	\item $\theta\in \mathbb{R}^{|{\cal A}||{\cal Y}|}$  and $z=\{z_c\}_{c\in {\cal C}}$ are two vectors of unobserved coefficients; $\Gamma$ is a matrix of unobserved coefficients conformable with vectors $Y_s$ and $H_c$.
	\item $\{\varepsilon_{sc}\}_{s\in {\cal S}, c\in {\cal C}}$ are random utility components that are assumed to be identically and independently distributed according to the standard normal distribution $\varepsilon_{sc}\sim N(0,1)$.  
\end{itemize}

Each student can be enrolled in up to a maximum of five courses. We assume that the utility from taking multiple courses equals the sum of utilities of taking individual courses. The utility of taking no course is normalized to $u_{s0}=0$. 

The parametric form for student utilities in equation \eqref{eq:student-utility-appendix} has three non-random components. The first, $Y'_s\theta$, is \textit{a horizontal component} that is determined solely by the student's characteristics. The vector of coefficients $\theta\in \mathbb{R}^{|{\cal A}||{\cal Y}|}$ has $28$ elements in total to calibrate. The horizontal component determines the importance of taking courses for student $s$, with a larger value meaning student $s$ wants to enroll in more courses. This accounts for variability in the data in the number of courses taken by students from different colleges and years of study, as shown in Table \ref{tab:students-n-courses-taken}.

The second component, $Y'_s\Gamma H_c$, is an \textit{interaction term} that determines the heterogeneity of students' preferences over courses. We restrict the matrix $\Gamma$ to have $49$ parameters $\gamma_{aa'}$, where $a$ is student $s$'s college and $a'$ is course $c$'s college. $\Gamma$ accounts for the heterogeneity in patterns of course enrollment across colleges. For example, Table \ref{tab:percent-courses-outside-college} shows students from college $A$ take $71\%$ of courses from their home college, while more than $95\%$ of courses taken by students from college $C$ are courses outside their home college. 

The third component $z_c$ is a \textit{vertical component} that determines the attractiveness of course $c$. This component is common among all students.

\renewcommand{\arraystretch}{1.2} 
\newcommand{\cwidth}{1.4cm}
\begin{table}[t!]
	\begin{center}
		\refstepcounter{tab}\label{tab:students-n-courses-taken}
		\caption{The average number of courses taken by students.}
		\begin{tabular}{C{\cwidth} |C{\cwidth}|C{\cwidth}|C{\cwidth}|C{\cwidth}|C{\cwidth}}
		\multicolumn{1}{c}{} & \multicolumn{1}{c}{Year 1} & \multicolumn{1}{c}{Year 2} & \multicolumn{1}{c}{Year 3} & \multicolumn{1}{c}{Year 4} \\[2mm]
			A & 3.6 & 3.4 & 3.1 & 2.9 \\
			B & 4.0 & 4.3 & 3.9 & 3.0 \\
			C & 4.7 & 4.3 & 4.2 & 3.5 \\
			D & 4.4 & 4.3 & 4.1 & 3.4 \\
			E & 4.3 & 4.3 & 3.8 & 3.3 \\
			F & 4.2 & 4.0 & 3.8 & 2.9 \\
			G & 4.6 & 4.7 & 4.7 & 4.1 \\ 
		\end{tabular}
	\end{center}
	\vspace{1mm}
	\begin{spacing}{0.8}
	{\footnotesize \textit{Notes:} Each entry lists the average number of courses taken by students in the row's college and column's year of study.}
\end{spacing}
\end{table}

\begin{table}[t!]
	\begin{center}
		\refstepcounter{tab}\label{tab:percent-courses-outside-college}
		\caption{Percentages of courses taken by students across colleges.}
		\begin{tabular}{C{\cwidth} |C{\cwidth}|C{\cwidth}|C{\cwidth}|C{\cwidth}|C{\cwidth}|C{\cwidth}|C{\cwidth}}
			\multicolumn{1}{c}{\%} & \multicolumn{1}{c}{A} & \multicolumn{1}{c}{B} & \multicolumn{1}{c}{C} & \multicolumn{1}{c}{D} & \multicolumn{1}{c}{E} & \multicolumn{1}{c}{F} & \multicolumn{1}{c}{G} \\[2mm]			
			A & 71.4  & 0.5 & 1.9  & 16.9  & 1.7  & 5.0  & 2.7  \\
			B & 1.9  & 43.8 & 0.6  & 18.4  & 19.0  & 12.2  & 4.1  \\
			C & 10.2  & 1.3  & 4.5  & 37.2  & 29.2  & 15.2  & 2.4  \\
			D & 2.5  & 1.0  & 0.8  & 63.9  & 12.0  & 13.0  & 6.8  \\
			E & 1.2  & 2.6  & 0.1 & 24.1  & 53.7  & 16.1  & 2.1  \\
			F & 1.5  & 1.8  & 0.8  & 23.3  & 17.9  & 52.4  & 2.3  \\
			G & 1.3  & 0.4  & 0.3  & 31.3  & 7.7  & 6.1  & 53.0  \\

		\end{tabular}
	\end{center}
	\vspace{1mm}
	\begin{spacing}{0.8}
	{\footnotesize \textit{Notes:} This table reports the percentage of courses taken by students across colleges. The rows correspond to students' colleges, and the columns correspond to courses' colleges.}
\end{spacing}
\end{table}

\vspace{2mm}
\noindent {\bf Parameter normalization.} To correctly interpret the interaction term $\Gamma$ and the vertical component $z$, we need to further normalize the model's parameters. As described above, for student $s$ from college $a$ and course $c$ from college $a'$, we have $Y'_s\Gamma H_c=\gamma_{aa'}$. Let $S_a$ and $S_{ay}$ be the set of students in college $a$ and the set of students in college $a$ and year $y$, respectively. The average utility of students in college $a$ from taking course $c$ in college $a'$ equals

\vspace{-4mm}
$$
	\overline{u}_{ac}=\frac{1}{|S_{a}|}\sum_{s\in {\cal S}_{a}}u_{sc}
	=\overline{\theta}_{a}+\gamma_{aa'}+z_{c},
$$
where $\overline{\theta}_{a}=\frac{1}{|S_{a}|}\sum_y|S_{ay}|\theta_{ay}$ is the average student-specific component among students in college $a$. We have three additive parameters that lead to the same value of the utility components $\overline{u}_{ac}$ and need a normalization. 

One possible normalization is $\overline{\theta}_a=0$ and $\gamma_{aa}=0$ for $a=A,...,G$. This normalization leads to the interpretation of $z_c$ as \textit{the average popularity} of the course among students from the same college. Then, $\gamma_{aa'}$ is \textit{an increase in student's utility} from taking a course in college $a'$ compared to a course in student's college $a$ with the same vertical component.

\vspace{2mm}
\noindent {\bf Students' choice sets.} The student choice sets are unobserved in the data. We do not track students' past enrollments, and as a result, we do not know each student's set of courses for which they satisfy course prerequisites. 
This complication requires us to take a stance on how students' choice sets are formed. We assume that each student's choice set is limited to $80$ courses. Each course in the student's choice set is drawn randomly, with probability proportional to the share of seats in that course taken by students from the same college and year of study in the actual data. Additionally, we make adjustments for courses with course reservations. For each course reservation with $r$ reserved seats, we ensure that the course is included in the choice sets of at least $2r$ students who qualify for the reservation. This adjustment prevents course reservations from artificially limiting a course's capacity due to an insufficient number of qualified students.\footnote{If there are fewer than $2r$ qualified students, the course is included in the choice set of all such students.}

Random choice sets are a limitation of our model. Any changes in the way we draw the course lists could potentially change the calibrated parameter values. \cite{crawford2021survey} provide an insightful discussion of preference estimation techniques with unobserved choice sets. Still, the size of the student choice sets we consider is reasonable, as previous studies report that students typically consider only a limited subset of courses in a given semester \citep[see][]{budish2012multi,diebold2017matching}. We explore how the performance of mechanisms change for various sizes of choice sets in Appendix \ref{subsec:appendix-additional-simulations-comparative-statics}.

\subsection{Course Reserve Adjustment}
\label{subsec:course-reserve-adjustment}

We use the method of simulated moments to calibrate the student utility model parameters. Some adjustments to the course reserves need to be performed first. Course reserves are typically set large at the beginning of the actual course allocation process and relaxed at later stages. Hence, the final course allocation can violate the initial course reserves. In fact, $150$ out of $756$ courses violate initial course reservations in the data.

We next describe the necessary and sufficient conditions for the existence of a matching that assigns a given population of students to a course and satisfies several reservation criteria. We begin by outlining the conditions for two reservation criteria and then extend these conditions to the general case. Afterward, we describe the process of adjusting course reserves to ensure that the number of reserved seats is consistent with the course seat assignment observed in the data.  

\vspace{2mm}
\noindent {\bf Two course reserve criteria.} Consider a course with a capacity of $q$ seats. Denote the number of course reserves for students who satisfy criterion $A_i$ (student's department and year of study) as $\mathfrak{r}_{A_i}$, $i=1,2$. We assume that $\mathfrak{r}_{A_1}\geq 0$ or $\mathfrak{r}_{A_2}\geq 0$ and $\mathfrak{r}_{A_1}+\mathfrak{r}_{A_2}\leq q$. One should interpret the reservation such that if a student satisfies both $A_1$ and $A_2$, she is eligible to take a seat satisfying either criterion, but if she takes $A_1$-type seats it does not decrease the number of $A_2$-type seats. This interpretation is consistent with the actual course enrollment process in the data. The number of students who satisfy criterion $A_i$ and who are matched to the course in the final allocation is denoted by $x_{A_i}$, $i=1,2$. Students who do not satisfy any course reserve criteria are called $0$-students, and we denote the number of such students assigned to the course as $x_0$. Then, the assignment satisfies course reserves if and only if

\vspace{-4mm}
\begin{equation}
	\label{eq:two-criteria}
	\left\{
	\begin{array}{l}
		x_0 + x_{A_1\cup A_2} \leq q\\[-1mm]
		x_0\leq q-\mathfrak{r}_{A_1}-\mathfrak{r}_{A_2}\\[-1mm]
		x_0+x_{A_1\backslash A_2}\leq q- \mathfrak{r}_{A_2} \\[-1mm]
		x_0+x_{A_2\backslash A_1}\leq q - \mathfrak{r}_{A_1}
	\end{array}\right..	
\end{equation}
To see that the above conditions are necessary, we consider the criteria one by one. The first inequality corresponds to the criterion that puts no restrictions on students, the second inequality corresponds to students who do not satisfy $A_1$ and $A_2$, the third corresponds to students who do not satisfy $A_2$, and the fourth corresponds to students who do not satisfy $A_1$. To show sufficiency, we first enroll students $x_0$, $x_{A_1\backslash A_2}$, and $x_{A_2\backslash A_1}$. It is clearly possible to do so satisfying all criteria. The remaining students $x_{A_1\cap A_2}$ are constrained by the total course capacity. The first inequality guarantees that they could be enrolled as well.

\vspace{2mm}
\noindent {\bf The general case.} The number of constraints increases exponentially with the number of separate seat reservations a course has. However, there are 4 or fewer seat reservations per course in the university data, which makes the analysis computationally manageable. Below, we present the necessary and sufficient conditions that guarantee the existence of a matching that satisfies multiple course reserve criteria that are based on the seminal \textit{Hall's Marriage Theorem} \citep[see][]{philip1935representatives}.

To state the theorem, let us consider a set of students $S$ and course reserves criteria $A_1,...,A_I$ for some $I\geq 1$. The number of reserved seats corresponding to criterion $A_i$ equals $\mathfrak{r}_{A_i}\geq 0$. We denote the set of available course seats as $Y$ with $|Y|=q$ and assume $\sum_{i=1}^I\mathfrak{r}_{A_i}\leq q$. Each student could potentially satisfy multiple criteria. We consider a bipartite graph $(S,Y,E)$ where $E$ is the set of edges connecting students and seats. A student $s\in {\cal S}$ is connected to seat $y\in Y$ with an edge if student $s$ is eligible to take seat $y$. For example, if there are no course reserves, each student is connected with an edge to each seat ($q$ edges). If there is one course reservation criterion $A$, each $0$-student has $q-\mathfrak{r}_A$ edges, and each student who satisfies criterion $A$ has $q$ edges. 

For a subset of students $S'\subseteq S$, let $N(S')$ denote the set of seats in $Y$ that are connected to at least one student in $S'$. We call these seats the neighbors of $S'$. A perfect matching is a matching between students and seats in which each student is assigned exactly one seat. 
\begin{theorem*}[Hall's Marriage Theorem]
	The necessary and sufficient condition for the existence of a perfect matching is that every subset $S'\subseteq S$ is connected to as many neighbors as the number of its elements
	\begin{equation}
		\label{eq:Hall-inequality}
		|S'|\leq |N(S')|.
	\end{equation}
\end{theorem*}
It is not hard to see that the conditions of Hall's Marriage Theorem transform to the system of inequalities \eqref{eq:two-criteria} when there are only two criteria. To obtain the necessary and sufficient constraints in the general case, we need to consider only those inequalities where for a fixed set of seats $N(S')$ there is the maximum possible $|S'|$. 
\begin{itemize}
	\item The number of students who satisfy no criteria, that is, $0$-students
	
	\vspace{-5mm}
	$$
	x_0\leq q-\sum_{i=1}^I\mathfrak{r}_{A_i}.
	$$
	\item The number of students who satisfy \textit{only} criterion $A_i$ or no criteria (adding $0$-students only increases the number on the left-hand side of inequality \eqref{eq:Hall-inequality} and does not change the number on the right-hand side of the inequality)
	
	\vspace{-5mm}
	$$
	x_{A_i\backslash \cup_{j\neq i} A_j}+x_0\leq q-\sum_{j\neq i} \mathfrak{r}_{A_j}
	$$
	\item Let us now denote ${\cal I}=\{1,...,I\}$ and take some subset of available course reserve criteria ${\cal I'}\subseteq {\cal I}$. We consider the number of students who satisfy at least one criterion in the set $\{A_i:i\in {\cal I'}\}$ or no criteria and who do not satisfy the criteria in the complementary set $\{A_i:i\in {\cal I}\backslash {\cal I'}\}$.  Hall's Marriage Theorem implies that the number of these students must satisfy
	
	\vspace{-5mm}
	\begin{equation}
		\label{eq:Hall-inequality-x-q-r}
		x_{\cup_{i\in {\cal I'}}A_i\backslash \cup_{j\notin {\cal I'}} A_j}+x_0\leq q- \sum_{j\notin {\cal I'}} \mathfrak{r}_{A_j}.	
	\end{equation}	
\end{itemize}
Overall, Hall's Marriage Theorem implies that the set of inequalities \eqref{eq:Hall-inequality-x-q-r} for all possible subsets of criteria ${\cal I'}\subseteq {\cal I}$ provides the necessary and sufficient conditions for the existence of a matching of students to course seats that satisfies the course reserve criteria.

\vspace{4mm}
\noindent {\bf Course reserve adjustment.} Using the above conditions, we have identified 150 out of 756 courses that violate the initial course reserve criteria in the university data. For these courses, we order course seat reservations by year of study in decreasing order. We then take the first course reserve and decrease the number of its reserves by one. If the Hall's inequalities are still violated, we move to the next one and also decrease the number of its reserves by one. If decreasing all course reserves by one does not make matching possible, we repeat this procedure. Clearly, the procedure terminates in a finite number of steps, as a matching with zero course reserves is always possible.\footnote{For over-enrolled courses, we first set $q$ equal to the maximum between course capacity and the number of students assigned in the university data.} 

\subsection{Moment Conditions}
\label{subsec:appendix-moment-conditions}

For each set of parameter values and realization of random components of the student utility model described in Section \ref{subsec:appendix-Model}, we use the Random Serial Dictatorship with \textit{adjusted course reserves} to simulate the allocations of courses to students. Then, we use these allocations to calculate the simulated moments specified below.

For each horizontal component of students' preferences $\theta_{ay},a\in {\cal A}, y\in {\cal Y}$, we consider the number of courses taken by students with the same characteristics as a moment, i.e.,

\vspace{-2mm}
$$
\breve{m}_{ay}=\sum_{s\in {\cal S}_{ay}}\sum_{c\in {\cal C}} x_{s,c},
$$
where $x_{s,c}\in \{0,1\}$ indicates student $s$'s enrollment into course $c$ and $S_{ay}$ is the set of students from college $a$ and year $y$. 

Similarly, the number of seats taken by students from college $a$ enrolling in courses in college $a'$ is a natural moment for the interaction term $\gamma_{aa'},a,a'\in {\cal A}$; that is, 

\vspace{-2mm}
$$
\hat{m}_{aa'}=\sum_{s\in {\cal S}_{a}}\sum_{c\in {\cal C}_{a'}} x_{s,c},
$$
where $C_{a'}$ is the set of courses in college $a'$ and $S_{a}$ is the set of students from college $a$. For each college $a$, the number of its course seats by its own college $a$ students can be written as $\hat{m}_{aa}= \sum_y \breve{m}_{ay} - \sum_{a'\neq a} \hat{m}_{aa'}$. To remove this redundancy, we drop the 7 own-college moments $\hat{m}_{aa}$, which correspond to the normalized parameters $\gamma_{aa} = 0$. The remaining $42$ cross-college moments correspond to the $42$ interaction terms $\gamma_{aa'}$, $a \neq a'$.\label{linear-dependence-label}

For the vertical component of course-specific qualities $z_c,c\in {\cal C}$, it is natural to consider the course enrollment as a moment, i.e.,

\vspace{-2mm}
$$
\tilde{m}_c=\sum_{s\in {\cal S}} x_{s,c}.
$$

\noindent{\bf Courses at maximum capacity.} A critical aspect of the data is that $85$ out of $756$ courses are at \textit{maximum capacity} in the data, which creates a problem with identifying course-specific parameters $z_c$ for such courses. To address this issue, we replace moments $\tilde{m}_c$ for such courses with moments based on the timing of course enrollment. We first define a ranking of students.

\vspace{-1mm}
\begin{definition}
	\label{definition:student ranking}
	The rank of student $s$, $\text{rank}_s$, is her position in the order of students to enroll in the Random Serial Dictatorship with course reserves after breaking ties.
\end{definition}

\vspace{-1mm}
\noindent Hence, more popular courses are filled up earlier with students of a smaller rank. Less popular courses are filled up later with students of a larger rank.\footnote{More senior students have earlier time slots in the registration process. Hence, measuring popularity by how fast a course fills up is biased towards the popularity of courses among more senior students.}

To define student ranking in the university data, we considered the earliest timestamp of each student in our dataset; that is, the time students registered for their first course. Using these timestamps, we rank students in the university data.\footnote{An alternative measure of how fast courses fill up is course registration times. Both student ranking and course registration time are closely related in our simulations, as each student registers for all desired courses simultaneously. This connection is less evident in the data. Many students do not register for courses at the same time. This choice pattern is probably connected with students changing their minds or new possibilities arising during the actual course registration process. We found that student ranking better connects our simulations to the actual university data than course registration time.} To measure course popularity, we use the mean student rank enrolled in the course:

\vspace{-4mm}

$$
Mean(\{\text{rank}_s:x_{s,c}=1,\,s\in {\cal S}\}).
$$
For courses at maximum capacity, we replace moments $\tilde{m}_c$ with the mean rank across the registered students.\footnote{We also tried to use the largest student rank and some quantiles of students registered for the course to measure course popularity. These measures did not work very well. Last-minute course registrations are a noisy measure that do not reflect course popularity well. Measures based on quantiles proved to be discontinuous, with a small change in utility parameters often resulting in jumps in the student ranks at a specific quantile and a bad model fit.} We keep using the symbol $\tilde{m}_c$ for all moments corresponding to vertical components $z$ to avoid clutter in notation. Overall, there are $756+28+42=826$ independent moments $m=(\tilde{m},\breve{m},\hat{m})$ and $819$ independent parameters $\pi=(\theta,\Gamma,z)$ to calibrate.\footnote{If every course used its enrollment moment, there would be seven additional identities, $\sum_{c \in {\cal C}_{a}} \tilde{m}_c = \sum_{a' \in {\cal A}} \hat{m}_{a'a}$ for each $a \in {\cal A}$, which would reduce the number of independent moments to $819$ and make the model exactly identified. As each college offers at least one course at maximum capacity, none of these identities holds.} Hence, the model is over-identified, allowing us to test for model fit in Section \ref{subsec:appendix-inference}.

\subsection{Calibration}
\label{subsec:appendix-calibration}

We use the Random Serial Dictatorship with adjusted course reserves to calculate simulated moments and find $\pi^*=(\theta^*,\Gamma^*,z^*)$ that minimize a criterion function:

\begin{equation}
	\label{eq:objective}
	||m-m^s(\pi)||_W^2=(m-m^s(\pi))'W(m-m^s(\pi)),
\end{equation}

\vspace{2mm}
\noindent where $m$ is a set of moments constructed from the empirical data, $m^s(\pi)$ is the average of moments constructed from $100$ simulation rounds for a given vector $\pi=(\theta,\Gamma,z)$ (drawing random student utility components and choice sets for each round), and $W$ is a weight matrix. We use the following procedure \citep[see][]{jalali2015using}:
\begin{itemize}
	
	\smallskip
	\item {\bf First Stage.} We calibrate the optimal parameters $\pi$ using an initial weight matrix $W_0$. $W_0$ is a diagonal matrix with diagonal element $d$ equal to $1/m^2_d$, where $m_d$ is the value of the moment constructed from the empirical data. Hence, the first stage minimizes the sum of squared percentage errors for each moment.
	
	Using estimated $\pi$, we calculate the variance-covariance matrix of simulated moments
	
	\vspace{-6mm}
	{\small
		\begin{equation}
			\hat{S}=\frac{1}{L_1}\sum^{L_1}_{\ell_1=1} \left[m^{\ell_1}(\pi,\varepsilon^{\ell_1})-\frac{1}{L_2}\sum^{L_2}_{\ell_2=1}m^{\ell_2}(\pi,\varepsilon^{\ell_2})\right] \left[m^{\ell_1}(\pi,\varepsilon^{\ell_1})-\frac{1}{L_2}\sum^{L_2}_{\ell_2=1}m^{\ell_2}(\pi,\varepsilon^{\ell_2})\right]', 
	\end{equation}}
	
	\noindent where $\varepsilon^{\ell_1}$ and $\varepsilon^{\ell_2}$ are draws of random component for $\ell_1=1,...,L_1$ and $\ell_2=1,...,L_2$. We take $L_1=8000$ and $L_2=1000$. The optimal weight matrix is the inverse of the variance-covariance matrix $W^*=\hat{S}^{-1}.$ We take a large number of draws, $L_1 = 8000$, to reduce the upward bias of $\hat{S}^{-1}$, which exceeds in expectation the inverse of the true variance-covariance matrix of the simulated moments by the factor $L_1/(L_1-|m|-1)$, where $|m|$ is the number of moments \citep[see Lemma 7.7.1 in][]{Anderson2003}.
	
	\smallskip
	\item {\bf Second Stage.} We start with the optimal weight matrix $W^*_\ell=W^*$ for $\ell=1$. We then use the optimal weight matrix to obtain more precise estimates for parameters $\pi^*_\ell$. Using the more precise estimate $\pi^*_\ell$, we compute the optimal weight matrix $W^*_{\ell+1}$ and repeat the process $\ell=2,3,...$ until the maximum change of calibrated parameters $|\pi^*_{l+1}-\pi^*_l|$ becomes smaller than $0.01$ to ensure convergence.
\end{itemize}

\vspace{2mm}
\noindent {\bf Calibration results.} Table \ref{tab:appendix-calibrated-parameters} provides the calibrated parameters of the student utility model. The top portion of the table displays the \textit{horizontal components} $\theta$, with rows corresponding to colleges and columns corresponding to year of study. The coefficient $\theta_{ay}$ represents the overall importance of taking courses for students in college $a$ and year $y$. Our normalization ensures that the sum of the coefficients in each row, weighted by the share of the student population across years, equals zero.

The horizontal components are positive for students in their first and second years and mostly negative in their third and fourth years. This observation suggests that students prefer enrolling in more courses during their initial years of study. The tendency to enroll in fewer than five courses in  later years may be due to several factors. Such students may choose part-term courses that are not in our data or have already accumulated enough credits for their degree, allowing them to ease their course load in the final years.

The middle portion of the table presents \textit{interaction terms} of matrix $\Gamma$. The general trend of negative utility changes emphasizes a preference for internal college courses, but specific positive changes indicate potential areas for beneficial inter-college course offerings. For example, students from colleges $C$ and $D$ value highly courses in college $A$. 

The bottom portion of the table presents quantiles of $756$ coefficients of \textit{vertical component} $z$, with each coefficient $z_c$ being responsible for the average popularity of the course among students from the same college. The values range from a minimum of $-2.29$ to a maximum of $2.41$. This wide range indicates significant variation in the popularity of courses within colleges. Most of the quantile values are negative, up to the $90$th percentile, indicating that a majority of the courses have negative values. This result is most likely manifested by the student utility model that relies on college-year parameters to identify student preferences. Students randomly draw $80$ courses in their choice set but take no more than $5$ courses. Our model, which relies on college-level and year-level parameters, matches this empirical pattern by putting the values of the vertical components of most courses in the negative region.

\subsection{Model Fit}
\label{subsec:appendix-inference}

We test for model fit using the calibrated parameters. Our utility model is over-identified, as we have $826$ moments and $819$ independent parameters. The $J$-test statistic is given by

\vspace{-4mm}
$$
J=\frac{K}{K+1}(m-m^s(\pi^*))'W^*(m-m^s(\pi^*))\sim\chi^2_{7},
$$
where $K=100$ is the number of random component draws and $W^*$ is the inverse of the variance-covariance matrix of moments estimated at the terminal step of the calibration process. This statistic has a chi-square distribution with $7$ degrees of freedom under the null hypothesis that the model fits the data. The $J$-test statistic $J=15.57$ is less than $18.48$, the $99$th percentile of the chi-square distribution with $7$ degrees of freedom. Hence, the $J$-test does not reject the null hypothesis that the model fits the data at the $99\%$ critical value.

\renewcommand{\arraystretch}{0.95} 

\begin{table}[H]
	\begin{center}
		\refstepcounter{tab}\label{tab:appendix-calibrated-parameters}
		\caption{The calibrated parameters of the student utility model.}
		
		\vspace{1mm}		
		\begin{subtable}{\linewidth}
			\centering
			\caption*{Coefficients $\theta$}%
			\begin{tabular}{C{\cwidth} |C{\cwidth}|C{\cwidth}|C{\cwidth}|C{\cwidth}|C{\cwidth}}
				\multicolumn{1}{c}{} & \multicolumn{1}{c}{Year 1} & \multicolumn{1}{c}{Year 2} & \multicolumn{1}{c}{Year 3} & \multicolumn{1}{c}{Year 4} \\ 
		A & 0.13 & 0.20 & -0.04 & -0.28 \\ 
        		B & 0.14 & 0.20 & 0.00 & -0.31 \\ 
        		C & 0.28 & 0.21 & 0.01 & -0.44 \\ 
        		D & 0.09 & 0.19 & -0.02 & -0.33 \\ 
        		E & 0.20 & 0.15 & -0.13 & -0.29 \\ 
        		F & 0.18 & 0.08 & -0.08 & -0.29 \\ 
        		G & 0.19 & 0.11 & 0.01 & -0.37 \\ 
			\end{tabular}
		\end{subtable}
		
		\vspace{5mm}
		\begin{subtable}{\linewidth}
			\centering
			\caption*{Coefficients $\gamma$}%
			\begin{tabular}{C{\cwidth} |C{\cwidth}|C{\cwidth}|C{\cwidth}|C{\cwidth}|C{\cwidth}|C{\cwidth}|C{\cwidth}|}
				\multicolumn{1}{c}{}& \multicolumn{1}{c}{A} & \multicolumn{1}{c}{B} & \multicolumn{1}{c}{C} & \multicolumn{1}{c}{D} & \multicolumn{1}{c}{E} & \multicolumn{1}{c}{F} & \multicolumn{1}{c}{G} \\ 
				A & 0.00 & -0.63 & -0.51 & -0.29 & -0.56 & -0.70 & -0.50 \\  
        				B & -0.13 & 0.00 & -0.49 & -0.26 & -0.10 & -0.47 & -0.48 \\  
				C & 0.36 & -0.21 & 0.00 & -0.01 & 0.01 & -0.27 & -0.26 \\  
				D & 0.14 & -0.14 & -0.31 & 0.00 & -0.16 & -0.32 & -0.27 \\  
				E & -0.01 & -0.55 & -0.28 & -0.16 & 0.00 & -0.33 & -0.40 \\  
				F & -0.05 & -0.65 & -0.52 & -0.21 & -0.17 & 0.00 & -0.56 \\  
				G & -0.18 & -0.56 & -0.53 & 0.04 & -0.22 & -0.49 & 0.00 \\ 
			\end{tabular}
		\end{subtable}
		
		\vspace{5mm}
		\begin{subtable}{\linewidth}
			\centering
			\caption*{The Quantiles of Coefficients $z$}%
			\begin{tabular}{cccccccc}
				\hline
				\hline\\[-2ex]
				Quantile& Min & 0.10 & 0.25 & 0.50 & 0.75 & 0.90 & Max\\[1ex]
				\hline\\[-2ex]
				 $z_c$ & -2.29 & -1.88 & -1.70 & -1.47 & -1.14 & -0.68 & 2.41
				 \\[1ex]
				\hline
				\hline
			\end{tabular}			
		\end{subtable}
	\end{center}	
	
	\vspace{2mm}
	\begin{spacing}{1}
		{\footnotesize \textit{Notes:} This table provides information about the calibrated parameters of the student utility model. The upper portion of this table presents the \textit{horizontal components} $\theta$ of students' utility, with the rows corresponding to the student's college and the columns corresponding to the student's year of study. Our normalization ensures that the sum of the parameters in each row weighted with the share of student population across years equals zero. The middle portion presents \textit{interaction terms} of matrix $\Gamma$. Each non-zero interaction coefficient $\gamma_{aa'},a,a'\in {\cal A}$ corresponds to the increase in student's utility from taking a course in college $a'$ (column) compared to a course in student's college $a$ (row) with the same vertical component. Note that we normalize $\gamma_{aa}=0,a\in {\cal A}$. The bottom portion presents some quantiles of \textit{vertical components} $z$. The interpretation of $z_c$ is the average popularity of the course among students within the course's college.}
	\end{spacing}	
\end{table}

\renewcommand{\theequation}{D\arabic{equation}} %
\setcounter{equation}{0}

\renewcommand{\thelemma}{D\arabic{lemma}}
\setcounter{lemma}{0}

\renewcommand{\thetab}{D.\arabic{tab}}
\setcounter{tab}{0}

\renewcommand{\thedefinition}{D\arabic{definition}}
\setcounter{definition}{0}

\renewcommand{\theproposition}{D\arabic{proposition}} %
\setcounter{proposition}{0}

\renewcommand{\thefig}{D\arabic{fig}} %
\setcounter{fig}{0}

\renewcommand{\thetheorem}{D\arabic{theorem}} %
\setcounter{theorem}{0}

\renewcommand{\theexample}{D\arabic{example}} %
\setcounter{example}{0}

\counterwithin{figure}{section}

\newpage
\section{Additional Simulations}
\label{sec:appendix-additional-simulations-results}

\counterwithin{figure}{section}
This appendix provides additional simulation results. First, for the simulations conducted in the main text, we present additional results on the equilibrium outcomes of the Pseudo-Market with Priorities (PMP) mechanism (Section \ref{subsec:appendix-pmp-prices}) and compare performance with some alternative mechanisms (Section \ref{subsec:appendix-additional-sumulations-other-mechanisms}). Second, we supplement the main text with simulation results with a department-first priority structure (Section \ref{subsec:appendix-additional-sumulations-dept-first}). Last, we report robustness checks by considering several sizes for student choice sets and the standard deviation of the noise utility parameter (Section \ref{subsec:appendix-additional-simulations-comparative-statics}).

\subsection{Outcomes of the Pseudo-Market with Priorities Mechanism}
\label{subsec:appendix-pmp-prices}

First, we present statistics on the equilibrium prices in the PMP mechanism for the simulations reported in Section \ref{subsec:simulation-results}. The top part of Table \ref{tab:pmp-prices} provides information about where the cutoff levels of priority are across courses and how high the prices are for students at the cutoff level of priority. The bottom part provides information about the percentage of students within each year of study allocated different numbers of courses with positive prices.

Around $70.8\%$ of courses have their cutoff priority at level $1$. The large share of courses at the lowest priority level reflects the non-binding enrollment of most courses in the university data, where just $11.2\%$ of courses are at or above their maximum enrollment capacity (see Table \ref{tab:binding-capacity} on p. \pageref{tab:binding-capacity}). Above level $1$, there is a noticeable difference in the percentage of courses with odd and even cutoff levels of priority. Substantially more courses have odd cutoff levels ($7.5\%$, $8.2\%$, $5.7\%$ for levels $3$, $5$, and $7$, respectively) than even ones ($2.3\%$, $3.0\%$, $1.8\%$, $0.6\%$ for levels $2$, $4$, $6$, and $8$). The hybrid priority structure we utilize places students at odd priority levels in a course if they are ineligible for any of its reserved seats. The students eligible for a course's reserved seats typically only comprise a small subset of the interested students. Furthermore, some courses do not reserve any seats. As a result, the cutoff level of priority often ends up at an odd level rather than an even one.

\begin{table}[h!]
	\begin{center}
		\footnotesize
		
		\vspace{-1mm}
		 \refstepcounter{tab}\label{tab:pmp-prices}
		\caption{Price statistics for the Pseudo-Market with Priorities mechanism.}		
		\begin{tabular}{cccccccccc}
			\hline
			\hline\\[-2mm]
			\multicolumn{9}{l}{\bf Cutoff levels of priority and their average prices} \\[2mm]
			\hline\\[-2mm]
			\makecell[lc]{Cutoff level of priority}& $1$ & $2$ & $3$ & $4$ & $5$ & $6$ & $7$ & $8$\\[1mm]
			\hline\\[-1ex]
			\makecell[lc]{\% of courses}& 70.8\% & 2.3\% & 7.5\% & 3.0\% & 8.2\% & 1.8\% & 5.7\% & 0.6\%\\[1mm]
			\makecell[lc]{Average cutoff price}& $0.06$ & $0.79$ & $0.84$ & $0.80$ & $0.89$ & $0.79$ & $0.95$ & $0.86$\\[2mm]

			\hline\\
			\multicolumn{9}{l}{\bf Courses with positive prices taken by students}\\[1ex]
			\hline\\[-2mm]
			\# of courses paid for &&& Year 1 & Year 2 & Year 3 & Year 4\\[2mm]
			\hline\\[-2mm]
			0  &&& 41.8\% & 37.8\% & 44.8\% & 60.3\% \\[1mm]
			1  &&& 46.0\% & 50.6\%  & 45.2\% & 36.0\% \\[1mm]
			2  &&& 11.0\% & 10.5\% & 9.1\% & 3.5\% \\[1mm]
			3  &&& 1.2\% & 1.0\% & 0.8\% & 0.2\% \\[1mm]
			4  &&& 0.06\% & 0.07\% & 0.04\% & 0.01\% \\[1mm]
			5  &&& 0.001\% & 0.002\% & 0.0\% & 0.0\% \\[2mm]
			\hline
			\hline			
		\end{tabular}
	\end{center}

\vspace{0mm}
{\footnotesize \textit{Notes:} This table presents statistics on the equilibrium prices in the Pseudo-Market with Priorities mechanism. Results are averages across $100$ runs with different random component draws.}	
\end{table}

Among courses with a cutoff priority level of $1$, the average price at the cutoff is just $0.06$. As many courses are free for all students, the average cutoff price for these courses is considerably lower than for courses with a higher cutoff level. For higher cutoff levels, the average cutoff prices lie between $0.79$ and $0.95$, with prices generally higher in odd cutoff levels. Student budgets are assigned in $[1, 1.25]$, meaning that, at the average cutoff price, each student can afford one but not two courses for which they are at the cutoff level of priority. Still, students manage to take $k = 5$ courses by taking courses for which they are above the cutoff level of priority, where they face a priority-specific price of $0$.

The bottom part of Table \ref{tab:pmp-prices} further explores the number of courses students pay a positive price for. Across all years of study, at least $37.8\%$ of students obtain all their course seats for free, and less than $12.5\%$ of students pay a positive price for more than one course. Many courses have a zero price for all students, leading some students to obtain their most preferred course schedule for free. In addition, students often cannot afford a second course with a positive price after obtaining one.

These findings differ across years of study. Less than $7\%$ of cutoff levels of priority are $7$ or $8$, implying that Year $4$ students face zero price for the vast majority of courses. Over $60\%$ of these students pay for none of their assigned courses, with about $4\%$ obtaining multiple course seats with positive prices. In earlier years of study, a larger percentage of students pay for course seats. Under $45\%$ of first- through third-year students obtain all their course seats for free, with more than $9.7\%$ needing to pay for multiple courses in each year.

\begin{table}[t!]
    \begin{center}
        \footnotesize
        
        \vspace{-1mm}
        \refstepcounter{tab}\label{tab:market-clearing}
        \caption{Over-enrolled courses in the Pseudo-Market with Priorities mechanism.}
        \begin{tabular}{p{0.35\textwidth}>{\centering}m{0.07\textwidth}>{\centering}m{0.07\textwidth}>{\centering}m{0.07\textwidth}>{\centering}m{0.07\textwidth}>{\centering\arraybackslash}m{0.07\textwidth}}
            \hline
            \hline\\[-1ex]
            \makecell[lc]{\# seats above max capacity}& $\geq 1$ & $\geq 2$ & $\geq 3$ & $\geq 4$ & $\geq 5$\\[2ex]
             \hline\\[-1ex]
             \makecell[lc]{\% of courses}& $3.03\%$ & $0.50\%$ & $0.10\%$ & $0.02\%$ & $0.0\%$ \\[2ex]
             \hline
             \hline
         \end{tabular}
     \end{center}
     
     \vspace{0mm}
     {\footnotesize \textit{Notes:} This table provides information about the distribution of over-enrolled courses in the Pseudo-Market with Priorities mechanism. Results are averages across 100 runs with different random component draws.}
\end{table}

We also present results on course over-enrollment in the PMP mechanism's outcomes. The average market-clearing error in our simulations is $21.19$, less than half of the theoretical worst-case bound of $\alpha \approx 43.5$. Table \ref{tab:market-clearing} reports that $3.03\%$, or roughly 23 courses, are over-enrolled by at least one seat. In addition, $0.50\%$ (4 courses) are over-enrolled by at least two seats, and $0.10\%$ ($<$1 course) are by at least three. No courses are over-enrolled by more than four seats. In contrast, the university's original allocation has $7.3\%$ of courses over-enrolled, with $1.9\%$ over-enrolled by at least five seats (see Table \ref{tab:binding-capacity}). A lower market-clearing error and over-enrollment are possible by allowing a longer runtime for our algorithm.

\vspace{0mm}
\subsection{Alternative Course Allocation Mechanisms}
\label{subsec:appendix-additional-sumulations-other-mechanisms}

Our main simulation results compare the Pseudo-Market with Priorities mechanism and two variants of the Deferred Acceptance mechanism, each of which \textit{respects the course priority structure}. That is, the outcomes of these mechanisms contain no student who wants to enroll in a course that assigns a seat to a student of strictly lower priority.

In turn, the Random Serial Dictatorship with course reserves (RSD) benchmark does not respect any priority structure. The cause of such violations is the presence of course reserves that can be interpreted as \textit{seat-specific priorities}. Each of a course's reserved seats has a priority order over students that places reserve-eligible students above others in the same year of study. Unreserved seats only prioritize students by year of study. Other mechanisms in the literature have seat-specific priorities. For example, the Deferred Acceptance with minority reserves mechanism introduced in \cite{hafalir2013effective} is a modification of the Deferred Acceptance mechanism with seat-specific priorities.

The main objective of this section is to compare the performance of the PMP mechanism with six alternative mechanisms. Four of these mechanisms also have \textit{seat-specific priorities} (DA-STB and DA-MTB with  reserves, RSD with base course reserves, Probabilistic Serial with seniority and reserves), one \textit{does not respect any priorities} (A-CEEI), and one \textit{respects priorities in a kludgy way} (A-CEEI with kludgy pricing). 

\noindent {\bf Deferred Acceptance (DA) with Reserves with single and multiple tie-breakings.} The Deferred Acceptance with minority reserves mechanism was proposed by \cite{hafalir2013effective} to address affirmative action policies in school choice. This mechanism is similar to the Deferred Acceptance mechanism, but it allows each school to prioritize a designated set of students in a fixed number of ``reserved'' seats. \cite{westkamp2013analysis} generalizes this mechanism to allow schools to have heterogeneous sets of students eligible for reserves and multiple, overlapping reserve groups. We extend this mechanism to a many-to-many matching setting with indifferences, using the optimal course reserves discussed in Section \ref{subsec:mechanisms}. We consider two cases: where indifferences are broken with a single tie-breaking rule and where indifferences are broken with multiple tie-breaking rules, one for each course.\footnote{We do not claim that the theoretical properties of this mechanism extend to many-to-many settings with indifferences. We leave the analysis of this important question for further research.}

\vspace*{2mm}
\noindent {\bf Random Serial Dictatorship (RSD) with Base Course Reserves.} Course reserves are typically set large at the beginning of the course allocation process at major US universities and relaxed at later stages. In the main text, motivated by this observation, we estimate the optimal reserved seats for each course in the RSD mechanism. Here, we run the same mechanism, using the university's base course reserves.

\vspace*{2mm}
\noindent {\bf Approximate Competitive Equilibrium from Equal Incomes (A-CEEI).} This mechanism was originally proposed by \cite{Budish2011} to allocate courses in business schools. It is a special case of the PMP mechanism in which all courses place all students at the same priority level, leading each course to have a single price. It shares many favorable theoretical properties with the PMP mechanism.

\vspace*{2mm}
\noindent {\bf Approximate Competitive Equilibrium from Equal Incomes with Kludgy Pricing (A-CEEI with Kludgy Pricing).} This variant of the A-CEEI mechanism aims to satisfy course priorities by using ``kludgy pricing''. As in the A-CEEI mechanism, this mechanism uses one price $p_c$ per course $c$. However, students receive a price discount depending on their priority. A student $s$ at level of priority $r_{s,c}$ in course $c$ faces price $p_c \cdot \left(1-(r_{s,c}-1)/R\right)$, where $R$ is the largest possible priority level. This pricing ensures that students always face a lower price than students at a lower level of priority but leaves open the possibility of a lower-priority student taking a seat when a higher-priority student cannot afford one.

\vspace*{2mm}
\noindent {\bf Probabilistic Serial (PS) with Seniority and Reserves.} The PS mechanism was proposed by \cite{bogomolnaia2001new} to randomly assign objects to agents in settings with single-unit demand for objects and single-unit supply of object types. \cite{kojima2009random} extended the mechanism to settings with multi-unit demand for objects, and \cite{budish2013designing} extended the mechanism to settings with multi-unit supply of objects. The undergraduate course allocation setting requires both multi-unit demand and multi-unit supply of courses. The only paper that considers both multi-unit demand and multi-unit supply of courses with preferences over individual courses is an unpublished paper by \cite{pycia2011}.\footnote{\cite{nguyen2016assignment} also consider multi-unit demand and multi-unit supply of objects but consider agent preferences over bundles of objects. They propose a version of a probabilistic serial mechanism called the Bundled Probabilistic Serial mechanism that has been successfully implemented at the Technical University of Munich by \cite{bichlerrandomized}. One of the versions of their probabilistic serial mechanism could also be applied to settings with agent preferences over individual objects.} The setting of undergraduate course allocation is further complicated by the presence of student seniorities and course reserves. We define an extension of the probabilistic serial mechanism called \textit{Probabilistic Serial with Seniority and Reserves}, using the optimal course reserves discussed in Section \ref{subsec:mechanisms}.

In this mechanism, each course seat is regarded as a divisible object of probability shares. We start with the most senior students (e.g., Year 4), each of whom ``eats'' with speed one the most preferred \textit{eligible} course seat that is available. Each student is eligible for unreserved seats and each reserved seat specifying her department and year of study. We assume that students first consume reserved course seats before consuming unreserved ones of a particular course. If some reserved/unreserved seats are exhausted, or some student has consumed a full probability share of a certain course's seat, we stop and redetermine the student's most preferred eligible course to eat next. This proceeds until time $k$, where the most senior students have consumed up to $k$ courses. We then repeat the procedure with the students of the next seniority level (e.g., Year 3, Year 2, and Year 1). The resulting profile of shares of course seats eaten by agents by time $4k$ corresponds to a random assignment.

\renewcommand{\arraystretch}{1.1} 
\newcommand{\cawidth}{0.108\textwidth}
\newcolumntype{C}[1]{>{\centering\arraybackslash}p{#1}} 

\begin{figure}[t!]
	\begin{center}
	
	\vspace{-1mm}
	\includegraphics[width=0.73\textwidth]{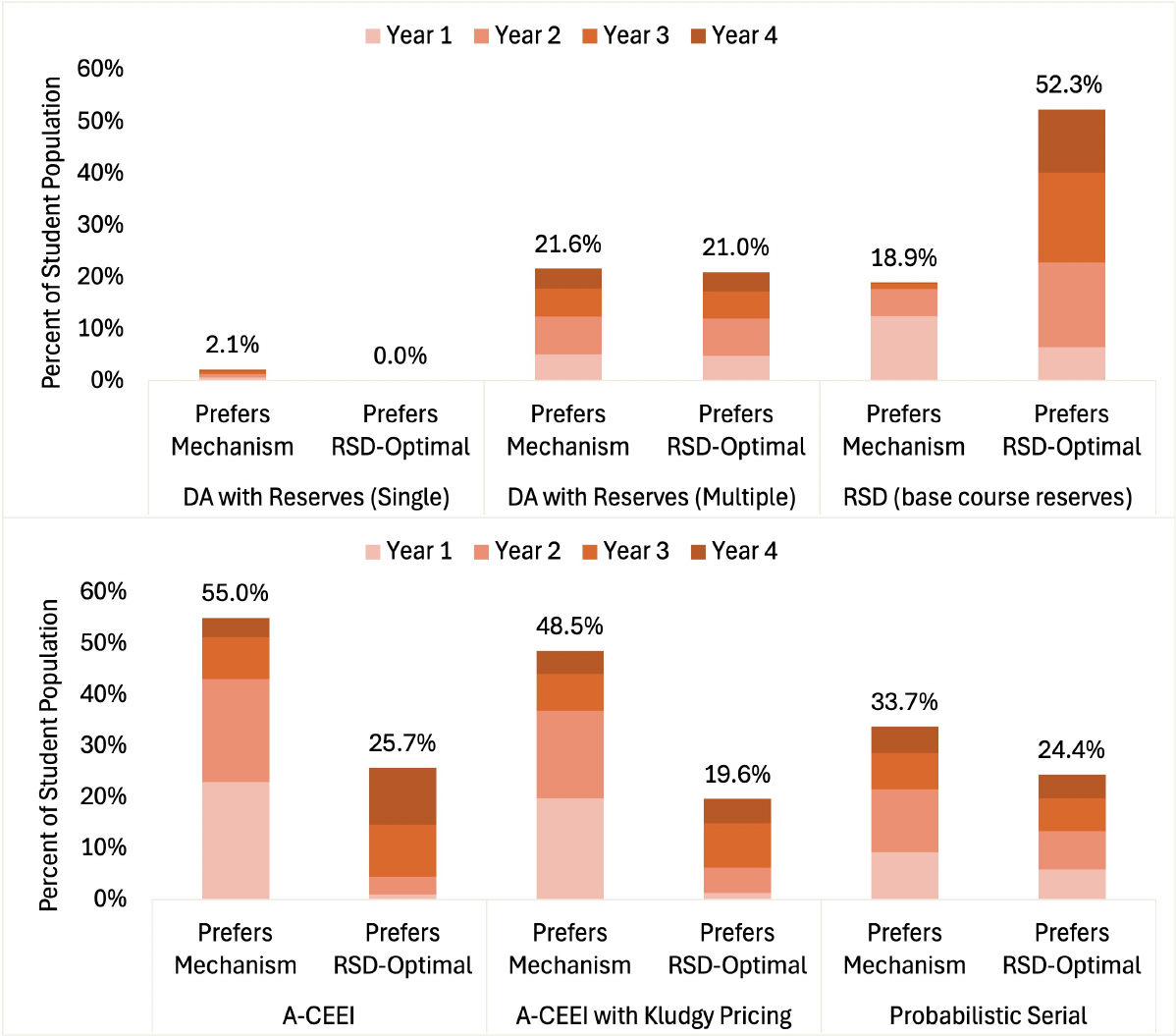}
	\caption{Each student's preferred mechanism for the alternative mechanisms.}
	\label{fig:preferred-mechanism-appendix-mechanisms}
	
	\end{center}
	
	\vspace{-2mm}
	{\footnotesize \textit{Notes:} This figure reports the percent of students who strictly prefer each alternative mechanism to the \textit{Random Serial Dictatorship with optimal course reserves} benchmark and the percent of students who strictly prefer the benchmark. Results are averages across 100 runs with different random component draws and they are based on \textit{the population of all students}.}

\end{figure}

\vspace{2mm}
{\bf Simulation results}. Our simulation results for the six alternative mechanisms are presented below. Figure \ref{fig:preferred-mechanism-appendix-mechanisms}, as in Figure \ref{fig:preferred-mechanism}, reports the percentage of students who prefer the allocation of a given mechanism to the benchmark of the Random Serial Dictatorship with optimal course reserves and the percentage of students who prefer the benchmark. Figure \ref{fig:utility-all-changing-students-appendix-mechanisms}, as in Figure \ref{fig:utility-all-changing-students}, reports the percent improvement in mean utility for students in different years of study for the six mechanisms. In the main text, Figures \ref{fig:preferred-mechanism} and \ref{fig:utility-all-changing-students} consider the limited universe of students that have at least two different schedules across the four mechanisms in a given run. The mechanisms studied here vary more widely, including a mechanism that does not account for student priorities and a random mechanism. As a result, Figures \ref{fig:preferred-mechanism-appendix-mechanisms} and \ref{fig:utility-all-changing-students-appendix-mechanisms} consider the full population of students. Hence, the numbers across the figures presented in the appendix and the figures presented in the main text cannot be directly compared.\footnote{To provide an accurate comparison with the results below, we note that, across all students, $18.3\%$ strictly prefer the PMP mechanism over the RSD benchmark and $12.9\%$ strictly prefer the RSD benchmark. Utility improvements are $0.3\%$, $2.2\%$, $2.8\%$, and $2.1\%$ for first, second, third, and fourth-year students, respectively.}

Figure \ref{fig:st-dev-appendix-mechanisms}, as in Figure \ref{fig:st-dev-PMP-DA-DAm}, presents the change in the standard deviation of student utility for each mechanism by year of study. Table \ref{tab:appendix-envy-other-mechamisms}, parallel to Table \ref{tab:student-envy}, reports students who experience schedule envy towards students of the same priority or lower priority. Both present results for the population of all students and can be directly compared to their counterparts in the main text. Table \ref{tab:appendix-priority-violations-other-mechamisms} presents the percentage of students who experience a \textit{priority violation}: who want to enroll in a course being taken by a student of strictly lower priority. The PMP, DA-STB, and DA-MTB mechanisms studied in the main text each have no priority violations.

\begin{figure}[t!]
	\begin{center}
	
	\vspace{-1mm}
	\includegraphics[width=0.73\textwidth]{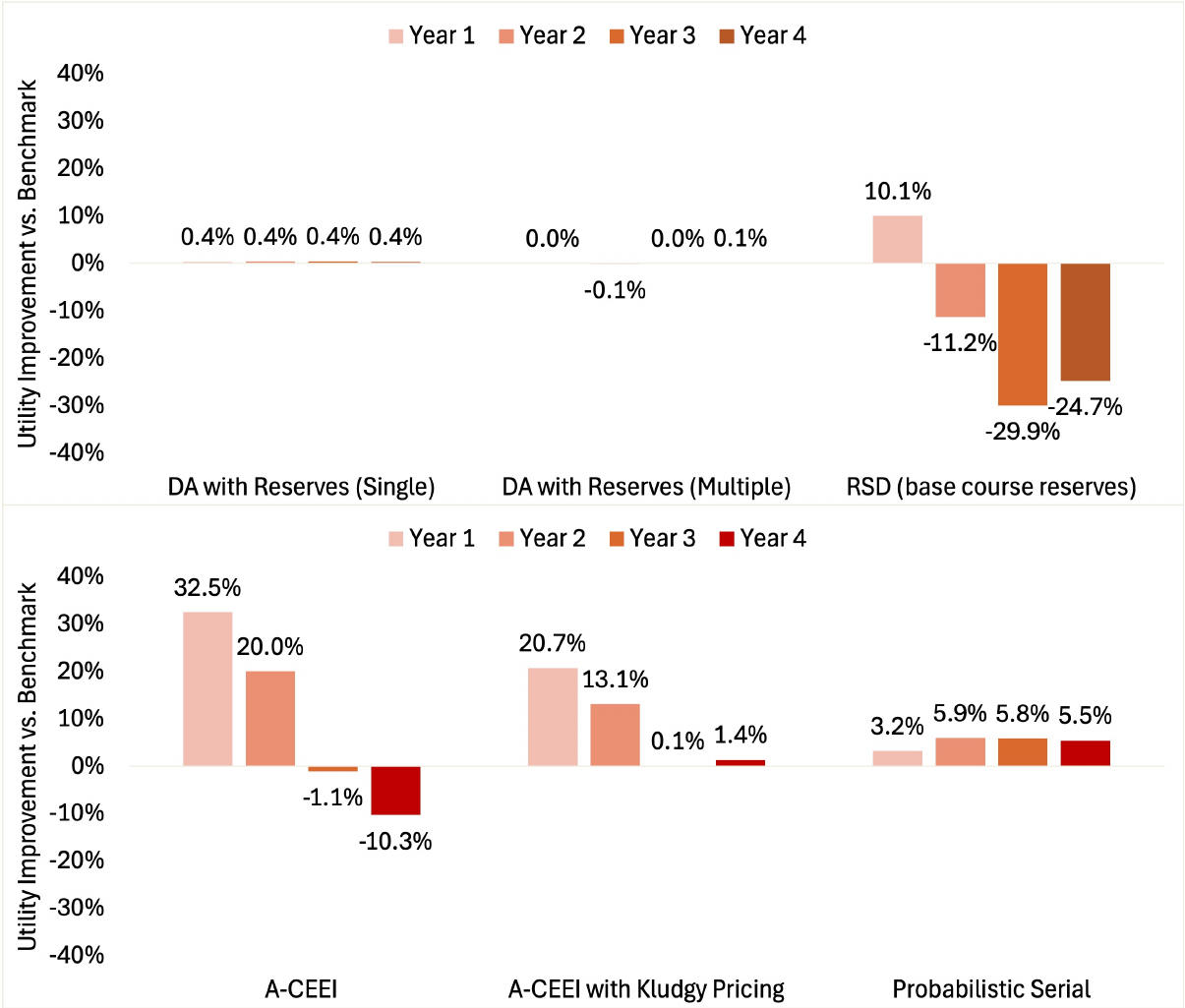}
	\caption{Improvements in utility for students in different years of study.}
	\label{fig:utility-all-changing-students-appendix-mechanisms}
	
	\end{center}
	
	\vspace{-2mm}
	 {\footnotesize \textit{Notes:} This figure reports the percent improvement in mean utility for each alternative mechanism over the \textit{Random Serial Dictatorship with optimal course reserves} benchmark. Results are averages across $100$ runs with different random component draws and they are based on \textit{the population of all students}.}        
	
\end{figure}

\vspace{2mm}
\label{DA with minority reserves comparison}
\textbf{The DA with reserves (with single and multiple tie-breakings)} differ significantly in their preferred mechanism results. Under single tie-breaking, the mechanism is a weak improvement over the benchmark: no student strictly prefers her outcome under the RSD mechanism.\footnote{This result holds for all possible environments. A proof is available upon request.} However, only $2.1\%$ of students strictly prefer the DA with reserves outcome. With about $98\%$ of students indifferent between the two mechanisms, DA with reserves has nearly identical utility, envy, and priority violation results to the benchmark. In contrast, less than $60\%$ of students are indifferent between the DA with reserves mechanism under multiple tie-breaking and the RSD mechanism, with a roughly even proportion of students strictly preferring each ($21.6\%$ vs. $21.0\%$). This mechanism has similar allocation fairness results to its counterpart without reserves. Still, a proportion of students ($17.6\%$) have priority violations, indicating that many students can benefit from a course seat assigned to a student of lower priority.

\begin{figure}[t!]
	\begin{center}
	
	\vspace{-1mm}
	\includegraphics[width=0.73\textwidth]{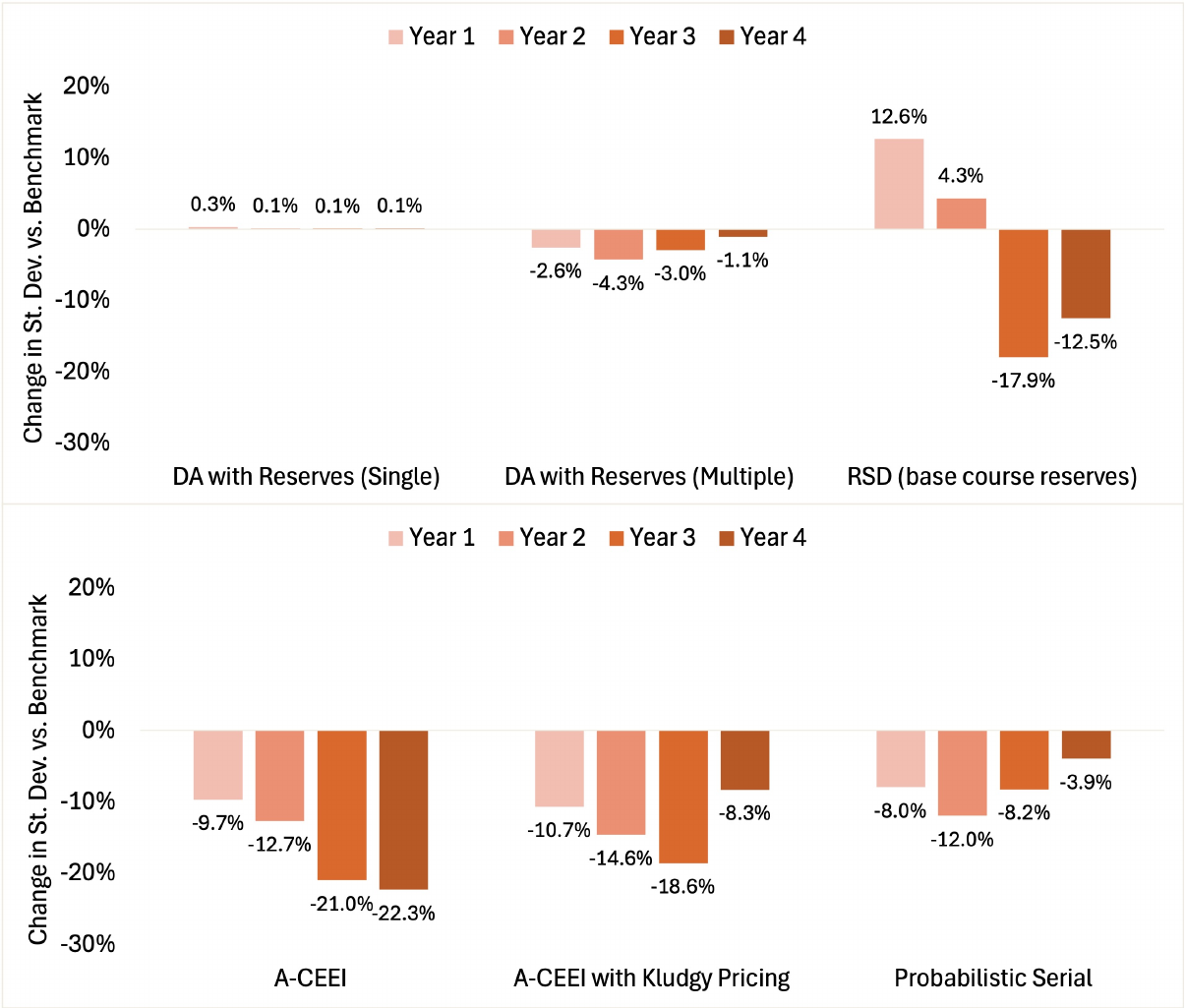}
	\caption{The change in the standard deviation of students' utilities for the alternative mechanisms.}
	\label{fig:st-dev-appendix-mechanisms}
	
	\end{center}
	
	\vspace{-2mm}
	{\footnotesize \textit{Notes:} This figure reports the percent change in the standard deviation of students' utilities for each alternative mechanism compared to the \textit{Random Serial Dictatorship with optimal course reserves} benchmark. Results are averages across $100$ runs with different random component draws and they are based on \textit{the population of all students}.}     
	
\end{figure}
\vspace{2mm}
{\bf The RSD with base course reserves} mechanism performs quite poorly compared to the benchmark. More than $50\%$ of students strictly prefer the RSD mechanism with optimal course reserves over the RSD mechanism with base course reserves. Using the base course reserves only improves outcomes for first-year students, who have greater access to reserved seats. The mechanism underperforms in terms of allocation fairness as well, with $27.2\%$ of students preferring the schedule of a student with a weakly lower priority in each course and with $59.1\%$ of students experiencing priority violations. \textit{These results justify why we consider the RSD mechanism with optimal course reserves as the benchmark mechanism rather than the RSD mechanism with base course reserves.} Using the base course reserves presents an unfair comparison, as the university sets reserves large at the beginning of the process and relaxes them later. This results in an allocation of courses to students, which, to a large extent, happens after the main registration process in practice.  

\begin{table}[t!]
	\begin{center}
		\footnotesize
		
		\vspace{-1mm}
		
		\refstepcounter{tab}
		\label{tab:appendix-envy-other-mechamisms}
		\caption{The percentage of students who experience schedule envy for the alternative mechanisms.}
	
		\vspace{0mm}
		\begin{tabular}{m{0.22\textwidth}>{\centering}m{0.09\textwidth}>{\centering\arraybackslash}m{0.09\textwidth}>{\centering\arraybackslash}m{0.09\textwidth}>{\centering\arraybackslash}m{0.11\textwidth}>{\centering\arraybackslash}m{0.11\textwidth}>{\centering\arraybackslash}m{0.11\textwidth}}
			\hline
			\hline\\[-3mm]
			& 0 courses & 1 course & 2 courses & 3 courses & 4 courses & 5 courses\\[2mm]
			\hline\\[-3mm]
			\makecell[lc]{DA with\\reserves (single)} & 82.5 (1.2) & 14.2 (1.0) & 2.6 (0.3) & 0.6 (0.2) & 0.06 (0.06) & 0.003 (0.01) \\[2mm]
			\hline\\[-3mm]
			\makecell[lc]{DA with\\reserves (multiple)} & 83.8 (1.1) & 15.3 (1.0) & 0.9 (0.2) & 0.02 (0.02) & 0.0 (0.0) & 0.0 (0.0) \\[2mm]
			\hline\\[-3mm]
			\makecell[lc]{RSD with\\base course reserves} & 72.8 (0.5) & 23.5 (0.5) & 3.4 (0.2) & 0.3 (0.1) & 0.01 (0.02) & 0.0 (0.002) \\[2mm]
			\hline\\[-3mm]
			\makecell[lc]{A-CEEI} & 95.2 (0.4) & 4.8 (0.4) & 0.0 (0.0) & 0.0 (0.0) & 0.0 (0.0) & 0.0 (0.0) \\[2mm]
			\hline\\[-3mm]
			\makecell[lc]{A-CEEI with \\kludgy pricing} & 96.5 (0.3) & 3.5 (0.3) & 0.0 (0.0) & 0.0 (0.0) & 0.0 (0.0) & 0.0 (0.0) \\[2mm]
			\hline\\[-3mm]
			\makecell[lc]{Probabilistic Serial\\ with reserves} & 99.0 (0.1) & 1.0 (0.1) & \makebox[\linewidth][c]{0.0 (0.002)} & 0.0 (0.0) & 0.0 (0.0)& 0.0 (0.0) \\[2mm]
			\hline
			\hline
	\end{tabular}

	\end{center}
		
	\vspace{0mm}
	\begin{spacing}{1}
		{\footnotesize \textit{Notes:} This table reports the percentage of students who prefer the schedule of a student with weakly lower priority in each course in each of the alternative mechanisms. The first column includes students who experience no envy. The other columns display students who experience schedule envy bounded by $1, ..., 5$ courses. Results are averages across 100 runs with different random component draws and they are based on \textit{the population of all students}.}		
	\end{spacing}
\end{table}

\vspace{2mm}
{\bf The A-CEEI and A-CEEI with Kludgy Pricing} mechanisms do not strictly enforce course priorities through prices, unlike the PMP mechanism. These relaxed restrictions result in more course seats available to junior students (Years 1 and 2), who overwhelmingly prefer each mechanism to the RSD benchmark (see the bottom section in Figure \ref{fig:preferred-mechanism-appendix-mechanisms}). While the results are reversed for senior students (Years 3 and 4), many more students overall prefer each of these mechanisms ($48.5\% - 55.0\%$) than the number of students who prefer RSD ($19.6\% - 25.7\%$). Both mechanisms reduce the standard deviation of student utility across all years of study ($8.3\%-22.3\%$). This decrease is substantially larger than that observed in Section \ref{sec:simulations} for the PMP mechanism, suggesting that the relaxed priority constraints allow students to receive more equal access to courses. Similarly, over $95\%$ of students experience no envy toward a student of weakly lower priority, with schedule envy bounded by a single course in each mechanism. However, without these constraints, priority violations are common in both mechanisms, with $61.0\%$ of students in the A-CEEI mechanism and $44.5\%$ of students in the A-CEEI with Kludgy Pricing mechanism whose schedule would benefit from a course seat taken by a student of strictly lower priority. Additionally, the average time to calculate the allocation for the A-CEEI and A-CEEI with Kludgy Pricing mechanisms (about 3200 and 3000 seconds, respectively) is almost twice that of the PMP mechanism (about 1700 seconds).

\begin{table}[t!]
	\begin{center}
		\refstepcounter{tab}
		\label{tab:appendix-priority-violations-other-mechamisms}
				
		\vspace{-1mm}
		\footnotesize
		\caption{The percentage of students with a priority violation for the alternative mechanisms.}
		\vspace{1mm}		
		\centering
		\begin{tabular}{C{0.168\textwidth}C{0.174\textwidth}C{0.137\textwidth}C{0.078\textwidth}C{0.135\textwidth}C{0.144\textwidth}}
				\hline
				\hline\\[-2mm]
				\makecell[cc]{DA with\\reserves (single)} & \makecell[cc]{DA with\\reserves (multiple)} & \makecell[cc]{RSD with \\ base reserves} & \makecell[cc]{A-CEEI\\}&\makecell[cc]{A-CEEI with\\ kludgy pricing} & \makecell[cc]{PS w. seniority\\ and reserves}\\[4mm]
				\hline\\[-2mm]
				16.7\% & 17.6\% & 59.1\% & 61.0\% & 44.5\% & 75.1\%\\[2mm]
				\hline
				\hline
			\end{tabular}
	\end{center}
	
	\vspace{0mm}
	\begin{spacing}{1}
		{\footnotesize \textit{Notes:} This table reports the percentage of students whose course schedule can be improved by taking a seat occupied by a student of strictly lower priority for each of the alternative mechanisms. Results are averages across $100$ runs with different random component draws and they are based on \textit{the population of all students}.}
	\end{spacing}		
\end{table}

\vspace{2mm}
{\bf The PS mechanism} is the only mechanism considered here that produces \textit{random assignments.} Random mechanisms are more flexible, allowing the assignment of probability shares of course seats to students rather than full seats. The PS mechanism, as a result, delivers beneficial outcomes to students in each year of study and has a larger percentage of students that strictly prefer the mechanism than strictly prefer the benchmark ($33.7\%$ vs. $24.4\%$). In terms of allocation fairness, the probabilistic serial mechanism reduces the standard deviation of students' utilities for students in each year of study and results in just $1\%$ of students with envy towards a student of weakly lower priority.\footnote{\cite{kojima2009random} and \cite{budish2013designing} prove that the probabilistic serial mechanism results in no envy in multi-unit demand and multi-unit supply settings, respectively. This result does not extend to settings where agents demand multiple units of objects and objects are supplied in multiple units simultaneously.} In doing so, the mechanism results in the largest percentage of students who experience priority violation ($75.1\%$). This exceeds even the A-CEEI mechanism, which does not account for course priorities. This can, in part, be explained by the flexibility of the mechanism. We count a student as experiencing a priority violation if there is a positive probability share of a course that she prefers to some assigned (possibly empty) probability share that is instead assigned to some student at a strictly lower priority level. This matches the definition of ex-ante stability proposed in \cite{kesten2015theory}. We also consider envy in terms of the ex-ante assignment, not an ex-post assignment. Hence, these results are not directly comparable.

\subsection{Department-First Priorities}
\label{subsec:appendix-additional-sumulations-dept-first}

Some U.S. universities consider department-specific priorities as more important than those based on year of study. For example, Dartmouth College explicitly states these priorities on its website.\footnote{See \url{https://www.dartmouth.edu/reg/registration/course_priorities.html}.} Here, we present additional simulation results comparing the performance of all considered mechanisms under a \textit{department-first priority structure}, where the department-specific priority takes precedence over the year-specific priority. We modify Example \ref{example:priority-structures} from Section \ref{sec:simulations} to provide an illustration.

\begin{example}
\label{example:priority-structures-dept-first}
	The course reservation for Real Analysis in Table \ref{tab:course-reservations} results in the following department-first priority structure for the course:
\begin{itemize}[itemsep=0cm]
	\item $r_{s, c} = 8$: 4th-year students in Dept 2
	\item $r_{s, c} = 7$: 3rd-year students in Dept 2
	\item $r_{s, c} = 6$: $\varnothing$
	\item $r_{s, c} = 5$: $\varnothing$
	\item $r_{s, c} = 4$: 4th-year students in all departments except Dept 2
	\item $r_{s, c} = 3$: 3rd-year students in all departments except Dept 2
	\item $r_{s, c} = 2$: 2nd-year students in all departments
	\item $r_{s, c} = 1$: 1st-year students in all departments
\end{itemize}	
\end{example}

\renewcommand{\arraystretch}{1.1} 
\begin{figure}[t!]
	\begin{center}
	
	\vspace{-1mm}
	
	\includegraphics[width=0.73\textwidth]{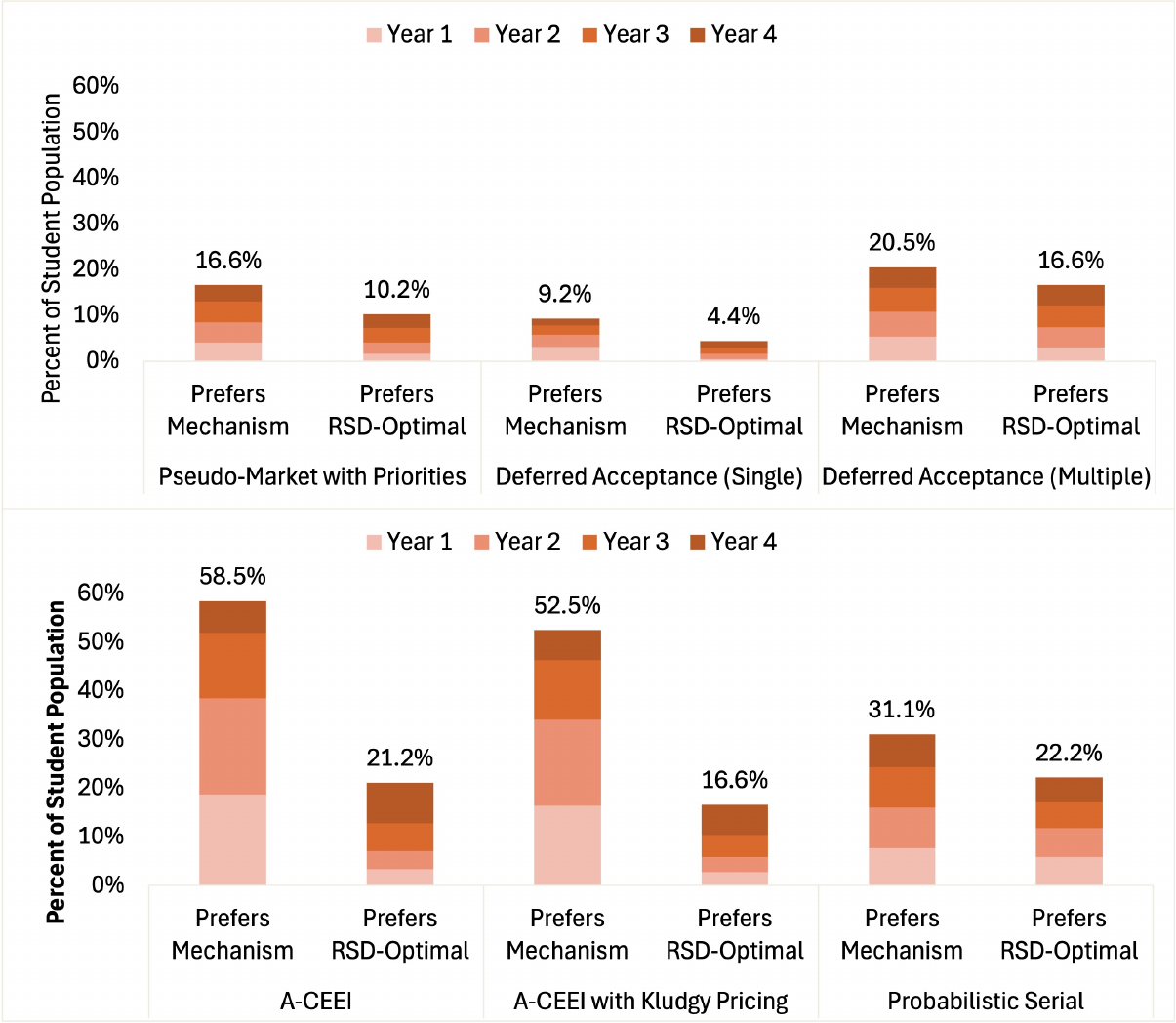}
	\caption{Each student's preferred mechanism under the department-first priority structure.}
	\label{fig:appendix-preferred-mechanism-other-mechamisms-dept-first}
	
	\end{center}
	
	\vspace{-2mm}
	 {\footnotesize \textit{Notes:} This figure reports the percent of students who strictly prefer each given mechanism to the \textit{Random Serial Dictatorship with optimal course reserves} benchmark and the percent of students who strictly prefer the benchmark under the department-first priority structure. Results are averages across 100 runs with different random component draws and they are based on \textit{the population of all students}.}         
	
\end{figure}

Here, we do not analyze Random Serial Dictatorship with base course reserves, as its performance is subpar compared to all other considered mechanisms. We also drop Deferred Acceptance with reserves under single and multiple tie-breaking, as these mechanisms do not depend on the priority structure given and perform very similarly to the Random Serial Dictatorship with optimal course reserves and the corresponding Deferred Acceptance mechanisms without reserves. The results of our simulations are presented in Figures \ref{fig:appendix-preferred-mechanism-other-mechamisms-dept-first} -- \ref{tab:appendix-st-dev-other-mechanisms-dept-first} and Tables \ref{tab:appendix-envy-other-mechamisms-dept-first} -- \ref{tab:appendix-priority-violations-other-mechamisms-dept-first}.

\begin{figure}[t!]
	\begin{center}
	
	\vspace{-1mm}
	\includegraphics[width=0.73\textwidth]{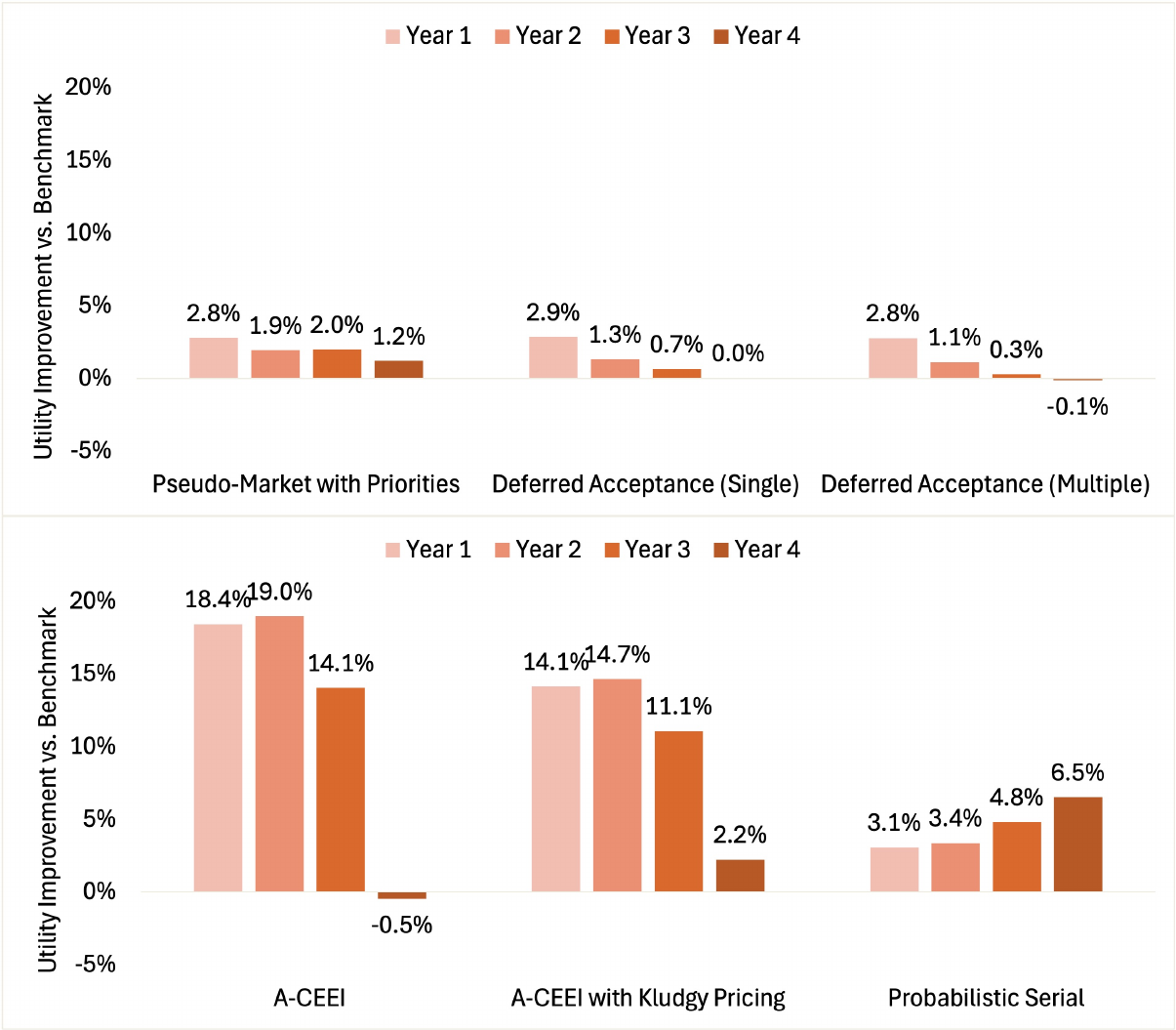}
	\caption{Improvements in utility under the department-first priority structure.}
	\label{fig:utility-all-changing-students-department-first}
	
	\end{center}
	
	\vspace{-2mm}
	 {\footnotesize \textit{Notes:} This figure reports the percent improvement in mean utility for each given mechanism over the \textit{Random Serial Dictatorship with optimal course reserves} benchmark under the department-first priority structure. Results are averages across $100$ runs with different random component draws and they are based on \textit{the population of all students}.}       
\end{figure}

Figure \ref{fig:appendix-preferred-mechanism-other-mechamisms-dept-first} displays the results on each student's preferred mechanism under the \textit{department-first priority structure}. As in Figure \ref{fig:preferred-mechanism-appendix-mechanisms}, the results are reported for the full population of students, so the two figures can be compared directly for every mechanism that appears in both (all mechanisms except the Pseudo-Market with Priorities and the DA mechanisms). Under the department-first structure, a larger share of students strictly prefers each mechanism to the RSD-Optimal benchmark than prefers the benchmark itself.

As opposed to the \textit{hybrid priority structure}, the department-first priority structure has all students eligible for reserved seats prioritized above more senior but ineligible students. For PMP, DA-STB, and DA-MTB, this preferential treatment benefits more junior students (Years $1$ and $2$) and harms senior students (Years $3$ and $4$). This is illustrated in Figure \ref{fig:utility-all-changing-students-department-first}, where the improvement in utility generally decreases in years of study, whereas the improvement in utility is generally increasing in years of study for Figure \ref{fig:utility-all-changing-students}. For the A-CEEI and A-CEEI with kludgy pricing mechanisms, the improvement in utility decreases in years of study under both priority structures (bottom parts of Figures \ref{fig:utility-all-changing-students-appendix-mechanisms} and \ref{fig:utility-all-changing-students-department-first}). Under the department-first priority structure, however, their advantage for first-year students is smaller and their results for third- and fourth-year students are better, mainly because the benchmark itself favors junior students who are eligible for reserved seats.

Figure \ref{tab:appendix-st-dev-other-mechanisms-dept-first} reports each mechanism's change in the standard deviation of students' utilities relative to the RSD benchmark. The largest improvements happen for the A-CEEI, A-CEEI with kludgy pricing, and Probabilistic Serial with seniority and reserves mechanisms. Table \ref{tab:appendix-priority-violations-other-mechamisms-dept-first} shows that this happens at the expense of priority violations, which at least $48\%$ of students experience in each of these three mechanisms. As in the hybrid priority structure results, this is a larger percentage of priority violations than we observe in the benchmark ($5.3\%$). In turn, the PMP, DA-STB, and DA-MTB mechanisms have no priority violations, but perform closer to the RSD benchmark in terms of the standard deviation of students' utilities.

\begin{figure}[t!]
	\begin{center}
	
	\vspace{-1mm}
	\includegraphics[width=0.73\textwidth]{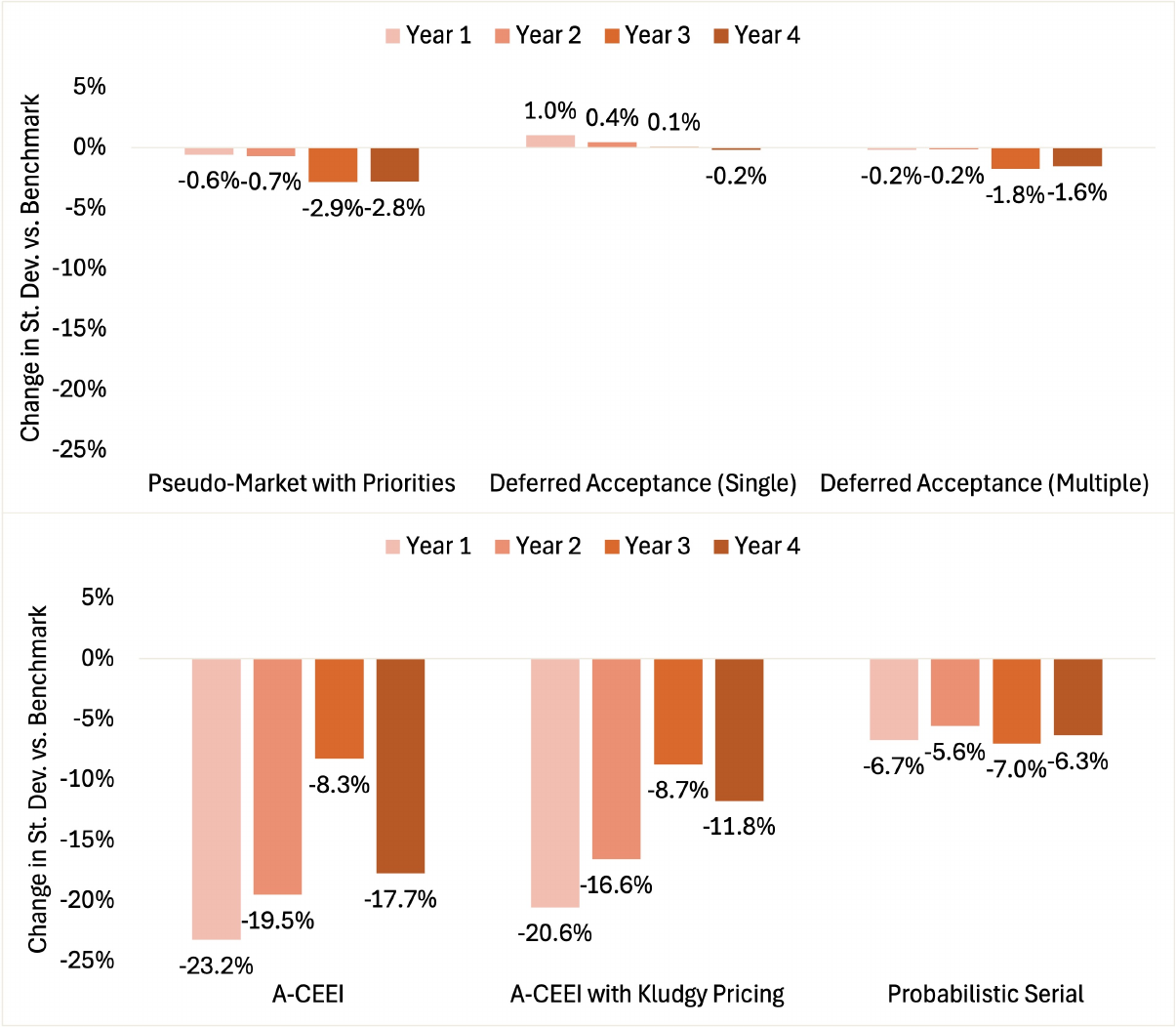}
	\caption{The change in the standard deviation of students' utilities under the department-first priority structure.}
	\label{tab:appendix-st-dev-other-mechanisms-dept-first}
	
	\end{center}
	
	\vspace{-2mm}
    {\footnotesize \textit{Notes:} This figure reports the percent change in the standard deviation of students' utilities compared to the \textit{Random Serial Dictatorship with optimal course reserves} benchmark under the department-first priority structure. Results are averages across 100 runs with different random component draws and they are based on \textit{the population of all students}.}
	
\end{figure}

Table \ref{tab:appendix-envy-other-mechamisms-dept-first} shows that the percentage of students experiencing schedule envy in each mechanism is similar in comparison to the hybrid priority structure (see Tables \ref{tab:student-envy} and \ref{tab:appendix-envy-other-mechamisms}). The PMP, A-CEEI, and A-CEEI with kludgy pricing mechanisms deliver schedule envy bounded by a single course toward students of weakly lower priority, with only $8.6\%$, $4.8\%$, and $3.9\%$ experiencing any envy, respectively. The DA-STB and DA-MTB mechanisms each result in about $13\%$ of students experiencing schedule envy, with about $2\%$ experiencing envy for several courses in the DA-STB mechanism and under $1\%$ in DA-MTB. The RSD with optimal course reserves mechanism has $14.9\%$ of students experiencing schedule envy, with about $2\%$ with envy for several courses. Only about $1\%$ of students experience schedule envy in the random assignment of the Probabilistic Serial with seniority and reserves mechanism.

\vspace{3mm}
\renewcommand{\arraystretch}{1} 
\begin{table}[h]
	\begin{center}
		\footnotesize
		\refstepcounter{tab}
		\label{tab:appendix-envy-other-mechamisms-dept-first}
		\caption{The percentage of students who experience schedule envy under the department-first priority structure.}
		
		\vspace{0mm}
		\begin{tabular}{m{0.22\textwidth}>{\centering}m{0.09\textwidth}>{\centering\arraybackslash}m{0.09\textwidth}>{\centering\arraybackslash}m{0.09\textwidth}>{\centering\arraybackslash}m{0.11\textwidth}>{\centering\arraybackslash}m{0.11\textwidth}>{\centering\arraybackslash}m{0.11\textwidth}}
			\hline
			\hline\\[-2mm]
			& 0 courses & 1 course & 2 courses & 3 courses & 4 courses & 5 courses\\[2mm]
			\hline\\[-3mm]
			\makecell[lc]{Pseudo-Market\\with Priorities} & 91.4 (0.7) & 8.6 (0.7) & 0.0 (0.0) & 0.0 (0.0) & 0.0 (0.0) & 0.0 (0.0) \\[2mm]
			\hline\\[-3mm]
			\makecell[lc]{RSD with\\optimal course reserves} & 85.1 (0.6) & 12.6 (0.6) & 1.8 (0.2) & 0.4 (0.1) & 0.03 (0.03) & \makebox[\linewidth][c]{0.001 (0.006)}\\[2mm]
			\hline\\[-3mm]
			\makecell[lc]{DA with \\ single tie-breaking} & 86.6 (0.6) & 11.5 (0.5) & 1.6 (0.2) & 0.3 (0.2) & 0.02 (0.03) & \makebox[\linewidth][c]{0.001 (0.008)} \\[2mm]
			\hline\\[-3mm]
			\makecell[lc]{DA with \\ multiple tie-breaking} & 87.5 (0.6) & 12.1 (0.6) & 0.4 (0.2) & \makebox[\linewidth][c]{0.002 (0.008)} & 0.0 (0.0) & 0.0 (0.0) \\[2mm]
			\hline\\[-3mm]
			A-CEEI & 95.2 (0.4) & 4.8 (0.4) & 0.0 (0.0) & 0.0 (0.0) & 0.0 (0.0) & 0.0 (0.0) \\[2mm]
			\hline\\[-3mm]
			\makecell[lc]{A-CEEI with\\kludgy pricing} & 96.1 (0.3) & 3.9 (0.3) & 0.0 (0.0) & 0.0 (0.0) & 0.0 (0.0) & 0.0 (0.0) \\[2mm]
			\hline\\[-3mm]
			\makecell[lc]{Probabilistic Serial with \\ seniority and reserves} & 99.2 (0.1) & 0.8 (0.1) & 0.0 (0.0) & 0.0 (0.0) & 0.0 (0.0) & 0.0 (0.0) \\[2mm]
			\hline
			\hline
		\end{tabular}
	\end{center}
	
	\vspace{0mm}
	\begin{spacing}{1}
		{\footnotesize \textit{Notes:} This table reports the percentage of students who prefer the schedule of a student with weakly lower priority in each course in each mechanism under the department-first priority structure. The first column includes students who experience no envy. The other columns display students who experience schedule envy bounded by $1,...,5$ courses. Results are averages across $100$ runs with different random component draws and they are based on \textit{the population of all students}.}		
	\end{spacing}
\end{table}

\vspace{0cm}
\renewcommand{\arraystretch}{1} 
\begin{table}[b!]
	\begin{center}
		\refstepcounter{tab}
		\label{tab:appendix-priority-violations-other-mechamisms-dept-first}
				
		\footnotesize
		\caption{The percentage of students with a priority violation under the department-first priority structure.}

		\vspace{1mm}		
		\centering
		\begin{tabular}{C{0.05\textwidth}C{0.119\textwidth}C{0.14\textwidth}C{0.15\textwidth}C{0.08\textwidth}C{0.134\textwidth}C{0.144\textwidth}}
				\hline
				\hline\\[-2mm]
				\makecell[cc]{PMP\\} & \makecell[cc]{RSD with\\ opt. reserves} &
				\makecell[cc]{DA with single\\tie-breaking} & \makecell[cc]{DA with mult.\\ tie-breaking} &  \makecell[cc]{A-CEEI}&\makecell[cc]{A-CEEI with\\ kludgy pricing} & \makecell[cc]{PS w. seniority\\ and reserves}\\[4mm]
				\hline\\[-2mm]
				0\% & 5.3\% & 0\% & 0\% & 61.1\% & 48.0\% & 79.1\%\\[2mm]
				\hline
				\hline				
			\end{tabular}
	\end{center}
		
	\vspace{0mm}
	\begin{spacing}{1}
		{\footnotesize \textit{Notes:} This table reports the percentage of students whose course schedule can be improved by taking a seat occupied by a student of strictly lower priority for each mechanism under the department-first priority structure. Results are averages across $100$ runs with different random component draws and they are based on \textit{the population of all students}.}
	\end{spacing}		
\end{table}

\pagebreak
\subsection{Robustness: Student Choice Sets and Noise Term}
\label{subsec:appendix-additional-simulations-comparative-statics}

This section performs some additional robustness checks. We take the calibrated student utility model and consider changes in two parameters: the size of student choice sets ($60$, $70$, $80$, $90$) and the standard deviation of the random component of student utility ($\sigma=0.75,1,1.25,1.5$). Figure \ref{fig:comparative-statics-choice-set-sigma} reports the performance of the PMP mechanism compared to the RSD benchmark. The graphs show the percentages of students who are indifferent (solid line), prefer the PMP mechanism (dashed line), and prefer the RSD benchmark (dotted line). Reported simulation results are based on $50$ random component draws.

The top four graphs show comparative statics results for the size of student choice sets. The percentage of the total population of students indifferent between the PMP and RSD mechanisms is around $70\%$ across all sizes and slightly lower for student choice sets of size $90$. A relatively equal percentage of Year $1$ students strictly prefer each mechanism. For Year $2$ through Year $4$ students, a larger percentage of students prefer the PMP mechanism regardless of the size of the choice set. 

The bottom four graphs show comparative statics results for $\sigma$. The percentage of students indifferent between the PMP mechanism and RSD mechanism significantly declines from $\sigma = 0.75$ to $\sigma = 1.5$, with the largest decrease for Year 1 students ($84.2\%$ to $13.7\%$). Again, a relatively equal percentage of Year $1$ students strictly prefer each mechanism. For Year $2$ through Year $4$ students, a larger percentage of students prefer the PMP mechanism regardless of the size of the standard deviation. The benefits of the PMP mechanism are larger for noisier environments.

Overall, the PMP mechanism outperforms the RSD benchmark for all considered values of the size of student choice sets and the standard deviation of the random component of student utility. Our results are robust in response to the changes in these parameters.

\newpage
\begin{figure}[h!]
	
	\caption{Additional Simulations: Student Choice Sets and Noise Term.}

	\vspace{-4mm}
	\begin{center}
		\includegraphics[height=0.38\textheight]{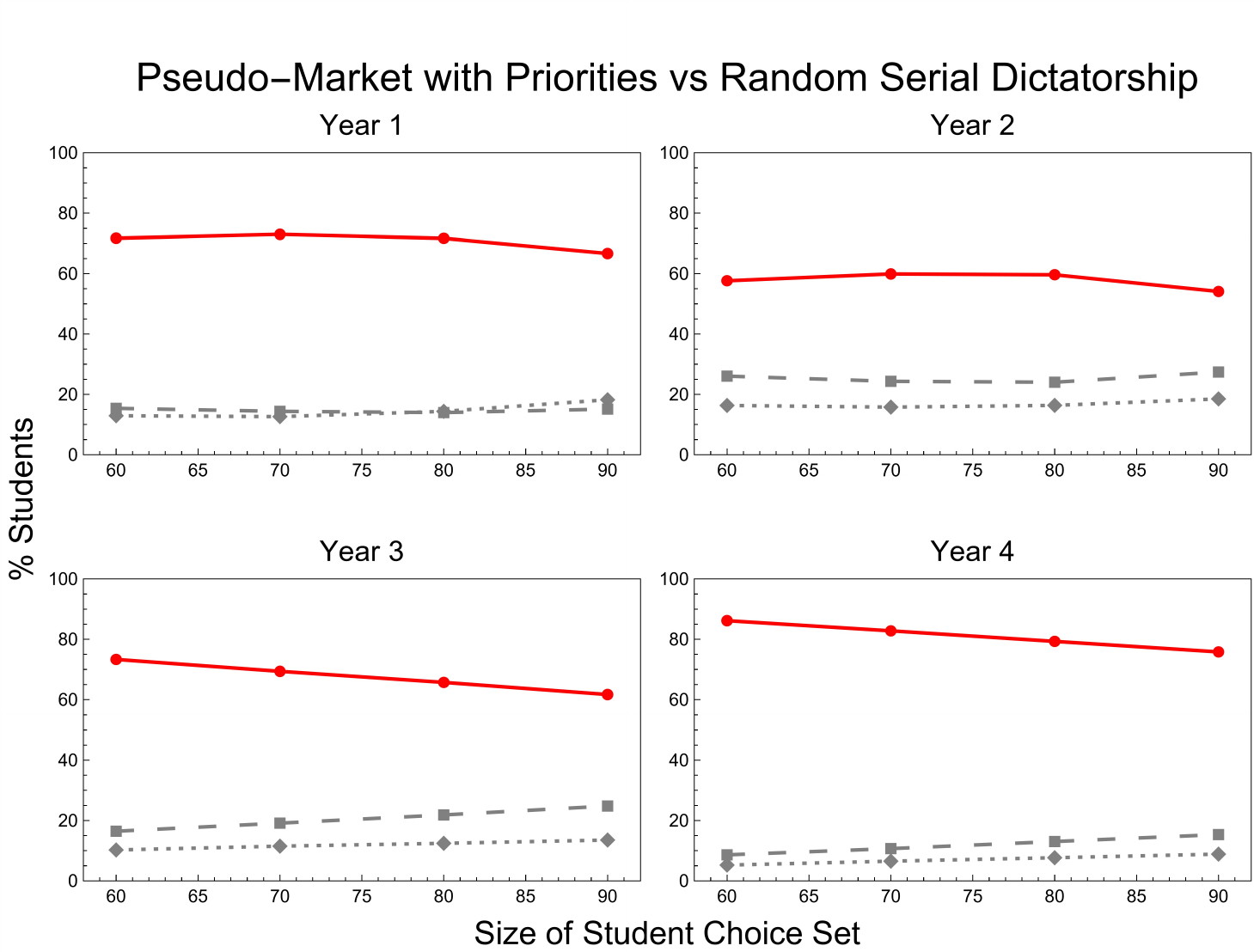}\\[4mm]
		\includegraphics[height=0.38\textheight]{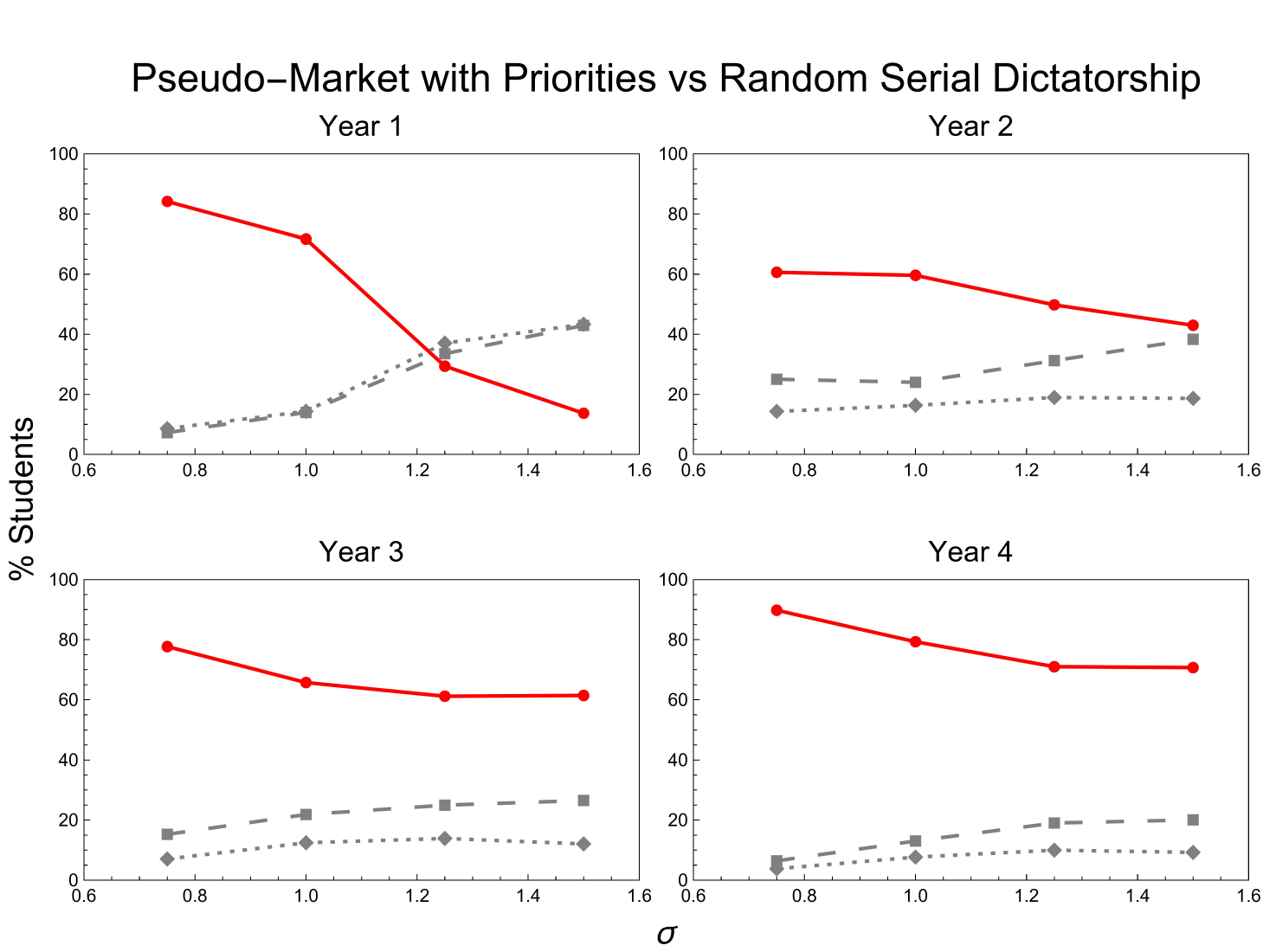}
	\end{center}
	
	\vspace{-4mm}
	\begin{spacing}{1}
		{\footnotesize \textit{Notes:} The performance of Pseudo-Market with Priorities (PMP) compared to  the Random Serial Dictatorship (RSD) with optimal course reserves. The top four figures present the percentages of students who are Indifferent (solid line), Prefer (dashed line), and Do Not Prefer (dotted line) the allocation of PMP to the RSD benchmark for different sizes of student choice sets: $60$, $70$, $80$, and $90$. The bottom four figures present the same comparison for the standard deviation of student utility $\sigma=0.75,1,1.25,1.5$. Results are averages across $50$ runs with different random component draws and they are based on \textit{the population of all students}.}
	\end{spacing}
	
	\label{fig:comparative-statics-choice-set-sigma}
\end{figure}

\renewcommand{\theequation}{E.\arabic{equation}} %
\setcounter{equation}{0}

\renewcommand{\thelemma}{E\arabic{lemma}}
\setcounter{lemma}{0}

\renewcommand{\thedefinition}{E\arabic{definition}}
\setcounter{definition}{0}

\renewcommand{\theproposition}{E\arabic{proposition}} %
\setcounter{proposition}{0}

\renewcommand{\thetheorem}{E\arabic{theorem}} %
\setcounter{theorem}{0}

\renewcommand{\theexample}{E\arabic{example}} %
\setcounter{example}{0}

\renewcommand{\thecorollary}{E\arabic{corollary}} %
\setcounter{corollary}{0}

\newpage
\section{Supplementary Materials}
\label{sec:supplementary-materials}

The supplementary materials are organized as follows. Section \ref{subsec:program} describes the Mathematica programs for calibrating the student utility model in Section \ref{sec:appendix-student-utility-estimation}
and performing the simulations comparing course allocation mechanisms in Section \ref{subsec:simulation-results} and Sections \ref{subsec:appendix-pmp-prices}--\ref{subsec:appendix-additional-simulations-comparative-statics}. Section \ref{subsec:algorithm} provides details about the algorithm developed to quickly find the Pseudo-Market Equilibrium with Priorities for large student populations. 

\subsection{Program}
\label{subsec:program}

Most simulations have been done using Mathematica 14.1.0 on a server with 36 Cores of CPU (2.3 GHz of an Intel Xeon E5-2697 v4 CPU), 256 GB of RAM, 128 GB of SSD hard drive space, and a Windows Server 2019 Standard operating system. The replication package runs the following programs in their listed order:

\begin{enumerate}
	\item \textbf{MSM-utility-model.nb} calibrates the student utility model (Section \ref{sec:appendix-student-utility-estimation}).
	\item \textbf{simulations-year-first.nb} simulates the four mechanisms of Section \ref{subsec:mechanisms} under the year-first priority structure (Sections \ref{subsec:simulation-results} and \ref{subsec:appendix-pmp-prices}).
	\item \textbf{simulations-dept-first.nb} simulates the same mechanisms under the department-first priority structure (Section \ref{subsec:appendix-additional-sumulations-dept-first}).
	\item \textbf{simulations-appendix-year-first.nb} simulates the six alternative mechanisms under the year-first priority structure (Section \ref{subsec:appendix-additional-sumulations-other-mechanisms}).
	\item \textbf{simulations-appendix-dept-first.nb} simulates the alternative mechanisms under the department-first priority structure (Section \ref{subsec:appendix-additional-sumulations-dept-first}).
	\item \textbf{comparison-RSD-Base.nb} compares the PMP, DA-STB, and DA-MTB mechanisms with RSD-Base (Section \ref{subsec:simulation-results}).
	\item \textbf{simulations-comparative-statics-rankcount.nb} varies the size of student choice sets (Section \ref{subsec:appendix-additional-simulations-comparative-statics}).
	\item \textbf{simulations-comparative-statics-sigma.nb} varies the standard deviation $\sigma$ of the random utility component (Section \ref{subsec:appendix-additional-simulations-comparative-statics}).
\end{enumerate}

\noindent Program 1 saves the calibrated parameters, which all other programs read. Programs 2 and 3 compute the optimal course reserves for their priority structure and save them; the later programs load the corresponding reserves. 

\bigskip
\noindent The program \textbf{MSM-utility-model.nb} has the following sections:
\begin{itemize}
	\item {\bf Functions, Data, and Parameters} sets the calibration parameters, loads in data, and defines three functions.
	\begin{itemize}
		\item {\bf ShortDataCreation} creates a smaller data set for testing purposes.
		\item {\bf Utility} returns the matrix of student utilities over courses for given parameters $\theta$, $\Gamma$, and $z$, choice sets, and random utility components.
		\item {\bf RSDallocation} finds the outcome of the Random Serial Dictatorship with course reserves for given student utilities, tie-breaking order, and course reserves.
	\end{itemize}

	\item {\bf Simulated Method of Moments} sets the main parameters of the method of simulation moments and defines all essential functions of the calibration process.
	\begin{itemize}
		\item {\bf EmpiricalMoments} computes the moments from the data.
		\item {\bf RandomParam} draws student choice sets and random utility components for a given random seed.
		\item {\bf SimulatedMoments} computes the RSD allocations for a number of draws and averages the moments that {\bf MomentSelection} extracts from them. {\bf SimulatedMomentsSeq} and {\bf SimulatedMomentsWeightCalibration} are variants used in the Jacobian computation and in the estimation of the variance-covariance matrix $\hat{S}$, respectively.
		\item {\bf W1} calculates the weight matrices of the first stage and {\bf Weight} calculates the variance-covariance matrix $\hat{S}$ of the second stage. {\bf RidgeShat} adds a small ridge to the diagonal of $\hat{S}$ if it is singular.
		\item {\bf MSMObjFn} returns the value of the objective function (\ref{eq:objective}) and its gradient.
		\item {\bf JacobianEstimate} estimates the Jacobian of the simulated moments by finite differences; {\bf JacobianDiagnostics} checks it and repairs its zero rows and columns, if necessary.
		\item {\bf Optimization} runs the gradient descent from given initial parameters. {\bf ParameterNormalization} normalizes the parameters as discussed in Section \ref{subsec:appendix-Model}.
		\item {\bf JTest} computes the $J$-test statistic of Section \ref{subsec:appendix-inference}.
	\end{itemize}	
	\item {\bf Run} performs the calibration in two iterations: a first run with $12$ draws per evaluation of the simulated moments, and a main run with $100$ draws. Each iteration consists of the two stages of Section \ref{subsec:appendix-calibration}. The first stage minimizes the objective function with the weight matrix $W_0$. The second stage updates the optimal weight matrix and minimizes the objective function until the largest parameter change falls below $0.01$ or a limit on the number of rounds is reached. In the second stage of the main run, the gradient uses the estimated Jacobian, updated every second round. The section performs a $J$-test and saves the calibrated parameters in file \textit{UTILITY-MODEL-PARAMETERS.xlsx}.
	
\end{itemize}
\bigskip
\noindent The program \textbf{simulations-year-first.nb} has the following sections:
\begin{itemize}
	\item {\bf Functions, Data, and Parameters} sets the simulation parameters, loads in data and the calibrated parameters, and defines the following functions:
	\begin{itemize}
		\item {\bf ShortDataCreation}, {\bf RandomParam}, {\bf Utility}, and {\bf RSDallocation} are the same as in \textbf{MSM-utility-model.nb}.
		\item {\bf DAmechanism} finds the outcome of the Deferred Acceptance mechanism for given student utilities and priority order.
		\item {\bf OptSetAsideFunction} computes the optimal course reserves of Section \ref{subsec:mechanisms} and saves them, or loads them from a file, depending on a flag.
	\end{itemize}

	\item {\bf Allocation Generation Functions} defines functions that find mechanism outcomes:
	\begin{itemize}
		\item {\bf PMPallocation} finds the outcome of the Pseudo-Market with Priorities mechanism with the algorithm described in Section \ref{subsec:algorithm}; {\bf PMPInitialPrice}, {\bf PMPPhaseI}, {\bf PMPPhaseII}, and {\bf FinalAllocation} implement its stages, and {\bf Error} calculates the market-clearing error.
		\item {\bf OptAllocation} finds all students' demands at given prices, solving each student's problem with {\bf OptAllocationInd}. {\bf OptAllocationFast} updates only the demands that may have changed since the previous step. {\bf NMaximizeDirect} and {\bf NMinimizeDirect} solve the student's problem when Mathematica's numerical solvers fail.
		\item {\bf OptAllocationFinalFeasibilityCheck} checks that each student's demand at the final prices is feasible, and {\bf OptAllocationFinalOptimizationCheck} checks that it is optimal; both correct any errors.
		\item {\bf DASTBallocation} and {\bf DAMTBallocation} find the outcomes of the Deferred Acceptance mechanism with single and multiple tie-breakings.
	\end{itemize}

	\item {\bf Simulation Run} is the main section. For each of the $100$ runs, computed in parallel, it sets a random seed, draws a tie-breaking order and student utilities, and finds the PMP, RSD-Optimal, DA-STB, and DA-MTB allocations. It records the allocations and the PMP prices and market-clearing errors in text files.

	\item {\bf Results Setup} loads the allocations and regenerates the students' utilities.

	\item The remaining sections compute the results, either for students with at least two different schedules across the mechanisms (``limited universe'') or all students (``full universe''):
	\begin{itemize}
		\item {\bf Preferred Mechanism Results (limited)} and {\bf Utility Results (limited)} produce Figures \ref{fig:preferred-mechanism} and \ref{fig:utility-all-changing-students}. {\bf Utility Results (full)} produces Figure \ref{fig:st-dev-PMP-DA-DAm}. {\bf Preferred Mechanism Results (full)} repeats the comparison of Figure \ref{fig:preferred-mechanism} for all students.
		\item {\bf Envy Results} produces Table \ref{tab:student-envy}, and {\bf Priority Violation Results} computes the priority violations discussed in Section \ref{subsec:simulation-results}.
		\item {\bf Market-Clearing Error and Price Results} produces Tables \ref{tab:pmp-prices} and \ref{tab:market-clearing}.
		\item {\bf Total Enrollment Results} computes seats assigned (Section \ref{subsec:simulation-results}).
	\end{itemize}
\end{itemize}

\noindent The program \textbf{simulations-dept-first.nb} differs in the priority structure, which is set to ``Department-First'', and uses the same functions. It computes the optimal course reserves for this structure, writes its output to separate files, and omits the limited-universe sections.

\bigskip
\noindent The program \textbf{simulations-appendix-year-first.nb} has the same sections as \textbf{simulations-year-first.nb}, with three differences:
\begin{itemize}
	\item It loads the optimal course reserves from \textbf{simulations-year-first.nb} and uses the same random seeds, so both programs simulate the same student utilities and tie-breakings.
	\item In {\bf Allocation Generation Functions}, the PMP functions are generalized to the A-CEEI and A-CEEI with kludgy pricing mechanisms (these versions carry the suffix APPENDIX, e.g., {\bf PMPallocationAPPENDIX}); {\bf ACEEIallocation} and {\bf Kludgyallocation} call them for the two mechanisms. {\bf DAreservesmechanism} finds the outcome of the Deferred Acceptance mechanism with reserves for a given priority order, and {\bf DASTBreservesallocation} and {\bf DAMTBreservesallocation} apply it with single and multiple tie-breakings. {\bf PSseniorityreservesallocation} finds the outcome of the Probabilistic Serial mechanism with seniority and reserves.
	\item {\bf Simulation Run} finds the allocations of the RSD-Optimal benchmark and the six alternative mechanisms; RSD-Base is found by {\bf RSDallocation} with the university's base course reserves.
\end{itemize}
The file produces Figures \ref{fig:preferred-mechanism-appendix-mechanisms}--\ref{fig:st-dev-appendix-mechanisms} and Tables \ref{tab:appendix-envy-other-mechamisms} and \ref{tab:appendix-priority-violations-other-mechamisms}.

\vspace{2mm}
\noindent The program \textbf{simulations-appendix-dept-first.nb} differs in the priority structure and loads the optimal course reserves saved by \textbf{simulations-dept-first.nb}. The results of the two department-first programs together produce Figures \ref{fig:appendix-preferred-mechanism-other-mechamisms-dept-first}--\ref{tab:appendix-st-dev-other-mechanisms-dept-first} and Tables \ref{tab:appendix-envy-other-mechamisms-dept-first} and \ref{tab:appendix-priority-violations-other-mechamisms-dept-first}.

\vspace{2mm}
\noindent The program \textbf{comparison-RSD-Base.nb} runs no new simulations. It replaces the RSD-Optimal allocations recorded by \textbf{simulations-year-first.nb} with the RSD-Base allocations and repeats the analysis of Figures \ref{fig:preferred-mechanism} and \ref{fig:utility-all-changing-students} with RSD-Base as the benchmark. 

\vspace{2mm}
\noindent The two comparative statics programs repeat the simulations of \textbf{simulations-year-first.nb} with $50$ runs for each of four parameter values. Each compares the PMP mechanism with the RSD benchmark for all students and draws one half of Figure \ref{fig:comparative-statics-choice-set-sigma}.

\subsection{Computational Algorithm}
\label{subsec:algorithm}

We developed an algorithm to find a Pseudo-Market Equilibrium with Priorities, taking in student utilities $v$, student priorities $priorities$, and a tie-breaking order over students $studentrank$, and outputting an allocation $bestdemand$ and a vector of course prices $bestt$. Its main stages are described below, with the pseudocode shown on p. \pageref{algorithm}.

\vspace{2mm}
\noindent \textit{Initialization}. We first reorder the set of student budgets $b$, evenly spread values from $1$ to $1+\beta = 1.25$, and assign them to students in order of $studentrank$. The function \textbf{PMPInitialPrice} takes in student utilities $v$ and provides an initial estimate for the price vector based on students' most-preferred course schedules \textit{initdemand}. Rather than working in $M\!R$-dimensional space, we leverage the approach in the proof of Theorem \ref{theorem:existence} and parameterize priority-specific prices as $t\in [0, R\overline{b}]^M$ . We collect these parameters in $pmpparam\gets \{v,b,priorities,initdemand\}$ and initialize $boundcheck\gets 0$ and $oversubcheck\gets 0$, which serve as exit flags for Phases I and II of the algorithm.

\vspace{2mm}
\noindent \textit{Phase I}. The first phase of the algorithm is a t$\hat{\text{a}}$tonnement process, searching for a price vector by repeatedly adjusting prices proportionally to excess demand. If the market-clearing error remains above the theoretical bound THEORYBOUND $\approx 43.5$, this phase ends when there is no improvement in the error over a fixed number of consecutive steps. If the market-clearing error is below the theoretical bound, we stop once there is no substantial improvement (more than $1\%$) in error over the same number of consecutive steps. The algorithm sets the flag $boundcheck=1$ if the error falls below THEORYBOUND and $boundcheck=0$ otherwise. As output, Phase I returns the price vector $bestt$ and the exit flag $boundcheck$. The upper limit on the consecutive steps with no improvement and some other parameters of Phase I depend on $nround$ making the search process more persistent for later rounds.

\vspace{2mm}
\noindent The main innovation in Phase I is that demands are only updated for two subsets of students in each step. The first includes students who can no longer afford the allocation determined in the previous step.  The second includes students who did not demand their most-preferred course schedule in the previous step and have a positive utility for courses with recently decreased prices (within the student's budget). Any student outside these two subsets has her demand unchanged between steps. This has proven particularly effective, significantly shortening the search process compared to updating every student's demand.

\vspace{2mm}
\noindent \textit{Phase II}. The second phase starts with the best price of Phase I and reduces excess course demand to levels below those observed in the data. For each positive excess demand course and each demanding student, we use a dual program to calculate the minimum price change needed so the student no longer demands the course and apply the smallest change that fully eliminates the course's excess demand. We repeat this process until two conditions are met: (i) the distribution of excess demand errors is stochastically dominated by the distribution observed in the data, and (ii) no course is oversubscribed by more than $k-1=4$ seats. If both (i) and (ii) are satisfied we set $oversubcheck = 1$. We also update the value of $boundcheck$ as the market-clearing error changes. The output of Phase II includes price vector $bestt$ and the exit flags $boundcheck$ and $oversubcheck$. 

\vspace{2mm}
\noindent \textit{Final Allocation}. After Phase II, we check each student's demand is correctly computed at the final prices and produce the final course allocation. If errors occur, the FinalAllocation function will correct them. If errors violate the Phase I and II conditions, we set $bestdemand = \emptyset$ and restart the process with a new initial price vector. 

\vspace{2mm}
The outer while loop of the algorithm continues as long as $boundcheck=0$ or\\ $oversubcheck=0$. At each restart, we add a small amount of noise to the initial price vector to begin the search from a slightly different starting point. The inner loop has the same exit conditions, with an additional condition on the number of rounds. If either $boundcheck$ or $oversubcheck=0$ after Phase II, we repeat the process. We allow up to $maxrounds=6$ iterations of Phase I and II within the inner loop making the search process more persistent during additional rounds. If the algorithm fails to find a solution after $maxrounds$ iterations, we restart with a new initial price vector. Typically only 1-4 iterations occur between Phases I and II in the inner loop. We did not encounter cases requiring a restart through the outer loop for the main simulations. The simulations in the appendix Section \ref{subsec:appendix-additional-simulations-comparative-statics} checking the robustness of our results do encounter some challenging cases where a restart is necessary.

\vspace{5mm}
\begin{spacing}{1}
	\begin{algorithm}[H]
		\label{algorithm}
		\renewcommand{\thealgorithm}{}
		\caption{Pseudo-Market with Priorities}
		\begin{algorithmic}[1]
			\INPUT student utilities $v$, student priorities $priorities$, student ranking $studentrank$ 
			\OUTPUT course allocation $bestdemand$, prices $bestt$
			\STATE \textbf{Initialization:} $b\gets b(studentrank)$, $\{t_0,initdemand\}\gets\textbf{PMPInitialPrice}(v)$, 
			\STATE $pmpparam\gets \{v,b,priorities,initdemand\}$, $boundcheck\gets 0$, $oversubcheck\gets 0$ 
			\WHILE{$boundcheck=0$ OR $oversubcheck=0$}
			\STATE $t\gets t_0+noise$, $nround \gets 1$
			\WHILE{($boundcheck=0$ OR $oversubcheck=0$) AND $nround\leq maxrounds$}
			\STATE $\{bestt,boundcheck\}\gets\textbf{PMPPhaseI}(t,pmpparam,nround)$
			\STATE $\{bestt,boundcheck,oversubcheck\}\gets\textbf{PMPPhaseII}(bestt,pmpparam)$
			\STATE $\{bestdemand\} \gets \textbf{FinalAllocation}(bestt, pmpparam)$
			\IF {$bestdemand==\emptyset$} 
			\STATE $boundcheck\gets 0, oversubcheck\gets 0, nround\gets maxrounds$
			\ENDIF
			
			\STATE $t\gets bestt$, $nround\gets nround+1$
			\ENDWHILE 
			\ENDWHILE				
		\end{algorithmic}
	\end{algorithm}
\end{spacing}

\end{document}